\documentclass[12pt]{article}

\usepackage[utf8]{inputenc}
\usepackage[margin=0.8in]{geometry}
\usepackage{multirow,multicol}
\usepackage{amsfonts}
\usepackage{amsmath,amssymb,amsthm}
\usepackage[mathscr]{euscript}
\usepackage{comment}
\usepackage{graphicx,color}
\usepackage{mathtools}
\usepackage{mathabx} %widecheck
\usepackage[linesnumbered,ruled,vlined]{algorithm2e}
\usepackage{setspace}
\usepackage{sectsty}
\usepackage{makecell}
\usepackage{chngcntr}
\usepackage{apptools}
\usepackage{booktabs}
\usepackage{caption}
\usepackage{lscape}
\usepackage{tabularx}
\usepackage{underscore} %underline
\usepackage{array} 
\usepackage{arydshln}
\usepackage{soul}
\usepackage{tikz}
\usetikzlibrary{matrix}
\usepackage{subfig}
\usepackage[english]{babel}
\usepackage[toc,page]{appendix}

\let\counterwithin\relax
\usepackage{lipsum}
\usepackage{natbib}
\usepackage{hyperref}
\usepackage{xr}
\AtAppendix{\counterwithin{lemma}{section}}
\newtheorem{assumption}{Assumption}

\newtheorem{lemma}{Lemma}
\newtheorem{proposition}{Proposition}
\newtheorem{theorem}{Theorem}
\newtheorem{corollary}{Corollary}

\theoremstyle{remark}
\newtheorem{remark}{Remark}

\newcolumntype{P}[1]{>{\centering\arraybackslash}p{#1}}

\newcommand{\bm}{\boldsymbol}
\newcommand{\cm}[1]{\mbox{\boldmath$\mathscr{#1}$}}
\newcommand{\cmt}[1]{\mbox{\boldmath\tiny$\mathscr{#1}$}}
\newcommand{\Fr}{{\mathrm{F}}}
\newcommand{\op}{{\mathrm{op}}}
\def\HH{{\mathrm{\scriptscriptstyle\mathsf{H}}}}

\newcommand{\est}{{\mathrm{est}}}
\newcommand{\pred}{{\mathrm{pred}}}
\newcommand{\trunc}{{\mathrm{trunc}}}

\newcommand{\hardt}[1]{\mathrm{HT(#1)}}
\DeclareMathOperator*{\rank}{rank}
\DeclareMathOperator*{\trace}{tr}
\DeclareMathOperator*{\argmin}{arg\,min}

\DeclareMathOperator*{\diag}{diag}

\DeclareMathOperator*{\var}{var}
\DeclarePairedDelimiter\floor{\lfloor}{\rfloor}
\usepackage{soul} % for using \st to strikethrough text

\renewcommand{\arraystretch}{1.0}

\numberwithin{equation}{section}

\title{Supervised Factor Modeling for High-Dimensional Linear Time Series}
\author{Feiqing Huang, Kexin Lu and Guodong Li
	\\ \textit{Department of Statistics and Actuarial Science, University of Hong Kong} }

\begin{document}
	
	\setlength{\parindent}{16pt}
	
	\maketitle
	
	\begin{abstract}
Motivated by Tucker tensor decomposition, this paper imposes low-rank structures to the column and row spaces of coefficient matrices in a multivariate infinite-order vector autoregression (VAR), which leads to a supervised factor model with two factor modelings being conducted to responses and predictors simultaneously.
Interestingly, the stationarity condition implies an intrinsic weak group sparsity mechanism of infinite-order VAR, and hence a rank-constrained group Lasso estimation is considered for high-dimensional linear time series.
Its non-asymptotic properties are discussed thoughtfully by balancing the estimation, approximation and truncation errors.
Moreover, an alternating gradient descent algorithm with thresholding is designed to search for high-dimensional estimates, and its theoretical justifications, including statistical and convergence analysis, are also provided.
Theoretical and computational properties of the proposed methodology are verified by simulation experiments, and the advantages over existing methods are demonstrated by two real examples.
	\end{abstract}
	
	
	\newpage	
	\section{Introduction}
The high-speed advance in technology has spurred the rapid growth of high-dimensional data, especially time-dependent data, and examples can be found in many fields such as economics, finance, biology and neuroscience \citep{HL13, DP16,nicholson2020high,PT21}.
It is urgent to develop suitable models and methods for these larger and more complex time series data and, as arguably the most widely used low-dimensional (or multivariate) time series model, the vector autoregressive (VAR) model has been a primary workhorse for many high-dimensional tasks \citep{basu2019low,basu2015regularized,kock2015oracle,ZC21}.
However, the performance of VAR models can be seriously limited by the finite order in real applications, since information further in the past is often needed for fully capturing the complex temporal dependence.
As a remedy, the vector autoregressive moving average (VARMA) model has a parametric VAR$(\infty)$ form, and it hence can provide a better forecasting accuracy when time series are longer \citep{CEK16,WBBM21}.  
Comparing with VAR models, the VARMA model is usually more difficult to fit, and it is also harder to interpret \citep{Tsay14}.
Moreover, a simple VARMA model may still have the difficulty in capturing more complicated temporal dynamics such as the latent seasonality and inactive lags. 
This motivates us to directly study the most general linear stationary time series, i.e. the general linear process, and it includes both VAR and VARMA models as special cases.
Specifically, the general linear process and its VAR($\infty$) representation have the forms of
\begin{equation}\label{eq:VAR_inf}
	\mathbf{y}_t =  \sum_{j=1}^{\infty}\bm{\Psi}_j\bm{\varepsilon}_{t-j}+\bm{\varepsilon}_t \hspace{9mm}\text{and}\hspace{9mm}\mathbf{y}_t=\sum_{j=1}^\infty \bm{A}_j\mathbf{y}_{t-j}+\bm{\varepsilon}_t,
\end{equation}
respectively, where $\mathbf{y}_t\in\mathbb{R}^N$, the white noise $\bm{\varepsilon}_t\in\mathbb{R}^N$, and $\bm{\Psi}_j$'s and $\bm{A}_j$'s are $N\times N$ coefficient matrices \citep{BD91,Lutkepohl2005,Tsay14}.
Under low-dimensional settings, a VAR$(T_0)$ sieve model is usually used for estimation, and the corresponding asymptotic properties have been well established by choosing a suitable running order $T_0$ to control truncation errors; see \cite{GZ01,Lutkepohl2005,Li-Leng-Tsai2014}.
There is no discussion on general linear processes under high-dimensional settings in the literature, and this paper aims to fill this gap by providing inference tools with non-asymptotic properties.

The number of parameters at \eqref{eq:VAR_inf} increases quadratically with respect to the variable dimension, i.e. there are $N^2$ parameters for each coefficient matrix $\bm{A}_j$, and it is necessary to first restrict the parameter space along this dimension to achieve reasonable estimation.
A direct method is to assume that the coefficient matrices are sparse and then to apply sparsity-inducing regularized estimation; see, e.g., the $\ell_1$-regularization for VAR models \citep{basu2015regularized,han2015direct,kock2015oracle} and VARMA models \citep{WBBM21}.
However, its success may be discounted due to the three reasons below, given that time series data have a special structure and all predictors are the lagged values of responses.
First, the VAR models with sparsity may not have the most appropriate interpretation unless they have special sparsity structures, such as the banded VAR \citep{guo2016high} and network VAR \citep{zhu2017network} models. This drawback becomes more serious for a sparse VARMA model \citep{WBBM21} since the sparse MA coefficient matrices may not lead to sparse coefficient matrices in its VAR$(\infty)$ representation.
Secondly, the stationarity condition is usually required to establish theoretical properties, and this may force the entries of coefficient matrices to shrink to zero as $N\rightarrow \infty$, undermining the elementwise sparsity mechanism; see Remark 1 in \cite{Wang2021High}.	
Finally, for financial and economic time series, one often observes strong cross-sectional dependence among the $N$ scalar series, and it can be described by assuming that the $N$ variables are driven by a small number of common latent factors \citep{BN08, lam2012factor, BW16,fan2022bridging}.
This is different from the application of sparsity constraints; see, for example, time-course gene expression data, where dependence among the $N$ genes is believed to be sparse \citep{LZLR09}.

Another commonly used approach for dimension reduction is to impose low-rank structures to coefficient matrices \citep{basu2019low,billio2023bayesian}, and it will lead to the reduced-rank \citep{Velu13} and low-multilinear-rank \citep{Wang2021High} models when the low-rank assumption is applied to VAR models.
Along this line and motivated by Tucker tensor decomposition, this paper first imposes different low-rank structures to the column and row spaces of coefficient matrices $\bm{A}_j$'s, and they can then be interpreted as projecting responses and predictors into a small number of latent factors, \textit{response and predictor factors}, respectively. 
The response factors are used to summarize all predictable components of the market, while the predictor factors contain all driving forces; see Section 2.1 for details.
In fact, in the econometric literature, high-dimensional time series are usually analyzed by factor models, and the estimated factors and their loading matrices can provide useful insights into the interactive mechanism of large financial and economic systems \citep{lam2012factor, BW16,fan2022bridging}.
As a result, the proposed model can be called the \textit{supervised factor model} to emphasize its interpretations from unsupervised factor modeling perspectives, marking the first main contribution of this paper.

% The infinite number of lagged values, $\mathbf{y}_{t-j}$'s, of responses is another hindrance for the high-dimensional VAR$(\infty)$ modeling at \eqref{eq:VAR_inf}, and this issue does not exist for traditional linear regression while is common in the literature of time series \citep{Lutkepohl2005,Tsay14}.
The infinite number of lagged values, $\mathbf{y}_{t-j}$'s, of responses is another hindrance for the high-dimensional VAR$(\infty)$ modeling at \eqref{eq:VAR_inf}, and this issue does not exist for traditional linear regression while is common in the literature of time series \citep{Lutkepohl2005,Tsay14}.
Note that all predictors $\mathbf{y}_{t-j}$'s, as well as the responses $\mathbf{y}_{t}$, are observed values of the same $N$ variables. 
This makes it complicated to construct reasonable dimension reduction methods along lags,
%This makes difficult the reasonable dimension reduction along lags, 
and most existing works ignore this issue by choosing a fixed order or even order one \citep{basu2015regularized,han2015direct}.
\cite{nicholson2017varx} considered several group-wise sparsity mechanisms to the lag dimension, such as the hierarchical group Lasso (HLag), and they are all carefully designed for certain special structure of time series data; see also \cite{wilms2017interpretable}.
These methods can potentially enhance interpretability and performance on particular datasets, while their frameworks may be inevitably limited for a general purpose.
Interestingly,  the stationarity condition of general linear processes intrinsically induces a natural constraint on coefficient matrices, i.e. they lie within a generalized $\ell_1$-ball, $\{\sum_{j=1}^{\infty}\|\bm{A}_j\|_{\Fr}\leq R\}$ with radius $0<R<\infty$.
This constraint exactly matches the weak group sparsity scenario with each coefficient matrix being a group of parameters \citep{raskutti2011minimax, Wainwright19}, and it, together with the low-rank assumption along the variable dimension, leads to the second main contribution of this paper, i.e. the high-dimensional rank-constrained group Lasso estimation and its non-asymptotic properties in Section 2.2 and 2.3, respectively. 
Specifically, we first choose a suitable running order $T_0$ to control truncation errors, and then the weak group sparsity design enables us to automatically account for the estimation and approximation errors; see Figure \ref{fig:glasso} for illustration.

We next consider algorithms to search for the high-dimensional estimate, which involves both the low-rank constraint and group sparsity. 
First, the low-rank structures are defined via matrix or tensor decomposition, and it is unavoidable for the estimating procedure to search for orthonomal matrices usually by the singular value decomposition, resulting in a time-consuming algorithm \citep{tu2016low,chi2019nonconvex}.
A novel solution in the literature is to add regularization terms rather than to assume the orthogonality directly \citep{wang2017unified,park2018finding,HWZ21}, and we adopt this method to tackle the low-rank constraints. 
Secondly, this is a non-convex optimization problem, and it is usually difficult to achieve the global optimality and to conduct the convergence analysis.
Following the non-convex matrix or tensor factorization literature, this paper considers an alternating gradient descent approach, and it is supposed to enjoy low computation cost and storage complexity; see \cite{chi2019nonconvex} and references therein.
Finally, due to the alternating mechanism, thresholding methods are more convenient to handle the group sparsity, and there are two choices, the hard- and soft-thresholding, in the literature \citep{Wainwright19}.
Although the soft-thresholding is more relevant to the lasso problem, the hard-thresholding will lead to a more stable algorithm for the VAR$(\infty)$ modeling; see Section 3.1 for details, and it is also more convenient to conduct convergence analysis \citep{tropp2010computational,shen2017tight}. More importantly, the hard-thresholding will result in unbiased estimators, while the soft-thresholding usually leads to biased ones; see Figure \ref{fig:glasso} for the illustration.
As a result, the third main contribution of this paper is to propose an alternating gradient descent algorithm with hard-thresholding in Section 3.1, and its theoretical justifications, including both statistical and convergence analysis, are also provided in Section 3.2.

\begin{figure}[t]
	\centering
	\includegraphics[width=1.\linewidth]{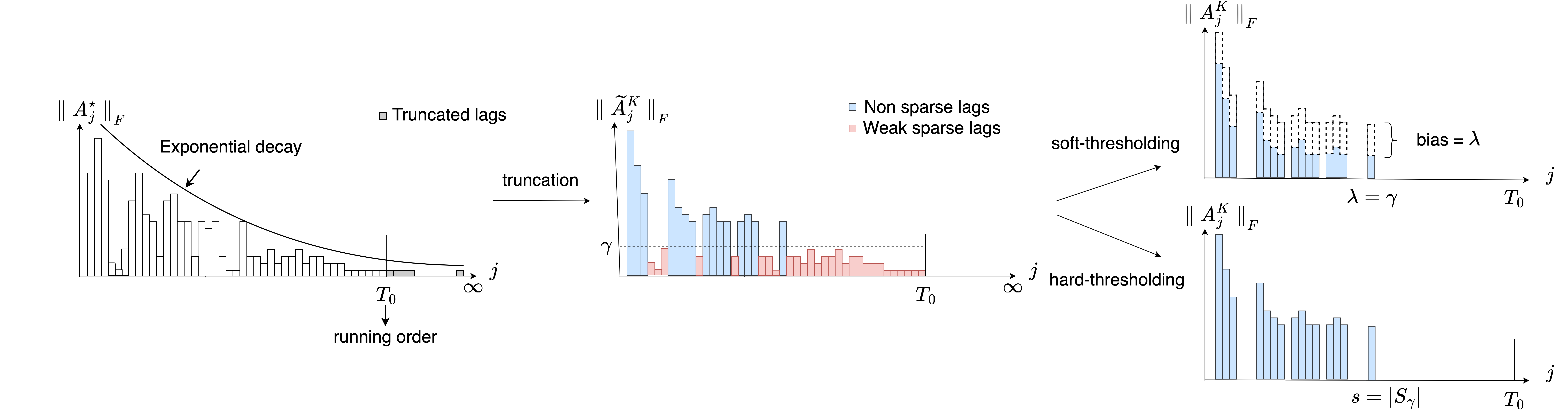}
	\caption{Dimension reduction along lags in the VAR($\infty$) modeling with weak group sparsity. Coefficient matrices with lags greater than $T_0$ are first truncated, and estimated ones with small magnitude are further suppressed, leading to the truncation and approximation errors, respectively.
		The remaining estimated coefficient matrices contribute to the estimation error, while the soft- or hard-thresholding will result in biased or unbiased estimators, respectively.
		%	The first panel involves truncating all coefficient matrices with lags greater than $T_0$, resulting in a truncation error. Next, smaller coefficient matrices are suppressed, leading to an approximation error (middle panel). The selected lags, determined by either soft- or hard-thresholding, can lead to biased or unbiased estimators (right panel), resulting in an estimation error.
		%coefficients' norms of a weakly sparse VAR($\infty$) process at all lags. 
		%Estimation error arises from the non sparse lags selected by either soft- or hard-thresholding, while the weak sparse lags and truncated lags lead to approximation and truncation errors, respectively.  For a the soft-thresholding operation 	All lags at $j\geq T_0+1$ are included in the truncation error. Consider the first $T_0$ lags, the non-zero index set $S_{\gamma}$ consists of all lags with norms exceeding a fixed threshold value $\gamma$, while $S_{\gamma}^c$ consists of the rest. Group lasso penalty translates into soft-thresholding operation with $\bm{A}_j =  (1-\lambda/\|\bm{A}_j\|_{\Fr})_{+}\bm{A}_j$ and $\lambda=\gamma$. Hard-thresholding operation selects the exact same lags by choosing $s=|S_{\gamma}|$. 
		\label{fig:glasso}}
\end{figure}

In addition, Section 4 conducts simulation experiments to evaluate the finite-sample performance of the proposed methodology, and its usefulness is further demonstrated by two empirical examples in Section 5. Section 6 gives a short conclusion and discussion, and all technical proofs are relegated to a supplementary file.
Throughout the paper, tensors are denoted by calligraphic capital letters; see, e.g. $\cm{A}$, $\cm{B}$, etc., and a brief introduction to tensor notations and Tucker decomposition is provided in the supplementary file. 
For two scalars $a$ and $b$, we denote $a\wedge b=\min\{a,b\}$ and $a\vee b=\max\{a,b\}$. For vectors $\bm{a}$ and $\bm{b}$, denote by $\langle\bm{a},\bm{b}\rangle = \sum_{j}a_jb_j$ and $\|\bm{a}\|_2=\sqrt{\langle\bm{a},\bm{a}\rangle}$ the inner product and $\ell_2$-norm, respectively.
For a matrix $\bm{A}\in\mathbb{R}^{d_1\times d_2}$, let $\bm{A}^\prime$, $\rank(\bm{A})$, $\sigma_{\max}(\bm{A})$ (or $\sigma_{\min}(\bm{A})$),  $\lambda_{\max}(\bm{A})$ (or $\lambda_{\min}(\bm{A})$), $\|\bm{A}\|_\op=\sigma_{\max}(\bm{A})$  and $\|\bm{A}\|_\Fr=\sqrt{\sum_{i,j}\bm{A}_{ij}^2}$
be its transpose, rank, largest (or smallest non-zero) singular value, largest (or smallest) eigenvalue, operator norm and  Frobenius norm, respectively. 
Moreover, for any $d_1\geq d_2$, the set of orthonormal matrices is denoted by $\mathcal{O}^{d_1\times d_2} := \{\bm{A}\in\mathbb{R}^{d_1\times d_2}\mid \bm{A}^\prime \bm{A} = \bm{I}_{d_2}\}$, where $\bm{I}_{d_2}$ is a $d_2\times d_2$ identity matrix.
On the other hand, for any two sequences $x_n$ and $y_n$, we denote $x_n\lesssim y_n$ (or $x_n\gtrsim y_n$) if there exists an absolute constant $C>0$ such that $x_n\leq C y_n$ (or $x_n\geq C y_n$). Write $x_n\asymp y_n$ if $x_n\lesssim y_n$ and $x_n\gtrsim y_n$, $x_n = O(y_n)$ if $x_n\lesssim y_n$, and $x_n = o(y_n)$ if $\lim_{n\rightarrow \infty} x_n/y_n = 0$. 

\section{High-dimensional linear time series modeling}\label{sec:lass}
\subsection{Low-rank linear time series and supervised factor models}

This subsection imposes low-rank assumptions to general linear processes, which leads to a supervised factor model \citep{wang2022high,wang2023high} in this paper.

Consider the $N$-dimensional general linear process at \eqref{eq:VAR_inf}, where the error terms $\{\bm{\varepsilon}_t\}$ are independent and identically distributed ($i.i.d.$).  Its coefficient matrices $\bm{\Psi}_j$'s may have a low-rank structure, especially when the number of variables $N$ is large.
Specifically, let $\mathbb{M}_1=\text{colspace}\{\bm{\Psi}_j,j\geq 1\}$ and $\mathbb{M}_2=\text{rowspace}\{\bm{\Psi}_j,j\geq 1\}$ be the column and row spaces of coefficient matrices $\bm{\Psi}_j$'s, respectively, and their dimensions are denoted by $r_i=\text{dim}(\mathbb{M}_i)\leq N$ with $i=1$ and 2. 
Note that $r_1$ and $r_2$ are not equal in general, and it is natural to consider the above low-rank constraints from the viewpoint of Tucker tensor decomposition; see discussions at the end of this subsection.
% and, when $r_1=r_2=N$, there is no dimension reduction.
In addition, the matrix polynomial for the general linear process is defined as $\bm{\Psi}(z)=\bm{I}-\sum_{j=1}^{\infty}\bm{\Psi}_jz^j$, where $z\in \mathbb{C}$, and $\mathbb{C}$ is the complex space. 

\begin{assumption}[Invertibility condition] \label{assum:glp}
	The determinant of $\bm{\Psi}(z)$ is not equal to zero for all $|z|<1$, and  $\sum_{j=1}^{\infty}\|\bm{\Psi}_j\|_{\op} < \infty$.
\end{assumption}

\begin{proposition}\label{prop:Aj}
	If Assumption \ref{assum:glp} holds, then the VAR$(\infty)$ form at \eqref{eq:VAR_inf} can be uniquely identified with $\sum_{j=1}^{\infty}\|\bm{A}_j\|_{\op} < \infty$. Moreover, $\mathbb{M}_1$ and $\mathbb{M}_2$ are also the column and row spaces of coefficient matrices $\bm{A}_j$'s, respectively. 
\end{proposition}

From the above proposition, the general linear process has an equivalent VAR$(\infty)$ form, whose coefficient matrices satisfy $\bm{A}_j = \bm{\Psi}_j - \sum_{k=1}^{j-1}\bm{\Psi}_{j-k}\bm{A}_k$ for all $j\geq 1$.
Moreover, it can be verified that the determinant of $\bm{A}(z)$ is not equal to zero for all $|z|<1$, where $\bm{A}(z)=\bm{I}-\sum_{j=1}^{\infty}\bm{A}_jz^j$ is the VAR$(\infty)$ matrix polynomial.
Lastly, coefficient matrices of the general linear process and its VAR$(\infty)$ representation share the same column and row spaces.

Consider two matrices $\bm{U}_1\in\mathcal{O}^{N\times r_1}$ and $\bm{U}_2\in\mathcal{O}^{N\times r_2}$, which contain the bases of subspaces $\mathbb{M}_1$ and $\mathbb{M}_2$, respectively. 
Then, by some tensor algebra, there exist two sequences of $r_1\times r_2$ matrices $\{\bm{H}_j, \bm{G}_j,j\geq 1\}$ such that $\bm{\Psi}_j = \bm{U}_1\bm{H}_j \bm{U}_2^\prime$ and $\bm{A}_j = \bm{U}_1\bm{G}_j \bm{U}_2^\prime$ for all $j\geq 1$.
As a result, under Assumption \ref{assum:glp} and the low-rank constraint, model \eqref{eq:VAR_inf} can be rewritten into
\begin{equation}\label{eq:VAR-low-rank}
	\mathbf{y}_t = \bm{U}_1\sum_{j=1}^{\infty}  \bm{H}_j \bm{U}_2^\prime \bm{\varepsilon}_{t-j} +\bm{\varepsilon}_t \hspace{6mm}\text{and}\hspace{6mm}\mathbf{y}_t = \bm{U}_1\sum_{j=1}^{\infty}  \bm{G}_j \bm{U}_2^\prime \mathbf{y}_{t-j} +\bm{\varepsilon}_t.
\end{equation}
Note that $\bm{U}_1$ and $\bm{U}_2$ are not unique, while the projection matrices of $\mathbb{M}_1$ and $\mathbb{M}_2$ can be uniquely defined by $\bm{P}_1=\bm{U}_1\bm{U}_1^\prime$ and $\bm{P}_2 = \bm{U}_2\bm{U}_2^\prime$, respectively. 
Moreover, for $i=1$ and 2, let $\bm{U}_i^{\perp}\in\mathcal{O}^{N\times (N-r_i)}$ such that $(\bm{U}_i,\bm{U}_i^{\perp})$ is an $N\times N$ orthonormal matrix, and then $\bm{P}_i^{\perp}=\bm{U}_i^{\perp}\bm{U}_i^{\perp\prime}$ is the projection matrix of $\mathbb{M}_i^{\perp}$, i.e. orthogonal complement of subspace $\mathbb{M}_i$.

Note that model \eqref{eq:VAR-low-rank} involves two types of dimension reduction.
We first consider the projection of $\mathbf{y}_t$ onto subspace $\mathbb{M}_1$ and its orthogonal complement, i.e. $\mathbf{y}_t=\bm{P}_1\mathbf{y}_t+\bm{P}_1^{\perp}\mathbf{y}_t$, and these two parts can be verified to have completely different dynamic structures,
%\begin{equation*}
\[
\bm{P}_1\mathbf{y}_t=\bm{U}_1\sum_{j=1}^{\infty}  \bm{H}_j \bm{U}_2^\prime \bm{\varepsilon}_{t-j} +\bm{P}_1\bm{\varepsilon}_t
\hspace{6mm}\text{and}\hspace{6mm}
\bm{P}_1^{\perp}\mathbf{y}_t=\bm{P}_1^{\perp}\bm{\varepsilon}_t,
\]
%\end{equation*}
where all information of $\mathbf{y}_t$ related to temporally dependent structures is contained in $\mathbb{M}_1$, whereas $\mathbb{M}_1^{\perp}$ includes only purely idiosyncratic and serially independent components.
%Moreover, it holds that $\mathrm{cov}(\bm{P}_1\bm{\varepsilon}_t,\bm{P}_1^{\perp}\bm{\varepsilon}_t)=\bm{P}_1\mathrm{var}(\bm{\varepsilon}_t)\bm{P}_1^{\perp\prime}=0$, and $\bm{P}_1\mathbf{y}_t$ and $\bm{P}_1^{\perp}\mathbf{y}_t$ are independent if $\bm{\varepsilon}_t$ is further assumed to be normally distributed.
In fact, model \eqref{eq:VAR-low-rank} has a form of static factor models,
\begin{equation*}\label{eq:factor}
	\mathbf{y}_t = \bm{U}_1 \bm{f}_t+\bm{\varepsilon}_t\hspace{5mm}\text{with}\hspace{5mm}\bm{f}_t = \sum_{j=1}^{\infty}  \bm{H}_j \bm{U}_2^\prime \bm{\varepsilon}_{t-j}=
	\sum_{j=1}^{\infty} \bm{G}_j \bm{U}_2^\prime \mathbf{y}_{t-j},
\end{equation*}
where $\bm{f}_t\in\mathbb{R}^{r_1}$ contains $r_1$ latent factors, and $\bm{U}_1$ is the corresponding loading matrix; see \cite{Pan-Yao2008,lam2012factor,BW16}.
Consequently, we call $\bm{f}_t=\bm{U}_1^{\prime}\mathbf{y}_t-\bm{U}_1^{\prime}\bm{\varepsilon}_t$ or $\bm{U}_1^{\prime}\mathbf{y}_t$ the \textit{response factor} since it summarizes all predictable components in the response, and accordingly $\mathbb{M}_1$ can be referred to as the \textit{response factor space}.

On the other hand, for the dimension reduction on predictors, we project $\mathbf{y}_{t-j}$ onto $\mathbb{M}_2$ and $\mathbb{M}_2^{\perp}$, and it holds that $\mathbf{y}_{t-j}=\bm{P}_2\mathbf{y}_{t-j}+\bm{P}_2^{\perp}\mathbf{y}_{t-j}$.
Moreover, for $N$-dimensional random vectors $\mathbf{x}_1$, $\mathbf{x}_2$ and $\mathbf{z}_j$'s with $j\geq 1$, the partial covariance function is usually used to measure the relationship between $\mathbf{x}_1$ and $\mathbf{x}_2$ after removing the effects of $\mathbf{z}_j$'s, and it has the form of 
$
\mathrm{pcov}(\mathbf{x}_1,\mathbf{x}_2| \mathbf{z}_1,\mathbf{z}_2,\ldots)=\mathrm{cov}(\mathbf{x}_1-\widehat{\mathbf{x}}_1,\mathbf{x}_2-\widehat{\mathbf{x}}_2),
$
where, for $i=1$ and 2, $\widehat{\mathbf{x}}_i=\sum_{j=1}^{\infty}\widehat{\bm{B}}_j^{(i)}\mathbf{z}_j$ and $(\widehat{\bm{B}}_1^{(i)},\widehat{\bm{B}}_2^{(i)},\ldots)=\argmin E\|\mathbf{x}_i-\sum_{j=1}^{\infty}\bm{B}_j\mathbf{z}_j\|_2^2$; see \cite{Fan2003,Tsay14}.
As a result, if $E\|\mathbf{y}_t\|_2^2<\infty$, then
\[
\mathrm{pcov}(\mathbf{y}_t, \mathbf{y}_{t-j} | \bm{P}_2\mathbf{y}_{t-1},\bm{P}_2\mathbf{y}_{t-2},\ldots) =\mathrm{pcov}(\mathbf{y}_t, \bm{P}_2^{\perp}\mathbf{y}_{t-j} | \bm{P}_2\mathbf{y}_{t-1},\bm{P}_2\mathbf{y}_{t-2},\ldots) =0 \hspace{2mm}\text{for all}\hspace{2mm}j\geq 1,
\]
% $\mathrm{cov}(\bm{P}_2\mathbf{y}_{t-j},\bm{P}_2^{\perp}\mathbf{y}_{t-l})=\bm{P}_2\mathrm{cov}(\mathbf{y}_{t-j},\mathbf{y}_{t-l})\bm{P}_2^{\perp\prime}=0$ for all $j, l\geq 1$. As a result,
%\[
%\cov(\mathbf{y}_t, \bm{P}_2^{\perp}\mathbf{y}_{t-j})=0 \hspace{3mm}\text{and}\hspace{3mm}
%\cov(\mathbf{y}_t, \bm{P}_2\mathbf{y}_{t-j})=\cov(\mathbf{y}_t, \mathbf{y}_{t-j}) \hspace{3mm}\text{for all $j\geq 1$}.
%\]
i.e. the space $\mathbb{M}_2$ can summarize all information of $\mathbf{y}_{t-j}$ that contributes to predicting $\mathbf{y}_{t}$, or $\bm{U}_2^\prime\mathbf{y}_{t-j}$ contains all driving forces of the market.
Thus, we call $\bm{U}_2^\prime\mathbf{y}_{t-j}$ the \textit{predictor factor} for simplicity, and $\mathbb{M}_2$ is referred to as the \textit{predictor factor space}.
Since model \eqref{eq:VAR-low-rank} is a supervised problem in nature, and we call it the \textit{supervised factor model} to emphasize the above interpretations from unsupervised factor modeling perspectives \citep{lam2012factor,BW16}.

\begin{remark}
	The response factor space $\mathbb{M}_1$ is exactly the factor space of $\{\mathbf{y}_t\}$ in the unsupervised factor modeling settings, and it hence can be easily estimated by many model-free methods in the literature \citep{Pan-Yao2008,Lam-Yao-Bathia2011,lam2012factor,gao2021two}. 
	On the other hand, for a simple diagonal VAR model, the proposed supervised factor model cannot be applied since both column and row spaces of coefficient matrices have full rank, and we may need a model with sparsity for it; see, e.g., \citep{basu2015regularized,guo2016high}. 
\end{remark}
\begin{remark}
Consider a special case of model \eqref{eq:VAR-low-rank} with $\bm{U}_1=\bm{U}_2$, i.e. the response and predictor factors are identical, and then the VAR$(\infty)$ representation can be rewritten into
\[
\bm{f}_t = \sum_{j=1}^{\infty}  \bm{G}_j \bm{f}_{t-j} +\bm{U}_1^{\prime}\bm{\varepsilon}_t \hspace{5mm}\text{with}\hspace{5mm}
\bm{f}_t =\bm{U}_1^{\prime} \mathbf{y}_{t},
\]
where $\{\bm{U}_1^{\prime}\bm{\varepsilon}_t\}$ is the new white noise sequence.
It hence leads to the commonly used dynamic factor model in the literature: we first conduct factor modeling on $\{\mathbf{y}_{t}\}$, and an VAR$(\infty)$ model is then considered for the summarized factors $\bm{f}_t$; see \cite{SW05, SW11}.
Since the response and predictor factors at \eqref{eq:VAR-low-rank} are not equal in general, i.e. $\bm{U}_1\neq \bm{U}_2$, the supervised factor model can provide a better performance in forecasting high-dimensional time series.
\end{remark}

Model \eqref{eq:VAR-low-rank} is partially motivated from tensor techniques, and the two types of dimension reduction can be imposed naturally from viewpoints of tensor decomposition; see the supplementary file for more details on tensor notations and decomposition.
Specifically, for the VAR$(\infty)$ form at \eqref{eq:VAR_inf}, we first rearrange the coefficient matrices into a tensor $\cm{A}_{\infty}\in\mathbb{R}^{N \times N\times \infty}$ such that its mode-1 matricization is $(\cm{A}_{\infty})_{(1)}=(\bm{A}_1,\bm{A}_2,\ldots)$, and then its mode-2 matricization assumes the form of $(\cm{A}_{\infty})_{(2)}=(\bm{A}_1^{\prime},\bm{A}_2^{\prime},\ldots)$.
Note that the column spaces of $(\cm{A}_{\infty})_{(1)}$ and $(\cm{A}_{\infty})_{(2)}$ are $\mathbb{M}_1$ and $\mathbb{M}_2$, respectively, and $r_1=\rank\{(\cm{A}_{\infty})_{(1)}\}$ and $r_2=\rank\{(\cm{A}_{\infty})_{(2)}\}$ are the first two Tucker ranks.
Accordingly, we have the Tucker decomposition \citep{tucker1966some,delathauwer2000multilinear} below,
\begin{equation}\label{eq:low-tucker-rank}
	\cm{A}_{\infty}=\cm{G}_{\infty}\times_1 \bm{U}_1\times_2 \bm{U}_2,
\end{equation}
where $\cm{G}_{\infty}\in\mathbb{R}^{r_1 \times r_2\times \infty}$, $\bm{U}_1\in\mathbb{R}^{N\times r_1}$ and $\bm{U}_2\in\mathbb{R}^{N\times r_2}$.
Similarly, coefficient matrices of general linear processes can also be represented in the form of tensors and tensor decomposition.

Note that Tucker decomposition at \eqref{eq:low-tucker-rank} is not unique, since $\cm{A}_{\infty}=\cm{G}_{\infty}\times_1 \bm{U}_1\times_2 \bm{U}_2=(\cm{G}_{\infty}\times_1 \bm{O}_1\times_2 \bm{O}_2)\times_1 (\bm{U}_1\bm{O}_1^{-1})\times_2 (\bm{U}_2\bm{O}_2^{-1})$ for any invertible matrices $\bm{O}_i\in\mathbb{R}^{r_i\times r_i}$ with $i=1$ and 2.
In particular, we can choose $\bm{U}_1$, $\bm{U}_2$ and $\bm{G}_j$'s in \eqref{eq:VAR-low-rank} and let $\bm{G}_j$ be the $j$-th frontal slice of $\cm{G}_{\infty}$ for $j\geq 1$, i.e. $(\cm{G}_{\infty})_{(1)}=(\bm{G}_1,\bm{G}_2,\ldots)$.
As a result, the VAR$(\infty)$ form at \eqref{eq:VAR_inf}, together with the low-Tucker-rank constraint at \eqref{eq:low-tucker-rank}, is equivalent to the supervised factor model at \eqref{eq:VAR-low-rank}.

%We next refer the supervised factor model to \eqref{eq:VAR_inf} and \eqref{eq:low-tucker-rank}, since the reminding of this paper concentrates on its statistical inference.

%In other words, the coefficient matrices of $\bm{U}_1$, $\bm{U}_2$ and $\bm{G}_j$'s at \eqref{eq:VAR-low-rank} are not identifiable, while $\bm{A}_j$'s can be uniquely identified.

\subsection{High-dimensional estimation}\label{sec:lasso-estimator}

For an observed time series $\{\mathbf{y}_1,\ldots,\mathbf{y}_T\}$ generated by model \eqref{eq:VAR_inf} with the low-rank constraint \eqref{eq:low-tucker-rank}, its true coefficient matrices are denoted by $\bm{A}_j^*$'s, and this paper adopts the VAR sieve approximation method to estimate them, i.e.
\begin{equation}\label{eq:trunc_tensor}
	\mathbf{y}_t=\sum_{j=1}^{T_0}\bm{A}_j^*\mathbf{y}_{t-j}+\widetilde{\bm{\varepsilon}}_t,\hspace{5mm}\widetilde{\bm{\varepsilon}}_t=\bm{\varepsilon}_t+  \bm{r}_t \hspace{2mm}\text{and}\hspace{2mm} \rank\{(\cm{A}_{\infty}^*)_{(i)}\} \leq r_i \text{ with $i = 1$ and 2},
\end{equation}
where $T_0$ is the running order of VAR models, $\bm{r}_t=\sum_{j=T_0+1}^{\infty}\bm{A}_j^*\mathbf{y}_{t-j}$ is the truncated term, and $\cm{A}_{\infty}^*$ is the true full coefficient tensor.

The coefficient matrices of model \eqref{eq:trunc_tensor} are $\bm{A}_j$'s with $1\leq j\leq T_0$, and they can be rearranged into a coefficient tensor $\cm{A}\in\mathbb{R}^{N\times N\times T_0}$ such that its mode-1 matricization is $\cm{A}_{(1)}=(\bm{A}_1,\bm{A}_2,\ldots,\bm{A}_{T_0})\in\mathbb{R}^{N\times NT_0}$.
Let $\mathbf{x}_{t}=(\mathbf{y}_{t-1}^\prime, \mathbf{y}_{t-2}^\prime, \ldots, \mathbf{y}_{t-T_0}^\prime)^\prime\in \mathbb{R}^{NT_0}$, 
$\bm{X}=(\mathbf{x}_T,\mathbf{x}_{T-1},\dots,\mathbf{x}_{T_0+1})\in \mathbb{R}^{NT_0\times (T-T_0)}$ and $\bm{Y}=(\mathbf{y}_T,\mathbf{y}_{T-1},\dots,\mathbf{y}_{T_0+1})\in\mathbb{R}^{N\times (T-T_0)}$. 
%\textcolor{purple}{The ordinary least squares loss function of model \eqref{eq:trunc_tensor} has the form of}
Consider the ordinary least squares estimation, and then the loss function has the form of
\[
\mathcal{L}(\cm{A})= \frac{1}{2T_1}\sum_{t=T_0+1}^{T}\|\mathbf{y}_t-\sum_{j=1}^{T_0}\bm{A}_j\mathbf{y}_{t-j}\|_{\Fr}^2 =\frac{1}{2T_1}\|\bm{Y}-\cm{A}_{(1)}\bm{X}\|_{\Fr}^2,
\]
where the effective sample size is $T_1=T-T_0$.
Denote by $\cm{A}^*\in\mathbb{R}^{N\times N\times T_0}$ the true coefficient tensor, and it is a truncated tensor obtained by removing all $\bm{A}_j^*$ with $j> T_0$ from the $N\times N\times \infty$ full coefficient tensor $\cm{A}^*_\infty$. 
Define the parameter space
\begin{equation*}%\label{eq:parspace1}
	\bm{\Theta}(r_1, r_2) = \{\cm{A}\in\mathbb{R}^{N\times N\times T_0}\mid \rank(\cm{A}_{(1)})\leq r_1, \rank(\cm{A}_{(2)})\leq r_2\},
\end{equation*}
and the low-Tucker-rank constraint at \eqref{eq:trunc_tensor} implies that $\cm{A}^*\in \bm{\Theta}(r_1, r_2)$.

%Denote by $\bm{r}_t=\sum_{j=T_0+1}^{\infty}\bm{A}_j^*\mathbf{y}_{t-j}$ the truncated term at model \eqref{eq:trunc_tensor}, i.e. $\widetilde{\bm{\varepsilon}}_t=\bm{\varepsilon}_t+ \bm{r}_t$.
To obtain theoretical justifications, the truncated term $\bm{r}_t$ is required to approach zero quickly, which requests $\bm{A}_j^*$ or $\bm{\Psi}_j^*$ to decay at a sufficient rate as $j\rightarrow \infty$, where $\bm{\Psi}_j^*$'s are true coefficient matrices of the corresponding general linear process.
While Assumption \ref{assum:glp} is sufficient to achieve asymptotic properties for the low-dimensional (or multivariate) time series, we need a stronger exponential decay below to establish non-asymptotic properties for the high-dimensional case.

\begin{assumption}[Exponential decay] \label{assum:Adecay}
	There exists some $\rho\in(0,1)$ such that $\|\bm{\Psi}_j^*\|_{\op} =O(\rho^j)$ and $\|\bm{A}_j^*\|_{\op} = O(\rho^j)$ as $j\rightarrow \infty$.
\end{assumption}

The above constraint is mild since all VAR and VARMA processes are still included.
%Moreover, since $\bm{A}_j^* = \bm{\Psi}_j^* - \sum_{k=1}^{j-1}\bm{\Psi}_{j-k}^*\bm{A}_k^*$ for all $j\geq 1$, we can show that, if $\|\bm{\Psi}_j^*\|_{\op} =O(\rho^j)$, then $\|\bm{A}_j^*\|_{\op} =O(\rho_1^j)$ with some $\rho_1\in(\rho, 1)$, and vice versa. In other words, it is sufficient to restrict one of $\bm{\Psi}_j^*$ and $\bm{A}_j^*$, while the above assumption can simplify the forthcoming presentation.
On the other hand, from Assumption \ref{assum:Adecay} and the low-rank condition at \eqref{eq:low-tucker-rank}, it can be verified that $\sum_{j=1}^{\infty}\|\bm{A}_j^*\|_{\Fr} \leq C\rho^{-1}(r_1\wedge r_2)$, where $C$ is an absolute constant, and then the true coefficients of model \eqref{eq:trunc_tensor} will be within a generalized $\ell_1$-ball, $\{\cm{A}: \sum_{j=1}^{T_0}\|\bm{A}_j\|_{\Fr}\leq R\}$.
This feature exactly matches the scenario of weak sparsity \citep{raskutti2011minimax,Wainwright19}, and hence an automated lag selection procedure can be inspired.
Specifically, if we regard each  $\bm{A}_j^*$ as a group of parameters,  since  any $\bm{A}_j^*$ with  a relatively large $j$ must be close to (if not exactly) a zero matrix, then the target coefficient tensor
$\cm{A}^*$ must be weakly group-sparse. 
As a result, this paper considers the following rank-constrained group Lasso estimator of $\cm{A}^*$,
\begin{equation}\label{eq:minimizer}
	\cm{\widehat{A}} = \argmin_{\scriptsize\cm{A}\in\bm{\Theta}(r_1, r_2)}\mathcal{L}(\cm{A}) +\lambda \|\cm{A}\|_{\ddagger} \hspace{5mm}\text{with}\hspace{3mm} \|\cm{A}\|_{\ddagger}=\sum_{j=1}^{T_0}\|\bm{A}_j\|_{\Fr},
\end{equation}
where $\cm{\widehat{A}}_{(1)}=(\bm{\widehat{A}}_1,\bm{\widehat{A}}_2,\ldots,\bm{\widehat{A}}_{T_0})$, and $\lambda>0$ is a tuning parameter of penalization.
The true full coefficient tensor $\cm{A}_{\infty}^*$ can be estimated by
$\cm{\widehat{A}}_\infty\in\mathbb{R}^{N\times N\times \infty}$, which appends infinitely many zero matrices to $\cm{\widehat{A}}\in\mathbb{R}^{N\times N\times T_0}$ such that
$	(\cm{\widehat{A}}_\infty)_{(1)}= (\cm{\widehat{A}}_{(1)}, \bm{0}_{N\times N}, \bm{0}_{N\times N}, \dots),
$
i.e. $\bm{\widehat{A}}_j=\bm{0}$ for $j>T_0$. 
Note that $\cm{\widehat{A}}_\infty$ satisfies the low-Tucker-rank constraint at \eqref{eq:trunc_tensor}.

\subsection{Non-asymptotic properties}
We next establish the non-asymptotic properties of $\cm{\widehat{A}}_\infty$, and its accuracy is measured in terms of both parameter estimation and prediction as in the literature \citep{Wainwright19},
\begin{align}
	\begin{split}\label{eq:errors}
		e_{\est}(\cm{\widehat{A}}_\infty) &=\|\cm{A}^*_{\infty}-\cm{\widehat{A}}_{\infty}\|_\Fr^2= \sum_{j=1}^{\infty} \|\bm{A}^*_j-\bm{\widehat{A}}_j \|_{\Fr}^2 = \|\cm{A}^*-\cm{\widehat{A}}\|_\Fr^2+e_{\trunc},\\
		e_{\pred}(\cm{\widehat{A}}_\infty) &= \frac{1}{T_1}\sum_{t=T_0+1}^{T} \Big \| \sum_{j=1}^{\infty}( \bm{A}^*_j-\bm{\widehat{A}}_j) \mathbf{y}_{t-j} \Big \|_2^2 = T_1^{-1}\|(\cm{A}^*-\cm{\widehat{A}})_{(1)}\bm{X}\|_{\Fr} + \widetilde{e}_{\trunc},
	\end{split}
\end{align}
where the truncation errors coming from the sieve approximation are given by
\begin{equation*}\label{eq:truncerrors}
	e_{\trunc}=\sum_{j=T_0+1}^{\infty}\|\bm{A}_j^*\|_{\Fr}^2 \quad\text{and}\quad
	\widetilde{e}_{\trunc}=
	\frac{1}{T_1}\sum_{t=T_0+1}^{T} \Bigg \{ 2 \Big  \langle\sum_{j=1}^{T_0}(\bm{A}^*_j-\bm{\widehat{A}}_j) \mathbf{y}_{t-j}, \bm{r}_t \Big \rangle + \| \bm{r}_t \|_2^2 \Bigg \},
\end{equation*}
respectively, and $\bm{r}_t$ is defined in \eqref{eq:trunc_tensor}. 
We first handle the VAR sieve estimator $\cm{\widehat{A}}$, which corresponds to the first terms on the right hand side of both equations at \eqref{eq:errors}, and the key intermediate step is to prove a more general result, which is known as the oracle inequality in the literature \citep{negahban2012unified}.

%This subsection will derive the estimation and prediction error bounds of
%$\cm{\widehat{A}}_\infty$, and the key intermediate step is to prove a more general result, known as the oracle inequality, for the VAR sieve estimator $\cm{\widehat{A}}$; see \cite{negahban2012unified}. 

%\begin{assumption}[Dependence measure]\label{assum:mu}
%	There exists an absolute constant $\mu>0$ such that $\mu_{\min}(\bm{\Psi}_*)= \min_{|z|=1}\lambda_{\min}(\bm{\Psi}_*^{\HH}(z)\bm{\Psi}_*(z))\geq \mu$, where $\bm{\Psi}_*(z) = \sum_{j=0}^{\infty}\bm{\Psi}_j^*z^j$ for $z \in \mathbb{C}$, and $\bm{\Psi}_*^{\HH}(z)$ is its complex conjugate.
%\end{assumption}
\begin{assumption}[Sub-Gaussian errors]\label{assum:error}
	Let $\bm{\varepsilon}_t = \bm{\Sigma}_\varepsilon^{1/2}\bm{\xi}_t$, where  $\{\bm{\xi}_t\}$ is a sequence of i.i.d. random vectors with zero mean and $\var(\bm{\xi}_t) = \bm{I}_{N}$. 
	In addition, the coordinates $(\bm{\xi}_{it})_{1\leq i\leq N}$ within $\bm{\xi}_t$ are mutually independent and $\sigma^2$-sub-Gaussian, where $\sigma^2>0$ is an absolute constant.
\end{assumption}
\begin{assumption}[Running orders]\label{assum:T0}
	There exists an absolute constant $C>0$ such that $T_1\rho^{T_0/2}\leq C$, or equivalently, 
	$T_0 \geq 2\{\log(1/\rho)\}^{-1}\log (T_1/C)$.
\end{assumption}

The sub-Gaussian condition in Assumption \ref{assum:error} is commonly used in the literature of high-dimensional time series \citep{zheng2019testing,ZC21,wang2023high}.
The running order $T_0$ is required to grow at a rate of $T_0\gtrsim \log(T_1)$ in Assumption \ref{assum:T0} such that the effect of truncated terms $\bm{r}_t$'s can be dominated; see the technical proofs for details.
In fact, to derive the asymptotic normality of low-dimensional VAR sieve estimation \citep{lewis1985prediction}, it is usually assumed that $T_1^{1/2}\sum_{j=T_0+1}^{\infty}\|\bm{A}_j^*\|_{\op} =o(1)$, which exactly corresponds to $T_0\gtrsim \log(T_1)$ since $\sum_{j=T_0+1}^{\infty}\|\bm{A}_j^*\|_{\op} \lesssim \rho^{T_0}$ under Assumption \ref{assum:Adecay}.

%{\color{blue} Different from finite-order VAR methods, VAR sieve estimation requires controlling the truncation errors after lag $T_0$, i.e. $e_{\trunc}$ and $\widetilde{e}_{\trunc}$. From low-rank constraint on $\cm{A}^*$, proofs of theorems in this section shows that the two error terms can be simplified into upper bounding $\sum_{j=T_0+1}^{\infty}\|\bm{A}_j^*\|_{\op}$. Traditional asymptotic analysis directly assumes $\sum_{j=T_0+1}^{\infty}\|\bm{A}_j^*\|_{\op}^2\to 0$ as $T_0, T \to\infty$; see, e.g., \cite{lewis1985prediction}. However, this assumption is not applicable to non-asymptotic framework where $T$ (and hence $T_0$) is treated as finite. Alternatively, since $\sum_{j=T_0+1}^{\infty}\|\bm{A}_j^*\|_{\op}^2 \lesssim \rho^{T_0}$ under Assumption \ref{assum:Adecay}, truncation errors are dominated by estimation error as long as $\rho^{T_0} \lesssim 1/T_1$, giving rise to the lower bound of $T_0$ in terms of $T_1$ in Assumption \ref{assum:T0}.}
%%If there exists some $\epsilon >0$ such that $T_1^{2-\epsilon}>T_0\rho^{T_0}$, then Assumptions \ref{assum:Adecay} and \ref{assum:T0} jointly imply that $\sqrt{T_0}\sum_{j=T_0+1}^{\infty}\|\bm{A}_j^*\|_{\Fr}\to 0$ as $T_1\to \infty$, which is a widely used assumption in low-dimensional VAR sieve approximation; see, e.g., \cite{lewis1985prediction}. 

We next state the oracle inequalities of $\cm{\widehat{A}}$, which will rely on the temporal and cross-sectional dependence of $\{\mathbf{y}_t\}$ \citep{basu2015regularized}.
To this end, we first define
\[
\mu_{\min}(\bm{\Psi}_*)= \min_{|z|=1}\lambda_{\min}(\bm{\Psi}_*^{\HH}(z)\bm{\Psi}_*(z)) \hspace{5mm}\text{and}\hspace{5mm}
\mu_{\max}(\bm{\Psi}_*)= \max_{|z|=1}\lambda_{\max}(\bm{\Psi}_*^{\HH}(z)\bm{\Psi}_*(z)),
\]
where $\bm{\Psi}_*(z) = \sum_{j=0}^{\infty}\bm{\Psi}_j^*z^j$ for $z \in \mathbb{C}$, and $\bm{\Psi}_*^{\HH}(z)$ is its complex conjugate.
Note that, from Assumption \ref{assum:Adecay}, $\mu_{\max}(\bm{\Psi}_*)\leq C(1-\rho^2)^{-1}$ with $C$ being an absolute constant.
Let $\kappa_{\mathrm{RSC}}=\lambda_{\min}(\bm{\Sigma}_\varepsilon)\mu_{\min}(\bm{\Psi}_*)$ and $\kappa_{\mathrm{RSS}}=\lambda_{\max}(\bm{\Sigma}_\varepsilon)\mu_{\max}(\bm{\Psi}_*)$, where $\lambda_{\min}(\bm{\Sigma}_{\varepsilon})$ and $\lambda_{\max}(\bm{\Sigma}_{\varepsilon})$ are the minimum and maximum eigenvalues of $\bm{\Sigma}_{\varepsilon}$, respectively, and $\kappa_{\mathrm{RSC}}$ and $\kappa_{\mathrm{RSS}}$ are key quantities related to the restricted strong convexity and smoothness conditions \citep{agarwal2010fast}.

\begin{theorem}[Group Lasso oracle inequalities]\label{prop:main}
	Let $S\subset \{1,\dots, T_0\}$ be an arbitrary index set with cardinality $|S|=s$ 
	and denote $S^c=\{1,\dots, T_0\} \setminus S$.  	
	Suppose that $\kappa_{\mathrm{RSC}}$ and $\kappa_{\mathrm{RSS}}$ are bounded away from zero and infinity, and Assumptions \ref{assum:glp}--\ref{assum:T0} hold.
	If 
	$T_1\gtrsim \{(r_1\wedge r_2) +s^2\} N+s^2\log T_0$,
	and  
	%\begin{equation*}
	$\lambda \gtrsim \sqrt{\{(r_1\wedge r_2) N+\log T_0 \}/T_1}$,
	%\end{equation*}
	then  with probability at least $1-C e^{-(r_1 \wedge r_2)N-\log T_0}$, 
	\begin{align*}
		\|\cm{\widehat{A}} - \cm{A}^*\|_{\Fr}^2  &\lesssim  \lambda^2s+ \underbrace{ \lambda\|\cm{A}^{*}_{S^c}\|_{\ddagger}+\tau^2\|\cm{A}^{*}_{S^c}\|_{\ddagger}^2 }_{\text{approx error}} \hspace{2mm} \text{and}\hspace{2mm}\\
		T_1^{-1}\|(\cm{\widehat{A}}-\cm{A}^*)_{(1)}\bm{X}\|_{\Fr}^2 &\lesssim \lambda^2s+ \underbrace{ \lambda\|\cm{A}^{*}_{S^c}\|_{\ddagger}+\tau^2\|\cm{A}^{*}_{S^c}\|_{\ddagger}^2 }_{\text{approx error}},
	\end{align*}
	where $\tau^2=C\sqrt{(N+\log T_0)/T_1}$, and $C$ is an absolute constant given in the proof.
\end{theorem}

Following the standard arguments for weak sparsity \citep{raskutti2011minimax,Wainwright19}, the two upper bounds in the above theorem hold for any subset $S$, and each of them consists of two terms: the estimation error, i.e. $\lambda^2 s$, is associated with estimating a total of $s$ unknown $\bm{A}^*_j$'s, and the remaining part of $\cm{A}^*$ that is not estimated, i.e. $\cm{A}^{*}_{S^c}$, gives rise to the approximation error.
The optimal $S$ can be chosen by trading off the estimation and approximation errors.
Moreover, by Assumptions \ref{assum:Adecay} and \ref{assum:T0}, it can be shown that the truncation errors, $e_{\trunc}$ and $\widetilde{e}_{\trunc}$, are dominated by the estimation error $\lambda^2s$. 
Hence,
by balancing the three types of errors and further considering the exponential decay of $\bm{A}_j^*$ as $j\rightarrow\infty$, we can establish the following convergence result.

\begin{theorem}[General linear process]\label{thm:main}
	Suppose that $\kappa_{\mathrm{RSC}}$ and $\kappa_{\mathrm{RSS}}$ are bounded away from zero and infinity, and Assumptions \ref{assum:glp}--\ref{assum:T0} hold.
	If
	\begin{equation}\label{eq:thm1}
		T_1\gtrsim \left [(r_1\wedge r_2) +\left \{\frac{\log T_1}{\log (1/\rho)}\right\}^2  \right] N+ \left \{\frac{\log T_1}{\log (1/\rho)}\right\}^2 \log T_0,
	\end{equation} 
	and $\lambda\asymp \sqrt{\{(r_1\wedge r_2) N+\log T_0 \}/T_1}$, then with probability at least $1-C e^{-(r_1 \wedge r_2)N-\log T_0}$, 
	\[
	e_{\est}(\cm{\widehat{A}}_\infty)  \lesssim  \frac{\{(r_1\wedge r_2) N+\log T_0\}\log T_1 }{T_1 \log(1/\rho)}
	\hspace{2mm} \text{and}\hspace{2mm}
	e_{\pred}(\cm{\widehat{A}}_\infty) \lesssim  \frac{\{(r_1\wedge r_2) N+\log T_0\}\log T_1 }{T_1 \log(1/\rho)},
	\]
	where $C$ is an absolute constant given in the proof.
\end{theorem}

The optimal choice of $S$ has cardinality $s\lesssim \log(\sqrt{N}/\lambda)/\log(1/\rho)$, which, together with the rate of $\lambda$ in Theorem \ref{thm:main}, implies that $s\lesssim \log(T_1)/\log(1/\rho)$, i.e. the number of active lags decreases as $\rho$ decreases.
This is natural since a smaller $\rho$ corresponds to a faster decay rate of $\bm{A}_j^*$ and hence a greater level of group sparsity of $\cm{A}^*$; i.e. there is a smaller cutoff $Q$  such that all $\bm{A}_j^*$'s with $j\geq Q$ are very close to zero matrices.
Secondly, the term $\log T_0$ appears in the above theorems, and this is due to the Lasso regularization on $T_0$ groups of coefficients, $\bm{A}_1, \dots, \bm{A}_{T_0}$.
Thirdly, the upper bound on $T_0$, namely $\log T_0 \lesssim T_1/(\log T_1)^2$, is looser than $T_0 = o(T^{1/3})$, which is necessary for the low-dimensional VAR sieve estimation  \citep{lewis1985prediction}. It is mainly due to the group Lasso penalty, and this makes it possible to consider a larger $T_0$ in real applications; see Section \ref{sec:experiment1} for numerical evidences. 
Finally, the above two theorems can be easily adjusted when the two quantities, $\kappa_{\mathrm{RSC}}$ and $\kappa_{\mathrm{RSS}}$, depend on $N$, $T_0$ and $T_1$, while the assumption that they are bounded away from zero and infinity can simplify the discussions on the three types of errors, as well as the running order selection.

%\subsection{Application to VAR processes}

Note that Theorem \ref{thm:main} gives error bounds uniformly for all general linear processes, while some VAR models may have finite orders and the coefficient matrices at some lags are even exactly zero; see, e.g., the commonly used seasonal VAR models in real applications \citep{cryer1986time}.
For this case, truncation errors disappear, and we even do not need to handle approximation errors.
Specifically, consider an VAR$(T_0)$ model with its coefficient matrices satisfying the low-Tucker-rank assumption at \eqref{eq:trunc_tensor}, i.e. $\bm{r}_t=0$ and $\widetilde{\bm{\varepsilon}}_t=\bm{\varepsilon}_t$, and Assumption \ref{assum:T0} is no longer needed.
Moreover, the following stationarity condition, rather than Assumption \ref{assum:glp}, is required.

%It is not necessary to consider the truncation error, and even approximation error, for this case.
%
%The running order $T_0$ is required to diverge at a certain rate in Assumption \ref{assum:T0}, and it is mainly used to control effects of the truncated part $\bm{r}_t$ in technical proofs.
%However, for VAR processes, it holds that $\bm{r}_t=0$ once $T_0$ is greater than the true order, and this constraint can be removed. 
%Specifically, consider an VAR$(T_0)$ model with its coefficient matrices satisfying the low-Tucker-rank condition at \eqref{eq:trunc_tensor}, i.e. $\bm{r}_t=0$ and $\widetilde{\bm{\varepsilon}}_t=\bm{\varepsilon}_t$.
%Denote its matrix polynomial by $\mathcal{A}(z) = \bm{I} - \sum_{j=1}^{T_0}\bm{A}_j z^j$.
\begin{assumption}[Stationarity condition]\label{assum:stationarity}
	The determinant of $\bm{A}(z)$ is not equal to zero for all $z\in\mathbb{C}$ and $|z|<1$, where $\bm{A}(z) = \bm{I} - \sum_{j=1}^{T_0}\bm{A}_j z^j$ is the VAR matrix polynomial.
\end{assumption} 
\begin{corollary}[Finite-order VAR process]\label{corollary:statistical}
	For an VAR$(T_0)$ process, if Assumptions \ref{assum:Adecay}, \ref{assum:error} and \ref{assum:stationarity} are satisfied, then Theorem \ref{prop:main} still holds. 
\end{corollary}

For a group-sparse VAR process, i.e. $\cm{A}^{*}_{S^c}=0$, the approximation error becomes zero, and the error bounds then have a rate of $s\{(r_1\wedge r_2) N+\log T_0 \}/T_1$, which is sharper than that in Theorem \ref{thm:main} when $s\leq T_0$ is fixed or has a much slower rate than $\log(T_1)$.

%\begin{corollary}[Finite-order VAR process]\label{corollary:statistical}
%	Suppose that $\kappa_{\mathrm{RSC}}, \kappa_{\mathrm{RSS}}$ and $\sigma^2$ are bounded away from zero and infinity, and Assumptions \ref{assum:Adecay}, \ref{assum:error} and \ref{assum:stationarity} hold. If
%	$T_1\gtrsim \{(r_1\wedge r_2) +s^2\} N+s^2\log T_0$,
%	and 
%	$\lambda \gtrsim \sqrt{\{(r_1\wedge r_2) N+\log T_0 \}/T_1}$,
%	then   
%	\[
%	\|\cm{\widehat{A}} - \cm{A}^*\|_{\Fr}^2  \lesssim  \lambda^2s+ \lambda\|\cm{A}^{*}_{S^c}\|_{\ddagger}+\tau^2\|\cm{A}^{*}_{S^c}\|_{\ddagger}^2,
%	\]
%	with probability at least $1-C e^{-(r_1 \wedge r_2)N-\log T_0}$,
%	where $\tau^2$ is defined in Theorem \ref{prop:main}, and $C$ is an absolute constant given in the proof.
%\end{corollary}

\section{Algorithm for high-dimensional estimation}
\subsection{Alternating gradient descent algorithm with hard-thresholding}

The rank-constrained group Lasso estimation in Section 2 involves both the low-rank constraint and sparsity, which makes the parameter estimation difficult.
This section introduces an alternating algorithm to search for the estimate with both low-rankness and group sparsity, and it consists of two key steps: gradient descent and hard-thresholding. 
Note that the resulting estimator is different from $\cm{\widehat{A}}$ in Section \ref{sec:lass}, and hence this section also discusses its theoretical properties, including both statistical and convergence analysis.

The first key step of the proposed algorithm is gradient descent, and the challenging part is to deal with the low-rank constraint.
Note that, for any $\cm{A}\in \bm{\Theta}(r_1, r_2)$, there exists a Tucker decomposition $\cm{A}=\cm{G}\times_1 \bm{U}_1\times_2 \bm{U}_2$ with $\cm{G}\in \mathbb{R}^{r_1 \times r_2\times T_0}$, $\bm{U}_1\in\mathbb{R}^{N\times r_1}$ and $\bm{U}_2\in\mathbb{R}^{N\times r_2}$.
As a result, the loss function in Section 2.2 can be rewritten as $\mathcal{L}(\cm{G},\bm{U}_1,\bm{U}_2):=\mathcal{L}(\cm{G}\times_1 \bm{U}_1\times_2 \bm{U}_2)$, and we further adjust it by adding two regularization terms,
\begin{equation*}\label{eq:hard-estimation}
	\mathcal{L}^{\mathrm{GD}}(\cm{G},\bm{U}_1,\bm{U}_2)= \mathcal{L}(\cm{G},\bm{U}_1,\bm{U}_2) + \frac{a}{2}(\|\bm{U}_1^\prime \bm{U}_1 - b^2\bm{I}_{r_1}\|_{\Fr}^2 +\|\bm{U}_2^\prime \bm{U}_2 - b^2\bm{I}_{r_2}\|_{\Fr}^2 ),
\end{equation*}
where $a, b >0$ are two tuning parameters.
The above method is motivated by \cite{HWZ21} for low-rank tensor estimation, and the regularization terms, $\|\bm{U}_1^\prime \bm{U}_1 - b^2\bm{I}_{r_1}\|_{\Fr}^2$ and $\|\bm{U}_2^\prime \bm{U}_2 - b^2\bm{I}_{r_2}\|_{\Fr}^2$, are used to keep $\bm{U}_1$ and $\bm{U}_2$ from being singular, and meanwhile they can also balance the scaling of $\cm{G}$, $\bm{U}_1$ and $\bm{U}_2$.

Denote by $\cm{\widetilde{G}}$, $\widetilde{\bm{U}}_1$ and $\widetilde{\bm{U}}_2$ the minimizers of $\mathcal{L}^{\mathrm{GD}}(\cm{G},\bm{U}_1,\bm{U}_2)$, and it then holds that $\bm{\widetilde{U}}_i^\prime\bm{\widetilde{U}}_i = b^2\bm{I}_{r_i}$, i.e. $b^{-1}\bm{\widetilde{U}}_i$ are orthonormal, for $i = 1$ and 2.
In fact, if this is not true, then there exist invertible matrices $\bm{O}_i\in\mathbb{R}^{r_i\times r_i}$ such that $\bm{\widetilde{U}}_i = \bm{\bar{U}}_i\bm{O}_i$ and $\bm{\bar{U}}_i^\prime\bm{\bar{U}}_i = b^2\bm{I}_{r_i}$, where $1
\leq i\leq 2$. Let $\cm{\widebar{A}} = (\cm{\widetilde{G}}\times_1\bm{O}_1\times_2\bm{O}_2) \times_1 \bm{\widebar{U}}_1\times_2\bm{\widebar{U}}_2$. Then $\mathcal{L}(\cm{\widebar{A}}) = \mathcal{L}(\cm{\widetilde{A}})$ while the regularization terms for $\bm{\widebar{U}}_1$ and $\bm{\widebar{U}}_2$ are reduced to zero. This leads to a contradiction with the definition of minimizers.
Moreover, the proposed algorithm is not sensitive to the choice of regularization parameters, $a$ and $b$, in $\mathcal{L}^{\mathrm{GD}}(\cm{G},\bm{U}_1,\bm{U}_2)$, and they are set to one in all our simulation and empirical studies.
Similar regularization methods have been widely applied to non-convex low-rank matrix estimation problems; see \cite{tu2016low}, \cite{wang2017unified} and references therein.

We now consider the gradient descent. First, the partial derivatives can be calculated as
\[
\nabla_{\mathcal{G}}\mathcal{L}^{\mathrm{GD}}=\nabla_{\mathcal{G}}\mathcal{L}(\cm{A})
\hspace{3mm}\text{and}\hspace{3mm}
\nabla_{{U}_i}\mathcal{L}^{\mathrm{GD}}= \nabla_{{U}_i}\mathcal{L}(\cm{A}) + a \bm{U}_i(\bm{U}_i^{\prime} \bm{U}_i - b^2\bm{I}_{r_i}) \hspace{3mm}\text{with $i=1$ and 2},
\]
where $\nabla\mathcal{L}(\cm{A})$, $\nabla_{{U}_1}\mathcal{L}(\cm{A}), \nabla_{{U}_2}\mathcal{L}(\cm{A})$ and $\nabla_{\mathcal{G}}\mathcal{L}(\cm{A})$ are the first order partial derivatives of $\mathcal{L}(\cm{A})$ with respect to $\cm{A}$, $\bm{U}_1, \bm{U}_2$ and $\cm{G}$, respectively; see Algorithm \ref{alg:AGD-HT} for details.

\begin{algorithm}[t] 
	\caption{Alternating gradient descent algorithm with hard-thresholding}
	\label{alg:AGD-HT}
	\textbf{Input:} Running order $T_0$, ranks $(r_1,r_2)$, sparsity parameter $s$, initialization $\cm{G}^0, \bm{U}_1^{0},  \bm{U}_2^{0}$, regularization parameters $a, b > 0$ and step size $\eta>0$. \\
	\textbf{For} $k=0,1,2,\dots,K-1$:\\\vspace{2mm}
	\hspace*{5mm}$ \bm{U}_1^{k+1} \leftarrow \bm{U}_1^{k} - {\eta}\left[ \nabla_{{U}_1}\mathcal{L}(\cm{A}^k) + a \bm{U}_1^k(\bm{U}_1^{k\prime} \bm{U}_1^k - b^2\bm{I}_{r_1}) \right]	$\\\vspace{2mm}
	\hspace*{5mm}$ \bm{U}_2^{k+1} \leftarrow \bm{U}_2^{k} - {\eta}\left[ \nabla_{{U}_2}\mathcal{L}(\cm{A}^k) + a \bm{U}_2^k(\bm{U}_2^{k\prime} \bm{U}_2^k - b^2\bm{I}_{r_2}) \right]	$\\\vspace{2mm}
	\hspace*{5mm}$ \cm{\widetilde{G}}^{k+1} \leftarrow \cm{G}^{k} - {\eta}\nabla_{\mathcal{G}}\mathcal{L}(\cm{A}^k)	$\\\vspace{2mm}
	\hspace*{5mm}$\cm{\widetilde{A}}^{k+1} = \cm{\widetilde{G}}^{k+1}\times_1\bm{U}_1^{k+1}\times_2\bm{U}_2^{k+1}$\\\vspace{2mm}
	\hspace*{5mm}$\cm{A}^{k+1} = \cm{G}^{k+1}\times_1\bm{U}_1^{k+1}\times_2\bm{U}_2^{k+1}\leftarrow \hardt{\cm{\widetilde{A}}^{k+1},s}$\\\vspace{2mm}
	\textbf{end for}\\\vspace{2mm}
	\textbf{return} \hspace*{2mm}$\cm{A}^{K} = \cm{G}^{K}\times_1\bm{U}_1^{K}\times_2\bm{U}_2^{K}$
\end{algorithm}

%For the second key step of Algorithm \ref{alg:AGD-HT}, we consider a hard-thresholding method to handle group Lasso in this paper, although the Lasso is a soft-thresholding problem \citep{Wainwright19}, and this is due to the three reasons below.
The second key step of Algorithm \ref{alg:AGD-HT} is hard-thresholding, and it aims at the group sparsity. We consider a hard-thresholding method in this paper although the Lasso is a soft-thresholding problem \citep{Wainwright19}, and this is due to the three reasons below.
First, the VAR sieve approximation method heavily depends on the running order $T_0$, and we may need to try different values of $T_0$ for a better performance in real applications.
In the meanwhile, as $T_0$ increases, the model complexity will increase, while the number of effective samples $T_1 = T - T_0$ will decrease. 
As a result, the numerically selected tunning parameter $\lambda$ with different $T_0$ may differ from each other dramatically, while the hard-thresholding is less sensitive to such type of changes; see simulation experiments in Section 4.2 for more evidences. 
Secondly, although the Lasso problem is usually convex, the loss function at \eqref{eq:minimizer} is still nonconvex due to the low-rankness. It is hence not necessary to insist on the Lasso method to induce the group sparsity of $\bm{A}_j$'s.
Finally, the hard-thresholding method is always faster than many common convex algorithms in the literature, including Lasso-type ones, by orders of magnitude, and its convergence analysis is also relatively easier to establish; see \cite{tropp2010computational,shen2017tight}. 

We next define the hard-thresholding operation. Consider a coefficient tensor $\cm{B}\in\mathbb{R}^{d_1\times d_2\times T_0}$ with $\cm{B}_{(1)}=(\bm{B}_1,\bm{B}_2,\ldots,\bm{B}_{T_0})$, where for each $1\leq j\leq T_0$ the coefficient matrix $\bm{B}_j\in\mathbb{R}^{d_1\times d_2}$ is the $j$-th frontal slice. Denote by $\hardt{\cm{B}, s}\in\mathbb{R}^{d_1\times d_2\times T_0}$ the hard-thresholding operator, which keeps the $s$ largest coefficient matrices in terms of $\|\bm{B}_j\|_{\Fr}$'s and suppresses the rest to zero.
Define a parameter space with both low-rankness and sparsity below,
\begin{equation*}\label{eq:Theta_s}
	\bm{\Theta}^{\mathrm{SP}}(r_1,r_2,s)= \{\cm{A}\in\mathbb{R}^{N\times N\times T_0}\mid \|\cm{A}\|_0 \leq s,\hspace{2mm} \rank(\cm{A}_{(1)}) \leq r_1, \hspace{2mm}\rank(\cm{A}_{(2)}) \leq r_2\},
\end{equation*}
where $\|\cm{A}\|_0=\sum_{j=1}^{T_0}I(\|\bm{A}_j\|_{\Fr}>0)$ is the number of active coefficient matrices.
Note that, for any $1\leq s\leq T_0$, $\bm{\Theta}^{\mathrm{SP}}(r_1,r_2,s)\subseteq \bm{\Theta}(r_1,r_2)$, and the hard-thresholding operator $\hardt{\cm{A}, s}$ is a projection from parameter spaces $\bm{\Theta}(r_1,r_2)$ into $\bm{\Theta}^{\mathrm{SP}}(r_1,r_2,s)$.
For a tensor $\cm{A} =\cm{G}\times_1\bm{U}_1\times_2\bm{U}_2$, it can be verified that $\hardt{\cm{A}, s}=\cm{\widetilde{G}}\times_1\bm{U}_1\times_2\bm{U}_2$, and $\cm{\widetilde{G}}=\hardt{\cm{G}, s}$ if $\bm{U}_1$ and $\bm{U}_2$ are assumed to have orthonormal columns.
The hard-thresholding operator $\hardt{\cm{A}, s}$ is used in Algorithm \ref{alg:AGD-HT} to project a non-group-sparse coefficient tensor into a group-sparse one.

In the meanwhile, we may also consider a soft-thresholding method since it is more relevant to Lasso problems \citep{hastie2015statistical}. Specifically, for a tuning parameter $\lambda>0$, define the soft-thresholding operator as $\cm{\widetilde{B}} = \text{ST}(\cm{B},\lambda)\in\mathbb{R}^{d_1\times d_2\times T_0}$ with  $\bm{\widetilde{B}}_j = (1-\lambda/\|\bm{B}_j\|_{\Fr})_{+}\bm{B}_j$ for $1\leq j\leq T_0$ and $\cm{\widetilde{B}}_{(1)}=(\bm{\widetilde{B}}_1,\bm{\widetilde{B}}_2,\ldots,\bm{\widetilde{B}}_{T_0})$, where $\bm{B}_j$ and $\bm{\widetilde{B}}_j$ are the $j$-th frontal slices of $\cm{B}$ and $\cm{\widetilde{B}}$, respectively, and the function $(x)_{+}=x$ for $x>0$ and zero otherwise. However, as shown by simulations in Section 4.2, the numerically selected $\lambda$ may change dramatically for slightly different running orders $T_0$ and sample sizes $T$, making the soft-thresholding method less attractive.

The proposed alternating gradient descent algorithm is detailed in Algorithm \ref{alg:AGD-HT}.
Specifically, at the $k$-th step, the three components, $\bm{U}_1^{k+1}, \bm{U}_2^{k+1}$ and $\cm{\widetilde{G}}^{k+1}$, are first updated by gradient descent separately, while the resulting estimator $\cm{\widetilde{A}}^{k+1} = \cm{\widetilde{G}}^{k+1}\times_1\bm{U}_1^{k+1}\times_2\bm{U}_2^{k+1}$ is a non-group-sparse coefficient tensor. We then project it into the group-sparse parameter space $\bm{\Theta}^{\mathrm{SP}}(r_1,r_2,s)$ by $\hardt{\cm{\widetilde{A}}^{k+1},s}=\cm{G}^{k+1}\times_1\bm{U}_1^{k+1}\times_2\bm{U}_2^{k+1}$, denoted by $\cm{A}^{k+1}$.
The above two steps are repeated $K$ times, after which we can obtain the estimate $\cm{A}^{K}$.
In Algorithm \ref{alg:AGD-HT}, we may conduct the hard-thresholding on $\cm{\widetilde{G}}^{k+1}$ directly, and a similar performance can be observed. However, it is more convenient to establish the corresponding convergence analysis for the hard-thresholding on $\cm{\widetilde{A}}^{k+1}$.

We next consider initialization of the algorithm.
First, the running order $T_0$ plays a key role, and we need to try many different values in real applications with the help of some statistical tools such as the sample autocorrelation function \citep{cryer1986time,Tsay14}.
Secondly, since the true model is most likely an VAR$(\infty)$ process, the high-dimensional Akaike information criterion (AIC) is more suitable to select the low ranks $r_1$ and $r_2$ and sparsity $s$,
\begin{equation*}\label{eq:AIC}
	\text{AIC}(r_1,r_2,s) = \log\left\{(2T_1)^{-1}\|\bm{Y} - \cm{\widetilde{A}}_{(1)}\bm{X}\|_{\Fr}^2\right\} + c {[(r_1+ r_2)N+\log T_0]s}/{T_1},
\end{equation*}
where $c>0$ is a tunning parameter; see \cite{bai2018strong}.
Thirdly, the $r_1\times r_2\times T_0$ tensor $\cm{G}^0$ can be set to zero, while $\bm{U}_1^0$ and $\bm{U}_2^0$ are initialized by some orthonormal matrices of sizes $N \times r_1$ and $N \times r_2$, respectively. Moreover, to provide a warm-start initialization in practice, we may skip the hard-thresholding operation for the first few iterations.

\subsection{Theoretical properties}

%This subsection provides theoretical justifications, including both statistical and convergence analysis, for Algorithm \ref{alg:AGD-HT}.

Consider the true coefficient tensor $\cm{A}^*$ with $\cm{A}_{(1)}^*=(\bm{A}_1^*,\bm{A}_2^*,\ldots,\bm{A}_{T_0}^*)\in\mathbb{R}^{N\times NT_0}$ in Section 2 and the parameter space $\bm{\Theta}^\mathrm{SP}(r_1,r_2,s)$ with low-rankness and group-sparsity in Section 3.1.
For a given threshold $\gamma>0$, let the active set $S_\gamma = \left \{j\in\{1, \dots, T_0\}\mid\| \bm{A}_j^*\|_\text{F} > \gamma \right \}$ with cardinality $s_{\gamma}=|S_{\gamma}|$, and $S_\gamma^c=\{1,\dots,T_0\}\setminus S_\gamma$. 
We then define a random quantity
\begin{equation*}
	e_{\gamma}(r_1,r_2,s) := \sup_{\cmt{M}\in\bm{\Theta}^{\mathrm{SP}}(r_1,r_2,s),\; \|\cmt{M}\|_{\Fr}=1}\langle \nabla\mathcal{L}(\cm{A}_{S_{\gamma}}^*) ,\cm{M}  \rangle.
\end{equation*}
This quantity is directly related to statistical errors of the proposed algorithm, and we first analyze it statistically by providing error bounds as in Section 2.3.

\begin{theorem}[Statistical analysis] \label{thm:stat}
	Suppose that $\kappa_{\mathrm{RSC}}$ and $\kappa_{\mathrm{RSS}}$ are bounded away from zero and infinity, and Assumptions \ref{assum:glp}--\ref{assum:T0} hold.
	For any $\gamma \gtrsim \sqrt{\{(r_1\wedge r_2)N+\log T_0\}/{T_1}}$, if 
	$T_1\gtrsim \{(r_1\wedge r_2) +s^2\} N+s^2\log T_0$,
	then  with probability at least $1-C e^{-(r_1\wedge r_2)N-\log T_0}$,
	\begin{align*}
		e_{\gamma}^2(r_1,r_2,s) \lesssim \gamma^2 s + \|\cm{A}^*_{S_\gamma^c}\|_{\Fr}^2 + \tau^2\|\cm{A}^*_{S_\gamma^c}\|_{\ddagger}^2,
	\end{align*}
	where $\tau^2=C\sqrt{(N+\log T_0)/T_1}$, and $C$ is an absolute constant given in the proof.
\end{theorem}

The upper bound in the above theorem has a form similar to that in Theorem \ref{prop:main}, and it increases as $\gamma$ increases. As a result, we can obtain the statistical error bound below,
\begin{equation}\label{eq:statistic-error}
	T_1^{-1}[(r_1\wedge r_2)N+\log T_0] [s+\log(T_1)/\log(1/\rho)]
\end{equation}
by choosing $\gamma \asymp\sqrt{\{(r_1\wedge r_2)N+\log T_0\}/{T_1}} $, and it can also be verified that $s_{\gamma} \lesssim \log T_1/\log(1/\rho)$.

We next conduct convergence analysis. To this end, for input $\cm{A}^{k}$ at the $k$-th iteration, denote its active set by $S_{k} = \{j\in\{1,\dots,T_0\} \mid \|\bm{A}_j^{k}\|_{\Fr} \neq 0\}$, and let $\nu_k = |S_k \cup S_{k+1}|/s-1$.
Note that it corresponds to the case with $S_k= S_{k+1}$ if $\nu_k=0$ and that with $S_k\cap S_{k+1}=\emptyset$ if  $\nu_k=1$, i.e. $\nu_k \in [0, 1]$ can be used to measure the size of overlap between $S_k$ and $S_{k+1}$.
Let 
$\nu_{\mathrm{min}} = \min_{0\leq k\leq K-1}\nu_k$, and $\nu_{\mathrm{max}} = \max_{0\leq k\leq K-1}\nu_k$.
Moreover, denote
$\sigma_L = \min  \{ \sigma_{\min}[(\cm{A}_{S_{\gamma}}^*)_{(1)}], \sigma_{\min}[(\cm{A}_{S_{\gamma}}^*)_{(2)}]\}$,
$\sigma_U = \max  \{ \sigma_{\max}[(\cm{A}_{S_{\gamma}}^*)_{(1)}], \sigma_{\max}[(\cm{A}_{S_{\gamma}}^*)_{(2)}]\}$,
and $\kappa=\sigma_U/\sigma_L$, where their dependence on $\gamma$ is suppressed without confusion.

\begin{theorem}[Convergence analysis] \label{thm:optimization}
	Consider Algorithm \ref{alg:AGD-HT} with step size $\eta = \eta_0(\kappa_{\mathrm{RSC}} + \kappa_{\mathrm{RSS}})^{-1}[(1+\sigma_U)(1+\sigma_U^{1/2})]^{-1}$, where $\eta_0\leq \min\{150^{-1},204\sigma_U^{-1}\}$ is a positive constant, and denote by $e_{\mathrm{stat}}^2=e_{\gamma}^2(r_1,r_2,3s)$ the statistical error.
	For a given $\gamma$, suppose that $b \asymp \sigma_U^{1/4}$, $ a \asymp (\kappa_{\mathrm{RSC}}^{-1} + \kappa_{\mathrm{RSS}}^{-1})^{-1}(\sigma_U^{1/2}+\sigma_U)$, $\|\cm{A}^0 - \cm{A}^*_{S_{\gamma}}\|^2_{\Fr} \lesssim \sigma_L^{5/2}\kappa^{-3/2}$, 
	$\nu_{\mathrm{max}}\lesssim \eta_0\kappa_{\mathrm{RSC}}^{2}\kappa_{\mathrm{RSS}}^{-2}\kappa^{-4}$, $s \geq \nu_{\mathrm{min}}^{-1}s_{\gamma}$, and
	$e_{\mathrm{stat}}^2\lesssim \eta_0^2\kappa_{\mathrm{RSC}}^{4}\kappa_{\mathrm{RSS}}^{-4}\kappa^{-8}$.	
	If Assumptions \ref{assum:glp}--\ref{assum:T0} hold, and $T_1 \gtrsim s^2(N+\log T_0)$, 
	\begin{align}\label{eq:iteration}
		\|\cm{A}^{K} - \cm{A}^*_{S_\gamma}\|_{\Fr}^2 \lesssim \kappa^{3/2}\sigma_L^{-1/2} \left(1 - \eta_0^2 \delta^2 \right)^K\|\cm{A}^0 - \cm{A}^*_{S_\gamma}\|_{\Fr}^2 +  \kappa^{7/2}\sigma_L^{-1/2}\kappa_{\mathrm{RSC}}^{-2} \eta_0^{-2} \delta^{-2} e_{\mathrm{stat}}^2
	\end{align}
	holds, with probability at least $1-C e^{-N-\log T_0}$,
	where $\delta=1088^{-1}\kappa_{\mathrm{RSC}}\kappa_{\mathrm{RSS}}^{-1}\kappa^{-2}$, $\eta_0\delta<1$, and $C$ is an absolute constant given in the proof.
\end{theorem}

The two terms at the right hand side of \eqref{eq:iteration} correspond to the optimization and statistical errors, respectively. The statistical error is discussed at Theorem \ref{thm:stat} and, since $\eta_0\delta<1$, the linear convergence rate can be implied for the optimization error.
Secondly, by triangle inequality, $0.5\|\cm{A}^{K} - \cm{A}^*\|_{\Fr}^2\leq \|\cm{A}^{K} - \cm{A}^*_{S_\gamma}\|_{\Fr}^2+\|\cm{A}^*_{S_\gamma^c}\|_{\Fr}^2$, and hence the convergence analysis at Theorem \ref{thm:optimization} can be readily extended to include approximation errors.
Thirdly, if we further assume that $\kappa_{\mathrm{RSS}}$, $\kappa_{\mathrm{RSC}}$, $\sigma_U$ and $\sigma_L$ are bounded away from zero and infinity as in all the other theorems, the tuning parameters $a$, $b$ and $\eta$ in Algorithm \ref{alg:AGD-HT} can then be at a constant level, i.e. they do not depend on $N$, $T_0$ or $T_1$. 
Finally, it is required by Theorem \ref{thm:optimization} that $s\geq s_{\gamma}$ and, from \eqref{eq:statistic-error}, we then can choose $s \asymp \log T_1/\log(1/\rho)$. This hence leads to the following results for general linear processes.

\begin{corollary}[General linear process with hard-thresholding]\label{cor:algorithm}
	Suppose that the conditions of Theorems \ref{thm:stat} and \ref{thm:optimization} hold, and we choose $s \asymp \log T_1/\log(1/\rho)$ and $\gamma \asymp\sqrt{\{(r_1\wedge r_2)N+\log T_0\}/{T_1}}$ in Algorithm \ref{alg:AGD-HT}. After $K$-th iteration with
	\begin{equation*}
		K \gtrsim \frac{\log(\kappa^{7/2}\sigma_L^{-5/2}\kappa_{\mathrm{RSC}}^{-2}\delta^{-2})}{\log(1-\eta_0^2\delta^2)},
	\end{equation*} 
	and $\eta_0$ and $\delta$ given in Theorem \ref{thm:optimization}, if $T_1\gtrsim \{(r_1\wedge r_2) +s^2\} N+s^2\log T_0$, it then holds that
	\begin{equation*}
		\|\cm{A}^{K} - \cm{A}^*\|_{\Fr}^2 \lesssim \frac{[(r_1\wedge r_2)N+\log T_0]s}{T_1}
	\end{equation*}
	with probability at least $1-C e^{-(r_1\wedge r_2)N-\log T_0}$, where $C$ is an absolute constant given in the proof.
\end{corollary}

The above corollary gives the same bound as that in Theorem \ref{thm:main} since $s \asymp \log T_1/\log(1/\rho)$.
Moreover, when the quantities of $\kappa_{\mathrm{RSC}}, \kappa_{\mathrm{RSS}}, \kappa$ and $\sigma_L$ are bounded away from zero
and infinity, the required number of iterations does not depend on $N$, $T_0$ or $T_1$, and this makes sure that the proposed algorithm can be applied to large
datasets without any difficulty. 
Finally, as in Section 2.3, for group-sparse VAR$(T_0)$ processes, i.e. $\cm{A}^*_{S_\gamma^c}=0$ for some $\gamma>0$, we can obtain the same bound as that in Corollary \ref{cor:algorithm}, while $s$ may be fixed or have a much slower rate than $\log(T_1)$.

\section{Simulation studies}
\subsection{Finite-sample performance of the proposed methodology}\label{sec:experiment1} 

This subsection conducts two simulation experiments to evaluate the finite-sample performance of VAR sieve estimators $\cm{A}^{K}$ from Algorithm \ref{alg:AGD-HT} in Section 3.

The first experiment is to evaluate how the three types of errors, i.e. the estimation, approximation and truncation errors, can be balanced numerically when the sparsity level $s$ varies. We consider two data generating processes below,
\begin{align*}
	\mathrm{VARMA:}\;& \mathbf{y}_t=\bm{\Phi}\mathbf{y}_{t-1}+\bm{\varepsilon}_t-\bm{\Theta}\bm{\varepsilon}_{t-1}\text{ with }\bm{\Phi}=-0.5\bm{B}\bm{J}\bm{B}^\prime\text{ and }\bm{\Theta}=0.7\bm{B}\bm{J}\bm{B}^\prime,\\
	\mathrm{VAR:}\;&
	\mathbf{y}_t = \bm{\Phi}_1\mathbf{y}_{t-1}+\bm{\Phi}_4\mathbf{y}_{t-4}+\bm{\Phi}_5\mathbf{y}_{t-5}+\bm{\Phi}_8\mathbf{y}_{t-8}+\bm{\Phi}_9\mathbf{y}_{t-9}+\bm{\varepsilon}_t\\
	&\hspace{41mm} \text{ with } \bm{\Phi}_j = c_j0.7^j\bm{B}\bm{J}\bm{B}^\prime,
\end{align*}
where $\mathbf{y}_t\in\mathbb{R}^N$, $\bm{B}\in\mathbb{R}^{N\times N}$ is an orthogonal matrix, $\bm{J}=\diag\{\bm{1}_r, \bm{0}_{N-r}\}$ with $\bm{1}_r=(1,\dots,1)^\prime$ is an $N$-dimensional diagonal matrix with rank $r$, $\{\bm{\varepsilon}_t\}$ are $i.i.d.$ $N$-dimensional standard normal random vectors, and $(c_1,c_4,c_5,c_8,c_9)=(1,2,-2,-1,1)$.
The VARMA model is stationary since the spectral radius of $\bm{\Phi}$ is $0.5$, and it has a weakly group-sparse VAR($\infty$) form at \eqref{eq:VAR_inf} with $\bm{A}_j^* = \bm{\Theta}^{j-1}(\bm{\Phi} - \bm{\Theta}) = -1.2\times 0.7^{j-1}\bm{B}\bm{J}\bm{B}^\prime$ for all $j\geq 1$.
The VAR model also satisfies the stationarity condition at Assumption \ref{assum:stationarity}, i.e. $|\det(\bm{A}(z))| = (1/z-0.7)(1/z^4-0.7^4)^2 \neq 0$ for all $z\in\mathbb{C}$ and $|z|\leq 1$, and it is group-sparse. 
In fact, the VAR process has a form of seasonal VAR(1)$\times$VAR(2)$_4$ models, $(\bm{I}_N-\bm{\Pi}_1L)(\bm{I}_N-\bm{\Pi}_2L^4-\bm{\Pi}_3L^8)\mathbf{y}_t = \bm{\varepsilon}_t$, where $\bm{\Pi}_1=0.7\bm{B}\bm{J}\bm{B}^\prime, \bm{\Pi}_2 = 2\times 0.7^4\bm{B}\bm{J}\bm{B}^\prime, \bm{\Pi}_3 = -0.7^8\bm{B}\bm{J}\bm{B}^\prime$, and $L$ is the backshift operator.
Note that the VAR coefficient matrices $\{\bm{A}_j^*,j\geq 1\}$ of both processes share the same row and column spaces, spanned by first $r$ columns of $\bm{B}$, and the corresponding supervised factor models have the rank of $r_1=r_2=r$.

We fix the settings at $(N,r,T)=(20,4,1500)$, and there are 500 replications with the orthogonal matrix $\bm{B}$ being generated randomly at each replication for better evaluation. Algorithm \ref{alg:AGD-HT} is applied to each generated sequence with $T_0=\floor{1.5\sqrt{T}}=58$ and $s$ varying from 3 to 35, where $\floor{\cdot}$ is the floor function, and we can obtain the output $\cm{A}^{K}$ until the algorithm converges. 
The estimation, approximation and truncation errors refer to $\|\cm{A}^{K} - \cm{A}_{S}^*\|_{\Fr}^2$, $\|\cm{A}_{S^c}^*\|_{\Fr}^2$ and $\sum_{j=T_0+1}^{\infty}\|\bm{A}_j^*\|_{\Fr}^2$, respectively, and the parameter estimation error is defined as $\|\cm{A}^{K} - \cm{A}^*\|_{\Fr}^2=\|\cm{A}^{K} - \cm{A}_{S}^*\|_{\Fr}^2+\|\cm{A}_{S^c}^*\|_{\Fr}^2$, where $S$ contains the indices of all estimated active coefficient matrices.
The truncation error is $1.32\times 10^{-17}$ for the VARMA process and zero for the VAR process, and hence they can be ignored numerically comparing with the other two types of errors.
Figure \ref{fig:rate_s} gives the estimation and approximation errors, averaged over 500 replications, and we have three findings below. 
First, as the sparsity level $s$ increases, linear growth in the estimation errors can be roughly observed for both the VARMA and VAR processes. The approximation error decreases quickly for both processes and, when $s>5$, it becomes almost zero for the VAR process since the active set can be correctly selected for most replications.
Secondly, for the VARMA process, the approximation error is dominating for the cases with $s<10$, while the estimation error has much larger values when $s>10$. As a result, as $s$ increases, the parameter estimation error  decreases first and then increases when the sparsity level $s>10$. This phenomenon can also be observed for the VAR process, and the parameter estimation error reaches the minimum at $s=5$, which is the true sparsity level.
Finally, for large sparsity levels $s$, the parameter estimation error exhibits linearity for both  processes, which is consistent with the theoretical findings at Corollary \ref{cor:algorithm}. 

%\begin{figure}[t]
%	\centering
%	\includegraphics[width=1.0\linewidth]{fig/s}
%	\caption{Estimation, approximation and parameter estimation errors, i.e. $\|\cm{A}^{K} - \cm{A}_{S}^*\|_{\Fr}^2$, $\|\cm{A}_{S^c}^*\|_{\Fr}^2$ and $\|\cm{A}^{K} - \cm{A}^*\|_{\Fr}^2$, at different sparsity levels $s$ for VARMA (left panel) and VAR (right panel) processes.
%		%The decreasing range of parameter estimation errors is shaded.
%		The range of decreasing parameter estimation errors is shaded. \label{fig:rate_s}}
%\end{figure}

The second experiment is to further verify the theoretical bound of parameter estimation errors at Corollary \ref{cor:algorithm}, and the two data generating processes in the first experiment are employed again.
Note that $r_1=r_2=r$ in both processes, and hence the parameter estimation error is expected to have a rate of $\beta = s(rN+\log T_0)/(T-T_0)$.
Moreover, since the linearity with respect to $s$ has already been confirmed in the first experiment, we fix the sparsity level $s=10$ for the VARMA process and $s=5$ for the VAR process in this experiment.
Finally, we consider three different rates for running orders, i.e. $T_0=\floor{cT^{\alpha}}$ with $\alpha=1/4$, $1/3$ or $1/2$, and the true order with $T_0=9$ is also considered for VAR models. The value of $c$ is set to $1.5$ for the case with $\alpha=1/2$, while $c=3$ for those with $\alpha=1/4$ and $1/3$ such that the resulting $T_0$ is not too small.

%\begin{figure}[t]
%	\centering
%	\includegraphics[width=1\linewidth]{fig/rate}
%	\caption{Plots of parameter estimation errors $\|\cm{A}^{K} - \cm{A}^*\|_{\Fr}^2$ against the error rate $\beta = [rN + \log T_0]s/(T-T_0)$ (left panel), rank $r$ (middle panel) and dimension $N$ (right panel), respectively. The data generating processes are VARMA (upper panel) and VAR (lower panel) models, and the running order $T_0$ is proportional to $T^{\alpha}$ with ``fix" referring to $T_0=9$.  \label{fig:rate}}
%\end{figure}

In order to verify the linearity of parameter estimation errors with respect to the rate $\beta$, rank $r$ and dimension $N$, we consider three group\textcolor{purple}{s} of settings to generate high-dimensional time series. First, by fixing $(N,r)=(20,4)$, one can vary the sample size $T$ such that the values of $\beta$ are equally spaced from $0.4$ to $1.0$. Secondly, we fix $(N,T)=(20,1200)$ for the VARMA process and $(20,700)$ for the VAR process, and let the rank vary among $\{2,3,4,5\}$.
Finally, the dimension $N$ varies among $\{10,20,30,40\}$, while $(r,T)$ is fixed at $(4,2000)$ for the VARMA process and $(4,1200)$ for the VAR process.
All the other settings are the same as those in the first experiment, and Figure \ref{fig:rate} presents the plots of parameter estimation errors, averaged over 500 replications, against the rate $\beta$, rank $r$ and dimension $N$, respectively.
The linearity can be observed roughly from all combinations, which confirms the theoretical findings at Corollary \ref{cor:algorithm}.
Moreover, for each combination of the triplet $(N,r,T)$, a larger parameter estimation error can be observed for a higher running order $T_0$, and it is more pronounced for the VAR process since its sample sizes are relatively small. This is due to the fact that, as $T_0$ increases, the model is more complicated, while the effective sample size becomes smaller.

%\begin{figure}[t]
%	\centering
%	\includegraphics[width=1\linewidth]{fig/hardvsoft}
%	\caption{Plots of parameter estimation errors against sparsity levels $s$ (left panel) for the hard-thresholding (HT) method and tuning parameters $\lambda$ (middle panel) for the soft-thresholding (ST) method, where the optimal setting on each curve with the minimum error is highlighted in red, and boxplots (right panel) for sparsity levels of 500 fitted models at the optimal values of $\lambda$. Two cases are considered: varying sample size $T$ but fixed running order $T_0$ (upper panel) and varying running order $T_0$ but fixed sample size $T$ (lower panel). \label{fig:hardvsoft}}
%\end{figure}

\subsection{Comparison of soft- and hard-thresholding methods}

This subsection conducts another two simulations to compare the finite-sample performance of the soft- and hard-thresholding methods in Section 3.1 at different values of $(T,T_0)$. 
%when the running order $T_0$ and sample size $T$ vary.

In the first experiment, we generate the data by using the VARMA process in the previous subsection with $(N,r)=(20,4)$, and four sample sizes are considered with $T=400$, 600, 800 and 1000. The running order is fixed at $T_0=100$, and the effective sample size is $T_1=T-T_0$.
There are 500 replications for each sample size, and the hard-thresholding method, i.e. Algorithm \ref{alg:AGD-HT}, is first considered to search for estimates with the sparsity level $s$ varying from 4 to 18.
Figure \ref{fig:hardvsoft} gives the parameter estimation errors, averaged over 500 replications, and it can be seen that the optimal sparsity level grows slowly from nine to eleven as the sample size $T$ increases.
In the meanwhile, Algorithm \ref{alg:AGD-HT} modified with the soft-thresholding method in Section 3.1 is also applied, and the tuning parameter $\lambda$ varies among the values of $\{0.5j\times 10^{-3}, 1\leq j\leq 33\}$.
The averaged parameter estimation errors are presented in Figure \ref{fig:hardvsoft}, and the optimal values of $\lambda$ change a lot when the sample size increases from $T=400$ to 1000.
Moreover, since the fitted models at each fixed $\lambda$ may have different sparsity levels, Figure \ref{fig:hardvsoft} also plots the sparsity levels of 500 fitted models at the optimal value of $\lambda$ for each sample size.
The sparsity level varies roughly from 20 to 30, and the variation decreases as the sample size increases.
Finally, although parameter estimation errors for both methods become larger as sample sizes decrease, those for the soft-thresholding method may vary dramatically, especially when the sample size is as small as $T=400$.

The second experiment fixes the sample size at $T=800$, while four running orders are considered with $T_0=25$, 50, 100 and 200.
The sparsity level for the hard-thresholding method varies from $s=5$ to 14, and the tuning parameter for the soft-thresholding method takes the values of $\lambda\in\{0.5j\times 10^{-3}, 4\leq j\leq 20\}$. All the other settings are the same as those in the first experiment, and the estimation results are also given in Figure \ref{fig:hardvsoft}.
The stability of the hard-thresholding method is confirmed again, while the optimal values of $\lambda$ are more sensitive to the change in $T_0$ for the soft-thresholding method.
In fact, when the running order is as large as $T_0=200$, the parameter estimation errors from the soft-thresholding method vary dramatically, and we can even observe a much higher variation in sparsity levels of the fitted models at the optimal values of $\lambda$. This further undermines the stability of the soft-thresholding method.

\section{Two empirical examples} \label{sec:empirical}

\subsection{Macroeconomic dataset} \label{subsec:macro}

This dataset contains observations of $N=20$ quarterly macroeconomic variables from June 1959 to December 2019, with $T=243$, retrieved from FRED-QD \citep{MN16}. These variables come from four categories: (i) interest rates, (ii) money and credit, (iii) exchange rates, and (iv) stock markets. These categories are usually considered in the construction of financial condition indices, since they reflect important factors that can affect the stance of monetary policy and aggregate demand conditions
\citep{goodhart2001asset,bulut2016financial,hatzius2010financial}. All sequences are first transformed to be stationary, and are then standardized to have zero mean and unit variance; see the supplementary file for the variables and their transformations.

We use Algorithm \ref{alg:AGD-HT} to conduct the VAR sieve estimation, and the running order is set to $4\leq T_0\leq 12$, which corresponds to observations in the past one to three years. The high-dimensional AIC at Section 3.1 is used to search for the values of $(r_1, r_2, s)$ with tunning parameter $c = 0.004$, and the sparsity level $s=2$ is always chosen for each running order $T_0$. 
Specifically speaking, the first and second lags are consistently selected across all $T_0$, i.e. the macroeconomic variables in the dataset are mainly driven by economic situations of the recent two quarters.
As a result, we fix the running order at $T_0=4$, and the ranks are at $({r}_1, {r}_2) = (7, 4)$, selected by the AIC.
Note that $\widehat{\bm{U}}_1$ and $\widehat{\bm{U}}_2$ are loadings of response and predictor factors in the fitted supervised factor model, respectively, and their values are plotted in Figure \ref{fig:loading}.
It can be seen that, for the seven fitted response factors (RF1--RF7), interest rates contribute most, and each interest rate variable actually occupies one response factor. The money and credit variables also have contributions to RF3 and RF7.
We may argue that, among the four categories of macroeconomic variables, interest rates have the most temporally dependent structures, i.e. they are most likely predictable, while money and credit variables can be predicted to some extent. 
On the other hand, low weightings are assigned to exchange rates and stock markets, and this implies that their movements cannot be explained well by using these endogenous variables only. In fact, it is hard to forecast these two types of variables in real applications, and perhaps we may need to involve some exogenous variables in order to improve the forecasting power.

It can also been seen from Figure \ref{fig:loading} that, different from response factors, compositions of predictor factors (PF1--PF7) are more balanced and diversified, and we argue that the four categories of variables are all driving forces of the market. Specifically, PF1 is positively correlated to short-term interest rate spreads, stock market performance and the strength of US dollar, and we may interpret it as a driver of short-term economic growth. In the meanwhile, PF4 can be treated as a driver of longer-term economic growth, since it is positively correlated to long-term interest rate spreads and the strength of US dollar, and negatively correlated to the stock market performance. The other two factors, PF2 and PF3, may be interpreted as the combined forces among all four groups of variables. 

%\begin{figure}[t]
%	\centering
%	\includegraphics[width=1.0\linewidth]{fig/factor.png}
%	\caption{Loadings for seven response (upper panel) and four predictor factors (lower panel) in the fitted supervised factor model for the macroeconomic dataset. Blue bars correspond to positive loadings, while red bars to negative ones, where the larger magnitudes are given darker colors. \label{fig:loading}}
%\end{figure}

We next compare the proposed model with five existing ones in the literature, including three VAR-based models and two VARMA-based models.  The three VAR-based models are by (a) the Lasso method \citep{basu2015regularized}, (b) multilinear low-rank (MLR) method with a low-rank structure, rather than weak sparsity, being assumed to the lag direction \citep{Wang2021High}, and (c) sparse higher-order reduced-rank (SHORR) method,  which further enforces sparsity on the factor matrices in (b). 
The two VARMA-based models are from \cite{WBBM21} with (d) the $\ell_1$-penalty and (e) HLag penalty, respectively. 
For the order selection and initial estimators in the above five models, we strictly follow the original papers, and a rolling forecast procedure is used for the comparison. Specifically, we first fit the model with $(T_0,r_1,r_2,s)=(4,7,4,2)$ to the historical data with the ending point iterating from the fourth quarter of 2015 to the third quarter of 2019, and then one-step ahead prediction is conducted for each iteration.
Table~\ref{tab:real} gives the mean squared forecast error (MSFE) and mean absolute forecast error (MAFE) for all six models, and the smallest forecast errors are achieved by the proposed model.
The outperformance of our model against model (a) confirms the existence of low-rank structures in the dataset, which has been widely recognized in the literature \citep{stock2012disentangling}, while that against model (b) verifies the necessity of group-sparsity for coefficient matrices.
In fact, the proposed model even outperforms model (c), i.e. model (b) with sparsity on factor matrices. 
On the other hand, the VARMA model can be used to fit VAR$(\infty)$ processes, since it has a VAR$(\infty)$ form and MA coefficient matrices can depict the decay patterns. Worse performance of two VARMA-based models at Table~\ref{tab:real} further confirms the efficacy of the proposed VAR sieve estimation.

\subsection{Realized volatility}

It is an important task in finance to predict realized volatility \citep{Chen2010}, and this subsection attempts to tackle this problem by considering the daily realized volatility for $N=46$ stocks from January 2, 2012 to December 31, 2013, with sample size $T=495$. The stocks are from S\&P 500 companies with the largest trading volumes on the first day of 2013, and they cover a wide range of sectors, including communication service, information technology, consumer, finance, healthcare, materials and energy. The tick-by-tick data are downloaded from the Wharton Research Data Service (WRDS), and the daily realized volatility is calculated based on five-minute returns \citep{andersen2006volatility}. 
The stationarity can be confirmed for these series by checking their sample autocorrelation functions, and each sequence is standardized to have zero mean and unit variance; see the supplementary file for information on the 46 stocks.

Algorithm \ref{alg:AGD-HT} is applied again to search for the VAR sieve estimates, and the running order is set to $6\leq T_0\leq 10$, i.e. the observations in the past one to two trading weeks. 
The AIC chooses the sparsity level $s=4$ for most running orders $T_0$ and $s=3$ for the others, while the selected lags are always among $\{1,3,4,6\}$. 
We may argue that the movement of market volatility is largely driven by intra-week information, with a bit spillover effect from the previous week.
As a result, the running order is fixed at $T_0=8$, giving two more lags as a buffer, and accordingly the selected ranks and sparsity level by the AIC are $(r_1,r_2,s)=(3,5,4)$.
The rolling forecast procedure in the first example is employed with the last 10\% observations being reserved for one-step ahead prediction, and Table~\ref{tab:real} gives the mean squared forecast error (MSFE) and mean absolute forecast error (MAFE) from the proposed method, as well as the five competing ones.
The proposed model has a better performance by a big margin, and this further confirms the usefulness of our method. 

\section{Conclusion and discussion}

This paper proposes a supervised factor model for high-dimensional time series by introducing low-rank structures to the coefficient matrices of VAR($\infty$) models. With the help of tensor techniques, the proposed model can be rewritten into a form of two factor models, which allows us to interpret it from unsupervised factor modeling perspectives.
For its application on high-dimensional time series, by making use of an interesting fact that the stationarity condition implies the weak group sparsity of coefficient matrices, a rank-constrained group Lasso estimation is considered, and its non-asymptotic properties are carefully investigated by trading-off the estimation, approximation, and truncation errors. 
Moreover, an alternating gradient descent algorithm with thresholding is suggested to search for the high-dimensional estimate, and its theoretical properties, including both statistical and convergence analysis, are also provided.
Finally, as illustrated by two empirical examples, the proposed model exceeds the existing methods in terms of forecasting accuracy while enjoys the nice interpretation of factor models.

The proposed methodology in this paper can be extended along three directions. 
First, to obtain a reliable estimator, the sample size is required to be $T_1\gtrsim \{(r_1\wedge r_2) +s^2\} N+s^2\log T_0$ in Theorem \ref{prop:main}, while the number of variables $N$ may be larger than the sample size $T_1$, say for time-course gene expression data \citep{LZLR09}.
To handle this case, we may further impose group sparsity to the rows of factor matrices $\bm{U}_1$ and $\bm{U}_2$. Specifically, let $s_u=s_{u1}\vee s_{u2}$, where $s_{u1}$ and $s_{u2}$ are the numbers of active rows in $\bm{U}_1$ and $\bm{U}_2$, respectively.
The coefficient tensor can then be estimated by $\widehat{\cm{A}} =\widehat{\cm{G}}\times_1\widehat{\bm{U}}_1\times_2\widehat{\bm{U}}_2$, and
\[
(\widehat{\cm{G}}, \widehat{\bm{U}}_1, \widehat{\bm{U}}_2)= \argmin \mathcal{L}(\cm{G},\bm{U}_1,\bm{U}_2) +\lambda \|\cm{G}\|_{\ddagger} +\lambda_1\|\bm{U}_1\|_{\ddagger} +\lambda_2\|\bm{U}_2\|_{\ddagger},
\]
where $\lambda$, $\lambda_1$ and $\lambda_2$ are three tuning parameters of penalization, $\|\cm{G}\|_{\ddagger}=\sum_{j=1}^{T_0}\|\bm{G}_j\|_{\Fr}$, and $\|\bm{U}_i\|_{\ddagger}=\sum_{j=1}^{N}\|\bm{u}_j^{(i)}\|_2$ with $\bm{u}_j^{(i)}$ being the $j$-th row of $\bm{U}_i$ and $i=1$ and 2.
The required sample size can be reduced to $T_1\gtrsim s_u^2(\log N + r_1 + r_2) + (\log T_0 + r_1r_2)s^2$, and error bounds can be derived by a method similar to that in Section 2.3.
We can also construct an alternating gradient descent method with three thresholdings, which is similar to Algorithm \ref{alg:AGD-HT} in Section 3.1.
Secondly, our current proving techniques heavily depend on the exponential decay of coefficient matrices, which actually excludes many important time series models such as the fractionally integrated autoregressive moving average (ARFIMA) model with long memory \citep{grange1980introduction}. It is urgent to look for a new proving technique to remove this restriction.
Finally, it is of interest to make further inference on estimated coefficient matrices such as checking the significance of some coefficients \citep{xia2022inference,cai2020uncertainty}, and we will leave it for future research.

\bibliography{VARinf}

\clearpage
\newpage
\begin{figure}[t]
	\centering
	\includegraphics[width=1.0\linewidth]{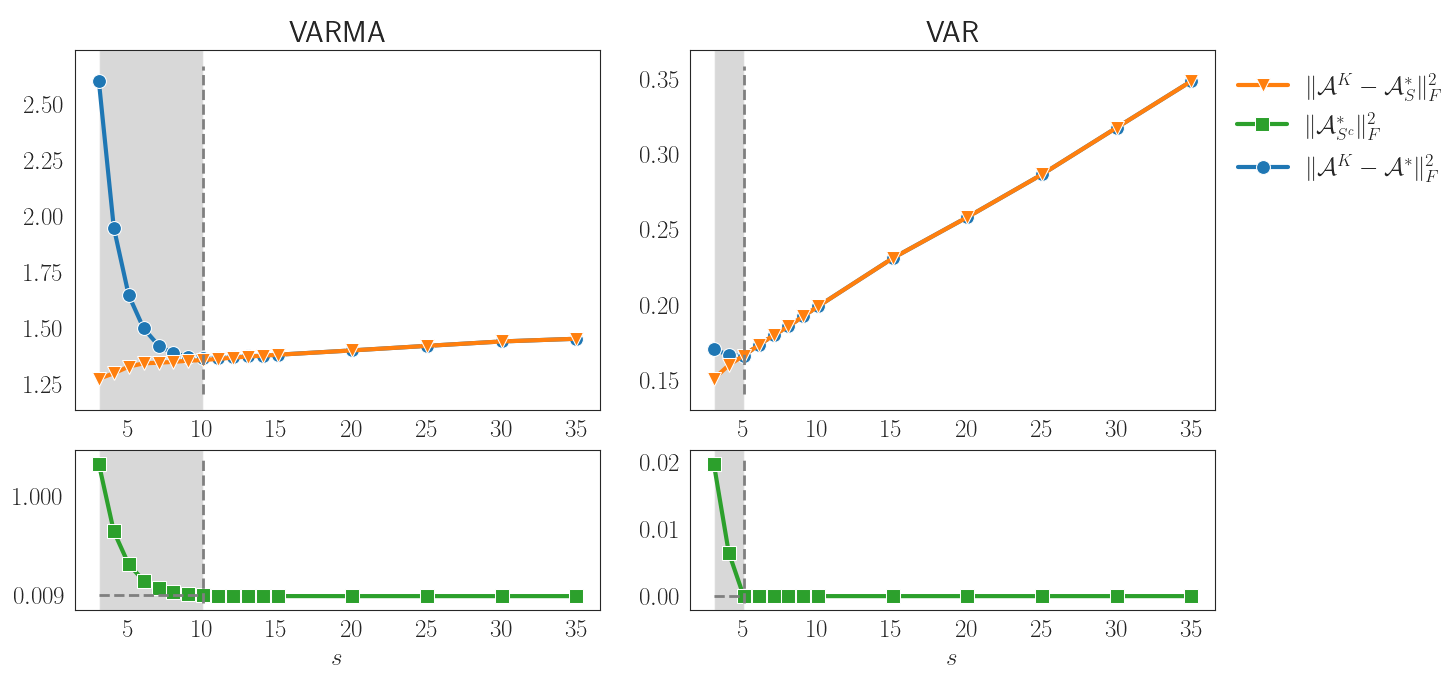}
	\caption{Estimation, approximation and parameter estimation errors, i.e. $\|\cm{A}^{K} - \cm{A}_{S}^*\|_{\Fr}^2$, $\|\cm{A}_{S^c}^*\|_{\Fr}^2$ and $\|\cm{A}^{K} - \cm{A}^*\|_{\Fr}^2$, at different sparsity levels $s$ for VARMA (left panel) and VAR (right panel) processes.
			%The decreasing range of parameter estimation errors is shaded.
			The range of decreasing parameter estimation errors is shaded. \label{fig:rate_s}}
\end{figure}

\begin{figure}[h]
	\centering
	\includegraphics[width=1\linewidth]{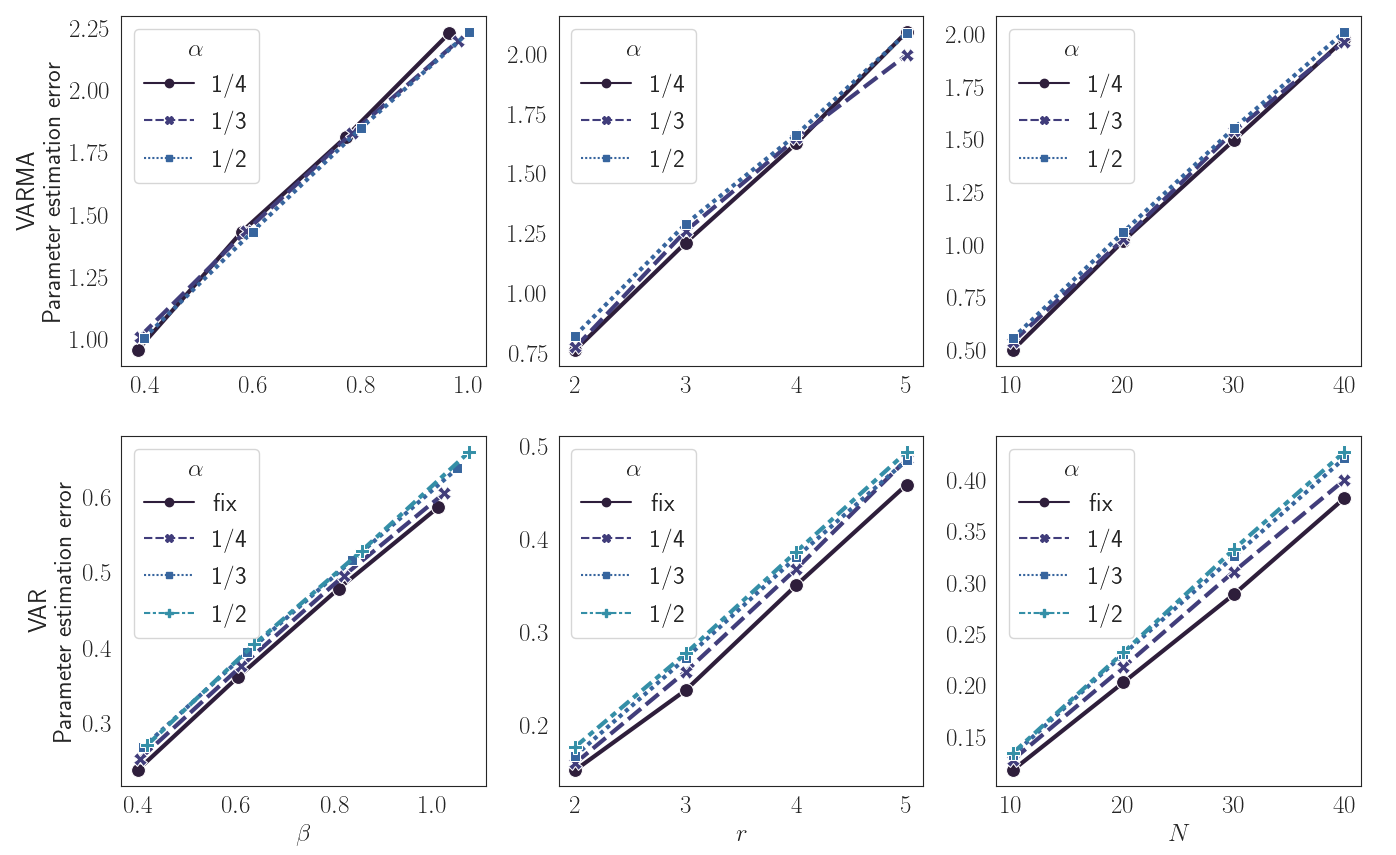}
	\caption{Plots of parameter estimation errors $\|\cm{A}^{K} - \cm{A}^*\|_{\Fr}^2$ against the error rate $\beta = [rN + \log T_0]s/(T-T_0)$ (left panel), rank $r$ (middle panel) and dimension $N$ (right panel), respectively. The data generating processes are VARMA (upper panel) and VAR (lower panel) models, and the running order $T_0$ is proportional to $T^{\alpha}$ with ``fix" referring to $T_0=9$.  \label{fig:rate}}
\end{figure}

\clearpage

\newpage

\begin{figure}[t]
	\centering
	\includegraphics[width=1\linewidth]{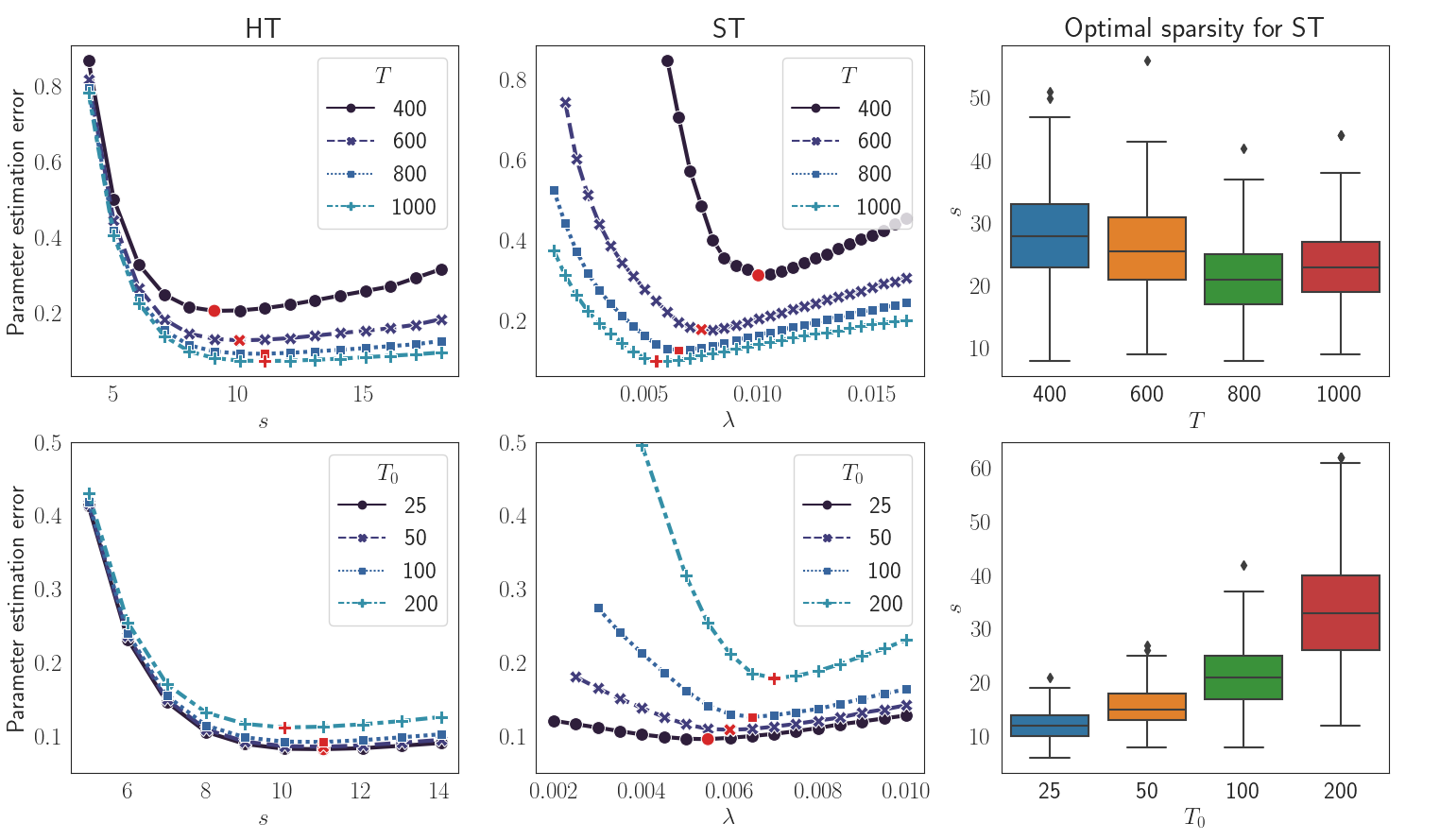}
	\caption{Plots of parameter estimation errors against sparsity levels $s$ (left panel) for the hard-thresholding (HT) method and tuning parameters $\lambda$ (middle panel) for the soft-thresholding (ST) method, where the optimal setting on each curve with the minimum error is highlighted in red, and boxplots (right panel) for sparsity levels of 500 fitted models at the optimal values of $\lambda$. Two cases are considered: varying sample size $T$ but fixed running order $T_0$ (upper panel) and varying running order $T_0$ but fixed sample size $T$ (lower panel). \label{fig:hardvsoft}}
\end{figure}

%\clearpage

\begin{table}[ht]
	\caption{Mean squared forecast errors (MSFE) and mean absolute forecast errors (MAFE) of the VAR sieve method and five competing models on macroeconomic and realized volatility datasets. The smallest numbers in each row are highlighted in bold.}
	\label{tab:real}
	\centering
	\begin{tabular}{@{}cccccccc@{}}
			\hline
			&&\multicolumn{3}{c}{\makecell{VAR}}&\multicolumn{2}{c}{\makecell{VARMA }} & \multirow{2}{*}{Ours} \\
			\cmidrule(r){3-5}\cmidrule(r){6-7}
			&&Lasso&MLR&SHORR&$\ell_1$&HLag&\\
			\midrule
			\multirow{2}{3.5cm}{\makecell{Macroeconomic}}&MSFE&2.78&2.77&2.71&2.80&2.79&\textbf{2.68}\\
			&MAFE&9.26&9.27&8.99&9.28&9.24&\textbf{8.68}\\
			\midrule
			\multirow{2}{3.5cm}{\makecell{Realized Volatility}}&MSFE&5.17&4.93&4.87&5.19&5.19&\textbf{4.72}\\
			&MAFE&21.58&19.02&18.22&21.70&21.70&\textbf{17.40}\\
			\bottomrule
		\end{tabular}
\end{table}

\begin{figure}[h]
	\centering
	\includegraphics[width=1.0\linewidth]{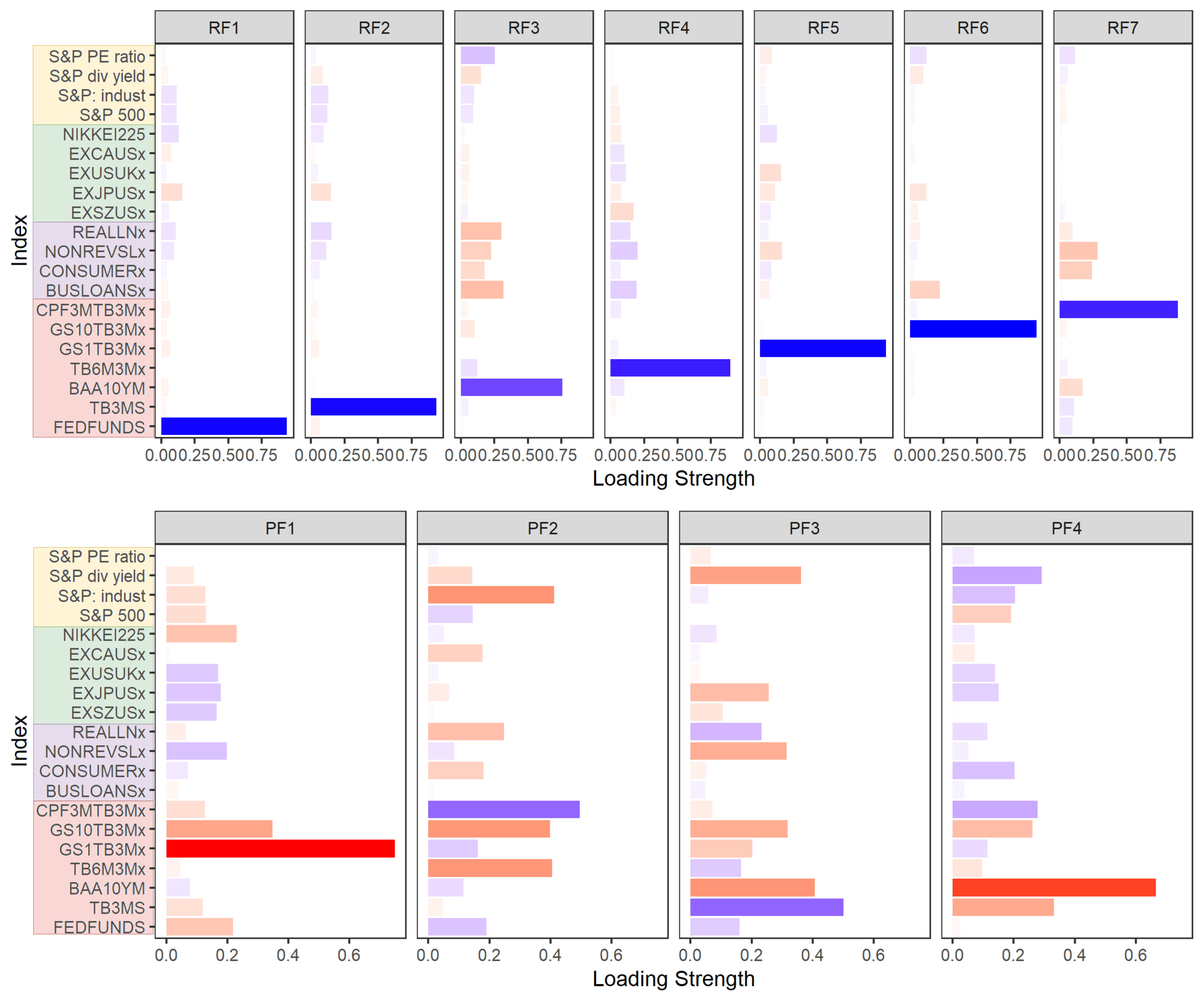}
	\caption{Loadings for seven response (upper panel) and four predictor factors (lower panel) in the fitted supervised factor model for the macroeconomic dataset. Blue bars correspond to positive loadings, while red bars to negative ones, where the larger magnitudes are given darker colors. \label{fig:loading}}
\end{figure}

\clearpage
\newpage
\renewcommand{\thesection}{Appendix A}
\renewcommand{\thelemma}{A.\arabic{lemma}}
\renewcommand{\thesubsection}{A.\arabic{subsection}}
\renewcommand{\theequation}{A.\arabic{equation}}
\setcounter{lemma}{0}

\section{Tensor notations and decomposition}

This appendix gives a brief introduction to tensor notations and Tucker decomposition, and a detailed review on tensor notations and operations can be referred to in \cite{Kolda09}.
Tensors, also known as multidimensional arrays, are higher-order extensions of matrices, and a multidimensional array $\cm{A}\in\mathbb{R}^{d_1\times\cdots\times d_K}$ is called a $K$-th-order tensor, where the order of a tensor is known as the dimension, way or mode.
This paper concentrates on third-order tensors.

For a tensor $\cm{A}\in\mathbb{R}^{d_1\times d_2\times d_3}$, its element is denoted by $\cm{A}_{ijk}$ for $1\leq i\leq d_1,1\leq j\leq d_2$ and $1\leq k\leq d_3$, and the Frobenius norm is defined as $\|\cm{A}\|_{\Fr} = \sqrt{\sum_{i,j,k}\cm{A}_{ijk}^2}$.
We define its mode-1 multiplication with a matrix $\bm{B}\in\mathbb{R}^{d_1\times p_1}$ as $\cm{A}\times_1\bm{B}\in\mathbb{R}^{p_1\times d_2\times d_3}$ with elements of
$(\cm{A}\times_1\bm{B})_{\ell jk}=\sum_{i=1}^{d_1}\cm{A}_{ijk}\bm{B}_{\ell i}$. The mode-2 and -3 multiplications, $\times_2$ and $\times_3$, can be defined similarly.

Matricization or unfolding is an operation to reshape a tensor into matrices of different sizes, and it can help to link the concepts and properties of matrices to those of tensors.
The mode-1 matricization of $\cm{A}$ is defined as $\cm{A}_{(1)}\in\mathbb{R}^{d_1\times d_2d_3}$, whose $\{i,(k-1)d_3+j\}$-th entry is $\cm{A}_{ijk}$ for all possible $i$'s, $j$'s and $k$'s, i.e. $\cm{A}_{(1)}$  contains all mode-1 fibers $\{(\cm{A}_{[:,j,k]})\in\mathbb{R}^{d_1}:1\leq j\leq d_2,1\leq k\leq d_3\}$.
We can similarly define the mode-2 and -3 matricizations of $\cm{A}$, denoted by $\cm{A}_{(2)}\in\mathbb{R}^{d_2\times d_1d_3}$ and $\cm{A}_{(3)}\in\mathbb{R}^{d_3\times d_1d_2}$, respectively.
When the tensor $\cm{A}$ has a form of $\cm{A}_{(1)}=(\bm{A}_1,\ldots,\bm{A}_{d_3} )$ with $\bm{A}_j\in\mathbb{R}^{d_1\times d_2}$ for all $1\leq j\leq d_3$, it holds that $\cm{A}_{(2)}=(\bm{A}_1^{\prime},\ldots,\bm{A}_{d_3}^{\prime} )$ and $\cm{A}_{(3)}=(\mathrm{vec}(\bm{A}_1),\ldots,\mathrm{vec}(\bm{A}_{d_3}) )^{\prime}$.

The multilinear ranks of a tensor $\cm{A}\in\mathbb{R}^{d_1\times d_2\times d_3}$ is defined as $(r_1,r_2,r_3)$, where
\[
r_1=\mathrm{rank}(\cm{A}_{(1)}),\hspace{2mm} r_2=\mathrm{rank}(\cm{A}_{(2)})\hspace{2mm}\text{and}\hspace{2mm} r_3=\mathrm{rank}(\cm{A}_{(3)}).
\]
Accordingly, there exists a Tucker decomposition \citep{tucker1966some},
\begin{equation*}
	\cm{A}=\cm{G}\times_1\bm{U}_1\times_2\bm{U}_2\times_3\bm{U}_3,
\end{equation*}
where $\cm{G}\in \mathbb{R}^{r_1\times r_2\times r_3}$ is the core tensor, $\bm{U}_j\in\mathbb{R}^{d_j\times r_j}$ with $1\leq j\leq 3$ are factor matrices.
Note that the Tucker decomposition is not unique, since
$\cm{A}=\cm{G}\times_1\bm{U}_1\times_2\bm{U}_2\times_3\bm{U}_3=(\cm{G}\times_1\bm{O}_1\times_2\bm{O}_2\times_3\bm{O}_3)\times_1 (\bm{U}_1 \bm{O}_1^{-1}) \times_2 (\bm{U}_2 \bm{O}_2^{-1}) \times_3(\bm{U}_3 \bm{O}_3^{-1})$
for any invertible matrices $\bm{O}_i\in\mathbb{R}^{r_i\times r_i}$ with $1\leq i\leq 3$. We can consider the higher order singular value decomposition (HOSVD) of $\cm{A}$, a special Tucker decomposition uniquely defined by choosing $\bm{U}_i$ as the tall matrix consisting of the top $r_i$ left singular vectors of $\cm{A}_{(i)}$ and then setting $\cm{G}=\cm{A}\times_1\bm{U}_1^\top \times_2\cdots\times_d\bm{U}_d^{\top}$.
Note that $\bm{U}_i$'s are orthonormal, i.e. $\bm{U}_i^\top \bm{U}_i=\bm{I}_{r_i}$ with $1\leq i\leq d$.

The three ranks, $r_1$, $r_2$ and $r_3$, are not equal in general.
In particular, when $r_3 = d_3$, the multilinear ranks of $\cm{A}$ are denoted by $(r_1,r_2)$ instead, omitting the rank of mode-3 matricization. 
The multilinear ranks are also known as Tucker ranks, as they are closely related to the Tucker decomposition.
There are many other tensor decomposition methods, such as CP decomposition \citep{Kolda09}, and the ranks of a tensor can be defined in many different ways.

\renewcommand{\thesection}{Appendix B}
\renewcommand{\thelemma}{B.\arabic{lemma}}
\renewcommand{\thesubsection}{B.\arabic{subsection}}
\renewcommand{\theequation}{B.\arabic{equation}}
\setcounter{lemma}{0}

\section{Technical details}
This section first gives the detailed technical proofs for theoretical results in Sections 2 and 3, and then gives additional information on two empirical datasets in Section 5.

\subsection{More notations and preliminaries}

This section gives more notations and preliminaries.
Let $\mathcal{S}^{d-1}=\{\cm{X}\in\mathbb{R}^{d} \mid \|\cm{X}\|_2=1\}$ be the unit sphere of $\mathbb{R}^{d}$ in Euclidean norm. For $n\geq 2$, let $\mathcal{S}^{d_1\times\cdots \times d_n-1}=\{\cm{X}\in\mathbb{R}^{d_1\times \cdots\times d_n} \mid \|\cm{X}\|_{\Fr}=1\}$ be the unit sphere of $\mathbb{R}^{d_1\times \cdots \times d_n}$ in Frobenius norm. 
Moreover, for any positive integer $r\leq N$, denote the set of unit matrices with rank at most $r$ by
\[
\bm{\Theta}_{\Fr}(r) = \{\bm{M}\in\mathbb{R}^{N\times N}\mid\|\bm{M}\|_{\Fr}=1,  \text{rank}(\bm{M}) \leq r\}.
\]

Consider a tensor $\cm{A}\in\mathbb{R}^{N\times N\times T_0}$ such that $\cm{A}_{(1)}=(\bm{A}_1,\bm{A}_2,\ldots, \bm{A}_{T_0})$, i.e. $\bm{A}_j$ is the $j$-th frontal slice of $\cm{A}$. For any index set $S\subseteq \{1,\dots,T_0\}$ with cardinality $|S|=s$, denote by $\cm{A}_S\in\mathbb{R}^{N\times N\times T_0}$ such that its $j$-th frontal slice is $\bm{A}_j$ if $j \in S$, and a zero matrix if $j \notin S$.
Let $\|\cm{A}\|_{\ddagger} = \sum_{j=1}^{T_0}\|\bm{A}_j\|_{\Fr}$, and it holds that
\begin{equation}\label{eq:decomposable}
	\|\cm{A}\|_{\ddagger}=	\|\cm{A}_S\|_{\ddagger}+\|\cm{A}_{S^c}\|_{\ddagger},
\end{equation} 
where $S^c=\{1,\dots, T_0\} \setminus S$.
Moreover, it can be verified that
\begin{equation}\label{eq:Cauchy}
	\|\cm{A}_S\|_{\ddagger}=\sum_{i\in S}\|\bm{A}_i\|_{\Fr}\leq \sqrt{|S|\sum_{i\in S}\|\bm{A}_i\|_{\Fr}^2}\leq \sqrt{s}\|\cm{A}\|_{\Fr}.
\end{equation}

Recall that $\bm{Y}=(\mathbf{y}_T,\mathbf{y}_{T-1},\dots,\mathbf{y}_{T_0+1})\in\mathbb{R}^{N\times T_1}$ and  $\bm{X}=(\mathbf{x}_T,\mathbf{x}_{T-1},\dots,\mathbf{x}_{T_0+1})\in \mathbb{R}^{NT_0\times T_1}$, where
$\mathbf{x}_t=(\mathbf{y}_{t-1}^\prime, \mathbf{y}_{t-2}^\prime, \dots, \mathbf{y}_{t-T_0}^\prime)^\prime$, for $T_0+1\leq t\leq T$. We further denote
\[
\bm{r}_t = \sum_{j=T_0+1}^{\infty}\bm{A}_j^*\mathbf{y}_{t-j}, \hspace{5mm}\bm{R}=(\bm{r}_T,\bm{r}_{T-1},\dots,\bm{r}_{T_0+1})\in\mathbb{R}^{N\times T_1}
\] and
\[
\widetilde{\bm{\varepsilon}}_{t} = \bm{r}_t + \bm{\varepsilon}_t, \hspace{5mm}
\widetilde{\cm{E}} = \bm{R} + \cm{E} = (\widetilde{\bm{\varepsilon}}_T,\widetilde{\bm{\varepsilon}}_{T-1},\dots,\widetilde{\bm{\varepsilon}}_{T_0+1})\in\mathbb{R}^{N\times T_1} 
\]
where 
$\cm{E} = (\bm{\varepsilon}_T,\bm{\varepsilon}_{T-1},\dots,\bm{\varepsilon}_{T_0+1})$.
In addition, we partition $\bm{X}$ into $T_0$ blocks such that $\bm{X}=(\bm{X}_1^\prime, \dots, \bm{X}_{T_0}^\prime)^\prime$, where the $j$-th block is the $N\times T_1$ matrix
\begin{equation}\label{eq:Xj}
	\bm{X}_j=(\mathbf{y}_{T-j}, \mathbf{y}_{T-1-j}, \dots, \mathbf{y}_{T_0+1-j})\in\mathbb{R}^{N\times T_1},\hspace{5mm} j=1,\dots, T_0.
\end{equation}
Let
\begin{equation} \label{eq:sigma_T0}
	\bm{\Sigma}_{T_0} = \mathbb{E}\left (\frac{\bm{X}\bm{X}^\prime}{T_1} \right )= \mathbb{E}(\mathbf{x}_t\mathbf{x}_t^\prime)
	=\left[\begin{matrix}
		\bm{\Gamma}(0)&\bm{\Gamma}^\prime(1)&\cdots&\bm{\Gamma}^\prime(T_0-1)\\
		\bm{\Gamma}(1)&\bm{\Gamma}(0)&\cdots&\bm{\Gamma}^\prime(T_0-2)\\
		\vdots&\vdots&\ddots&\vdots\\
		\bm{\Gamma}(T_0-1)&\bm{\Gamma}(T_0-2)&\cdots&\bm{\Gamma}(0)
	\end{matrix}\right],
\end{equation}
whose $(i,j)$-th block is
\[
\bm{\Gamma}(i-j)=\mathbb{E}\left (\frac{\bm{X}_i\bm{X}_j^\prime}{T_1} \right )=\mathbb{E}(\mathbf{y}_{t-i}\mathbf{y}_{t-j}^\prime)\in\mathbb{R}^{N\times N}, \hspace{5mm} 1\leq i\leq j\leq T_0.
\]

%$r_{\min}=r_1\wedge r_2$

%%================================================================================================

\subsection{Proofs of theoretical results in Section \ref{sec:lass}}\label{appen:sec2}

\renewcommand{\thelemma}{B.\arabic{lemma}}
\setcounter{lemma}{0}

We first give the proofs of Proposition \ref{prop:Aj}, Theorems \ref{prop:main} and \ref{thm:main}, and Corollary \ref{corollary:statistical} in Sections B.2.1-B.2.4, respectively.
Section B.2.5 gives four lemmas used in the proof of Theorem \ref{prop:main}, and two auxiliary lemmas are provided in Section B.2.6.

\subsubsection{Proof of Proposition \ref{prop:Aj}} 

For a certain general linear process, $	\mathbf{y}_t = \bm{\varepsilon}_t+ \sum_{j=1}^{\infty}\bm{\Psi}_j\bm{\varepsilon}_{t-j}$, with fixed $r_1$, $r_2$ and $N$, it is implied by Assumption \ref{assum:glp} that $\sum_{j=0}^{\infty}\|\bm{\Psi}_j\|_{\Fr} < \infty$ and, by directly following (2.2) in \cite{lewis1985prediction} and its discussion, its VAR($\infty$) representation can be uniquely identified below,
\begin{equation}\label{eq:VAR}
	\mathbf{y}_t=\sum_{j=1}^\infty \bm{A}_j\mathbf{y}_{t-j}+\bm{\varepsilon}_t, \hspace{5mm}\text{with}\hspace{5mm}\bm{A}_j = \bm{\Psi}_j - \sum_{k=1}^{j-1}\bm{\Psi}_{j-k}\bm{A}_k,\hspace{2mm}\forall j\geq 1.
\end{equation}
Moreover, it can be easily verified that $\sum_{j=1}^{\infty}\|\bm{A}_j\|_{\op} < \infty$. 

Let $\mathbb{M}_1=\text{colspace}\{\bm{\Psi}_j,j\geq 1\}$ and $\mathbb{M}_2=\text{rowspace}\{\bm{\Psi}_j,j\geq 1\}$ be the column and row spaces of coefficient matrices $\bm{\Psi}_j$'s, respectively. We first show by induction that for all $j\geq 1$,
\begin{equation}\label{eq:A<Psi}
	\text{colspace}(\bm{A}_j) \subseteq \mathbb{M}_1\hspace{4mm}\text{and}\hspace{4mm}\text{rowspace}(\bm{A}_j)\subseteq \mathbb{M}_2.
\end{equation}
Since $\bm{A}_1=\bm{\Phi}_1$, \eqref{eq:A<Psi} holds trivially for $j=1$. Suppose that \eqref{eq:A<Psi} holds for all $j\leq k$, then, by the fact that $\bm{A}_j = \bm{\Psi}_j - \sum_{k=1}^{j-1}\bm{\Psi}_{j-k}\bm{A}_k$, \eqref{eq:A<Psi} also holds for $j=k+1$. And hence the induction is completed. 

In addition, it can also be verified with $\bm{\Psi}_j = \bm{A}_j+\sum_{k=1}^{j-1}\bm{A}_k\bm{\Psi}_{j-k}$ for all $j\geq 1$. By a method similar to \eqref{eq:A<Psi}, we can show that
\[
\mathbb{M}_1 \subseteq \text{colspace}\{\bm{A}_j, j\geq 1\}\hspace{4mm}\text{and}\hspace{4mm}	\mathbb{M}_2 \subseteq\text{rowspace}\{\bm{A}_j, j\geq 1\}.
\]
This hence accomplishes the proof.

%%%%%%%%%%%%%%%%%%%%%%%%%%%%%%%%%%%%%%%%%%%%%%%%%%%%%%%%%%%%%%%%%%%%%%%%%%%%%%
\subsubsection{Proof of Theorem \ref{prop:main}}\label{sec:prop_main}
Let $\bm{\widehat{\Delta}}=\cm{\widehat{A}}-\cm{A}^*$. Since $\cm{\widehat{A}}, \cm{A}^*\in \bm{\Theta}(r_1,r_2)$, we can show that $\bm{\widehat{\Delta}} \in \bm{\Theta}(2r_1,2r_2)=\{\cm{A}\in\mathbb{R}^{N\times N\times T_0}\mid \rank(\cm{A}_{(i)})\leq 2r_i, i=1,2\}$. 
By the optimality of $\cm{\widehat{A}}$, we have
\begin{equation}\label{eq:add1}
	\frac{1}{2T_1}\|\bm{Y}-\cm{\widehat{A}}_{(1)}\bm{X}\|_{\Fr}^2 +\lambda\|\cm{\widehat{A}}\|_{\ddagger}
	\leq \frac{1}{2T_1}\|\bm{Y}-\cm{A}^*_{(1)}\bm{X}\|_{\Fr}^2 +\lambda\|\cm{A}^*\|_{\ddagger}.
\end{equation}

Define the event
\begin{equation*}
	\mathcal{E}_1=\left \{ \sup_{\bm{\Delta}\in\bm{\Theta}_\ddagger(2r_1,2r_2)}
	\langle\bm{\Delta}_{(1)}, \frac{1}{T_1}\widetilde{\cm{E}}\bm{X}^\prime
	\rangle \leq \lambda/2\right \},
\end{equation*}
where 
\begin{equation}\label{eq:space1}
	\bm{\Theta}_\ddagger(r_1,r_2) = \{\cm{A}\in\mathbb{R}^{N\times N\times T_0} \mid  \rank(\cm{A}_{(i)})\leq r_i, i=1,2, \|\cm{A}\|_{\ddagger}=1\}.
\end{equation}
It is then implied by \eqref{eq:add1} that, on the event $\mathcal{E}_1$, 
\begin{align} 
	\frac{1}{T_1}\|\bm{\widehat{\Delta}}_{(1)}\bm{X}\|_{\Fr}^2
	&\leq 2 \langle\bm{\widehat{\Delta}}_{(1)}, \frac{1}{T_1}\widetilde{\cm{E}}\bm{X}^\prime\rangle+ 2\lambda \left (\|\cm{A}^*\|_{\ddagger}-\|\cm{\widehat{A}}\|_{\ddagger} \right ) \notag\\
	&\leq 2\|\bm{\widehat{\Delta}}\|_{\ddagger} \sup_{\bm{\Delta}\in\bm{\Theta}_\ddagger(2r_1,2r_2)}
	\langle\bm{\Delta}_{(1)}, \frac{1}{T_1}\widetilde{\cm{E}}\bm{X}^\prime
	\rangle + 2\lambda \left (\|\cm{A}^*\|_{\ddagger}-\|\cm{\widehat{A}}\|_{\ddagger} \right ) \notag\\
	&\leq  \lambda \left \{ \|\bm{\widehat{\Delta}}\|_{\ddagger} + 2\left (\|\cm{A}^*\|_{\ddagger}-\|\cm{\widehat{A}}\|_{\ddagger} \right ) \right \}\notag\\
	%	&\leq \lambda \left \{\|\bm{\widehat{\Delta}}_{S}\|_{\ddagger}+\|\bm{\widehat{\Delta}}_{S^c}\|_{\ddagger}  + 2\left (\|\cm{A}_{S}^*\|_{\ddagger}+\|\cm{A}_{S^c}^*\|_{\ddagger}-\|\cm{A}_{S}^*\|_{\ddagger}-\|\bm{\widehat{\Delta}}_{S^c}\|_{\ddagger}+\|\cm{A}_{S^c}^*\|_{\ddagger}+\|\bm{\widehat{\Delta}}_{S}\|_{\ddagger} \right ) \right \} \notag\\
	&\leq \lambda\left (4 \|\cm{A}^*_{S^c}\|_{\ddagger}+3\|\bm{\widehat{\Delta}}_{S}\|_{\ddagger}-\|\bm{\widehat{\Delta}}_{S^c}\|_{\ddagger} \right ),  \label{eq:basic-eq}
\end{align}
where the  last inequality follows from \eqref{eq:decomposable} and the triangle inequality. 
Moreover, since the left hand side of \eqref{eq:basic-eq} is non-negative, we have $\bm{\widehat{\Delta}}\in\mathbb{C}(S)\cap \bm{\Theta}(2r_1,2r_2)$, where the restricted set $\mathbb{C}(S)$ is defined as
\[
\mathbb{C}(S)=\left \{\bm{\Delta}\in\mathbb{R}^{N\times N\times T_0}\mid \|\bm{\Delta}_{S^c}\|_{\ddagger}\leq 3\|\bm{\Delta}_{S}\|_{\ddagger} + 4\|\cm{A}^{*}_{S^c}\|_{\ddagger} \right \}.
\]
For any $\bm{\Delta}\in\mathbb{C}(S)$, by  the triangle inequality and \eqref{eq:Cauchy}, we have
\begin{equation} \label{eq:frobenius-link}
	\|\bm{\Delta}\|_{\ddagger}^2\leq \left (4\|\bm{\Delta}_{S}\|_{\ddagger} + 4\|\cm{A}^{*}_{S^c}\|_{\ddagger} \right )^2\leq 32s\|\bm{\Delta}\|_{\Fr}^2 + 32\|\cm{A}^{*}_{S^c}\|_{\ddagger}^2.
\end{equation}

On the other hand, let $\mathcal{E}_2$ be the event that the following restricted eigenvalue (RE) condition holds: 
\begin{equation*}
	\frac{1}{T_1}\|\bm{\Delta}_{(1)}\bm{X}\|_{\Fr}^2 \geq \kappa_{\mathrm{RSC}} \|\bm{\Delta}\|_{\Fr}^2 - {\tau}^2\|\bm{\Delta}\|_{\ddagger}^2 \quad\text{for all}\quad \bm{\Delta}\in\mathbb{R}^{N\times N\times T_0},
\end{equation*}
where $\tau^2=C\sqrt{(N+\log T_0)/T_1}$.
Note that, from \eqref{eq:Cauchy} and \eqref{eq:basic-eq},
${T_1}^{-1}\|\bm{\widehat{\Delta}}_{(1)}\bm{X}\|_{\Fr}^2 \leq 4\lambda \|\cm{A}^*_{S^c}\|_{\ddagger}+3\lambda\sqrt{s}\|\bm{\widehat{\Delta}}\|_{\Fr}$.
As a result, from \eqref{eq:frobenius-link} and on the event $\mathcal{E}_1\cap\mathcal{E}_2$,
\begin{equation}\label{eq:qdr}
	4\lambda \|\cm{A}^*_{S^c}\|_{\ddagger}+3\lambda\sqrt{s}\|\bm{\widehat{\Delta}}\|_{\Fr}\geq \left (\kappa_{\mathrm{RSC}}-32{\tau}^2 s \right )\|\bm{\widehat{\Delta}}\|_{\Fr}^2-32{\tau}^2\|\cm{A}^{*}_{S^c}\|_{\ddagger}^2 \geq\frac{\kappa_{\mathrm{RSC}}}{2}\|\bm{\widehat{\Delta}}\|_{\Fr}^2-32{\tau}^2\|\cm{A}^{*}_{S^c}\|_{\ddagger}^2,
\end{equation}
if $s\leq \kappa_{\mathrm{RSC}}/(64 {\tau}^2)$, i.e., as long as
\begin{equation}\label{eq:T1cond}
	T_1 \gtrsim  s^2 (N+\log T_0).
\end{equation}
In view of  \eqref{eq:qdr}, by solving the quadratic function in $\|\bm{\widehat{\Delta}}\|_{\Fr}$  we can show that, on the event $\mathcal{E}_1\cap\mathcal{E}_2$, if  \eqref{eq:T1cond} holds,  then
\begin{equation}\label{eq:Frerror}
	\|\bm{\widehat{\Delta}}\|_{\Fr}^2 \lesssim \kappa_{\mathrm{RSC}}^{-1}(\lambda^2s+ \lambda\|\cm{A}^{*}_{S^c}\|_{\ddagger}+{\tau}^2\|\cm{A}^{*}_{S^c}\|_{\ddagger}^2). 
\end{equation}

Since 
\begin{equation*}
	\sup_{\bm{\Delta}\in\bm{\Theta}_\ddagger(2r_1,2r_2)}\langle\bm{\Delta}_{(1)}, \frac{1}{T_1}\widetilde{\cm{E}}\bm{X}^\prime\rangle
	\leq \sup_{\bm{\Delta}\in\bm{\Theta}_\ddagger(2r_1,2r_2)} \langle\bm{\Delta}_{(1)}, \frac{1}{T_1}\cm{E}{\bm{X}}^\prime\rangle
	+\sup_{\bm{\Delta}\in\bm{\Theta}_\ddagger(2r_1,2r_2)} \langle\bm{\Delta}_{(1)}, \frac{1}{T_1}\bm{R}\bm{X}^\prime\rangle,
\end{equation*}
by Lemmas  \ref{lemma:devbd1} and  \ref{lemma:devbd2}, if $\lambda\gtrsim \sqrt{\{(r_1\wedge r_2)N+\log T_0\}/{T_1}}$, we then have
\begin{equation}\label{eq:devbd}
	\mathbb{P}(\mathcal{E}_1^c) \leq 
	\mathbb{P}\left\{ \sup_{\bm{\Delta}\in\bm{\Theta}_\ddagger(2r_1,2r_2)}
	\langle\bm{\Delta}_{(1)}, \frac{1}{T_1}\widetilde{\cm{E}}\bm{X}^\prime
	\rangle \geq C \sqrt{\frac{(r_1\wedge r_2)N+\log T_0}{T_1}} \right \}  \leq  C e^{-(r_1\wedge r_2)N-\log T_0}.
\end{equation}	
In addition, by Lemma \ref{lemma:RE},
\begin{equation}\label{eq:RE}
	\mathbb{P}(\mathcal{E}_2^c) \leq C e^{-N-\log T_0}.
\end{equation}
Combining  \eqref{eq:Frerror}--\eqref{eq:RE}, we prove the upper bound for $\|\bm{\widehat{\Delta}}\|_{\Fr}$.

Lastly, by \eqref{eq:Cauchy} and \eqref{eq:basic-eq}, we have
\begin{align*} %\label{eq:epred1}
	\frac{1}{T_1}\|\bm{\widehat{\Delta}}_{(1)}\bm{X}\|_{\Fr}^2 
	\leq \lambda \left (4 \|\cm{A}^*_{S^c}\|_{\ddagger}+3\sqrt{s}\|\bm{\widehat{\Delta}}\|_{\Fr} \right ) \lesssim  \lambda\|\cm{A}^*_{S^c}\|_{\ddagger} + \lambda^2 s + \|\bm{\widehat{\Delta}}\|_{\Fr}^2,
\end{align*}
since $2\lambda\sqrt{s}\|\bm{\widehat{\Delta}}\|_{\Fr}\leq \lambda^2 s + \|\bm{\widehat{\Delta}}\|_{\Fr}^2$.
Combining this with \eqref{eq:Frerror}--\eqref{eq:RE} and the condition in  \eqref{eq:T1cond}, we can similarly prove the upper bound for ${T_1}^{-1}\|\bm{\widehat{\Delta}}_{(1)}\bm{X}\|_{\Fr}^2$.

%%%%%%%%%%%%%%%%%%%%%%%%%%%%%%%%%%%%%%%%%%%%%%%%%%%%%%%%%%%%%%%%%%%%%%%%%%%%%%
\subsubsection{Proof of Theorem \ref{thm:main}}

By Assumptions \ref{assum:Adecay} \& \ref{assum:T0} and the low-rank conditions at \eqref{eq:VAR-low-rank},
\begin{equation}\label{eq:e_trunc1}
	e_\trunc =  \sum_{j=T_0+1}^{\infty}\|\bm{A}_j^*\|^2_{\Fr} \leq  \frac{(r_1\wedge r_2)\rho^{2T_0}}{1-\rho} \lesssim \frac{r_1\wedge r_2}{T_1^4},
\end{equation}
which is clearly dominated by $\|\bm{\widehat{\Delta}}\|_{\Fr}^2$ in \eqref{eq:Frerror} under the  condition on $\lambda$.
Therefore, combining  \eqref{eq:Frerror} and \eqref{eq:e_trunc1}, we have that,
with probability at least $1-Ce^{-(r_1\wedge r_2)N-\log T_0}$,
\begin{equation}\label{eq:e_est}
	e_{\est}(\cm{\widehat{A}}_{\infty}) = \|\cm{\widehat{A}}_{\infty} - \cm{A}^*_{\infty}\|_{\Fr}^2 = \|\bm{\widehat{\Delta}}\|_{\Fr}^2 + e_\trunc 
	\lesssim  \kappa_{\mathrm{RSC}}^{-1}(\lambda^2s+\lambda\|\cm{A}^{*}_{S^c}\|_{\ddagger}+{\tau}^2\|\cm{A}^{*}_{S^c}\|_{\ddagger}^2).
\end{equation}

Moreover, by \eqref{eq:Cauchy} and \eqref{eq:basic-eq}, we have with probability at least $1-Ce^{-(r_1\wedge r_2)N-\log T_0}$,
\begin{align} \label{eq:epred1}
	\begin{split}
		e_{\pred}(\cm{\widehat{A}}_{\infty})
		&= \frac{1}{T_1}\sum_{t=T_0+1}^{T}\|\bm{\widehat{\Delta}}_{(1)}\mathbf{x}_t + \bm{r}_t\|_2^2 \\
		& =
		\frac{1}{T_1}\|\bm{\widehat{\Delta}}_{(1)}\bm{X}\|_{\Fr}^2 
		+ 
		\frac{2}{T_1}\langle\bm{\widehat{\Delta}}_{(1)}\bm{X}, \bm{R}\rangle  + \frac{1}{T_1}\|\bm{R}\|_{\Fr}^2 \\
		& \leq \frac{1}{T_1}\|\bm{\widehat{\Delta}}_{(1)}\bm{X}\|_{\Fr}^2  +  2\|\bm{\widehat{\Delta}}\|_{\Fr}\sup_{\bm{\Delta}\in\bm{\Theta}_\ddagger(2r_1,2r_2)}\langle\bm{\Delta}_{(1)}, \frac{1}{T_1}\bm{R}\bm{X}^\prime
		\rangle + \frac{1}{T_1}\|\bm{R}\|_{\Fr}^2\\
		&\lesssim  \kappa_{\mathrm{RSC}}^{-1}(\lambda^2s+\lambda\|\cm{A}^{*}_{S^c}\|_{\ddagger}+{\tau}^2\|\cm{A}^{*}_{S^c}\|_{\ddagger}^2),
	\end{split}
\end{align}
where the last inequality uses Theorem \ref{prop:main}, Lemmas \ref{lemma:devbd2} and \ref{lemma:errtrunc1}.

Next, we characterize the optimal bounds for $e_{\est}(\cm{\widehat{A}}_{\infty})$ in \eqref{eq:e_est} and $e_{\pred}(\cm{\widehat{A}}_{\infty})$ in \eqref{eq:epred1}. Consider a family of subsets indexed by a threshold $\gamma>0$:
%\begin{equation*}%\label{eq:wgs}
$S_\gamma = \left \{j\in\{1, \dots, T_0\}\mid\| \bm{A}_j^*\|_\text{F} > \gamma \right \}$,
%\end{equation*}
and let $S_\gamma^c=\{1,\dots,T_0\}\setminus S_\gamma$.  Note that 
under Assumption \ref{assum:Adecay} and the low-rank conditions at \eqref{eq:trunc_tensor}, $\|\bm{A}_j^*\|_\Fr\leq C\sqrt{r_1\wedge r_2}\rho^j$ for $j\geq 1$. 
Let $Q_\gamma$ be the smallest integer such that $C\sqrt{r_1\wedge r_2}\rho^j\leq \gamma$ for all $j\geq Q_\gamma$. Then,
\begin{equation}\label{eq:Q}
	Q_\gamma = \left \lceil \frac{\log(C\sqrt{r_1\wedge r_2}/\gamma)}{\log(1/\rho)} \right \rceil,
\end{equation}
where $\lceil\cdot\rceil$ is the ceiling function. Moreover, since $C\sqrt{r_1\wedge r_2}\rho^{Q_\gamma}\leq \gamma$, we have	
\begin{align}\label{eq:approx_bound}
	\|\cm{A}^{*}_{S_\gamma^c}\|_{\ddagger} = \sum_{j\in S_\gamma^c\cap \{1,\dots, Q_\gamma\}}\|\bm{A}_j^*\|_{\Fr} + \sum_{j=Q_\gamma+1}^{T_0}\|\bm{A}_j^*\|_{\Fr}
	\leq \gamma Q_\gamma + \sum_{j=Q_\gamma+1}^{\infty}C\sqrt{r_1\wedge r_2}\rho^j
	\lesssim \gamma Q_\gamma,
\end{align}
as long as $Q_\gamma\geq c$ for some absolute constant $c>0$.
Since  $\|\bm{A}_j^*\|_\Fr> \gamma \geq C\sqrt{N}\rho^{Q_\gamma}$ for any $j\in S_\gamma$, we have $S_\gamma\subseteq \{1,\dots, Q_\gamma\}$, and hence $|S_\gamma|\leq Q_\gamma$. Then, the upper bounds of $\lambda^2 |S_\gamma|$ and $\lambda\|\cm{A}^{*}_{S_\gamma^c}\|_{\ddagger}$, i.e., 
$\lambda^2 Q_\gamma$ and $\lambda  \gamma Q_\gamma$, are balanced when
$\gamma \asymp \lambda$. With this choice of $\gamma$ and the rate of $\lambda$ specified in the theorem, by \eqref{eq:Q}, we can show that 
\[
Q_\gamma \lesssim   \frac{\log T_1}{\log(1/\rho)}
\]
and $\tau^2 Q_\gamma \leq C$. Thus,  ${\tau}^2\|\cm{A}_{S_{\gamma}^c}^*\|_{\ddagger}^2$ is dominated by $\lambda\|\cm{A}_{S_{\gamma}^c}^*\|_{\ddagger}$. From \eqref{eq:e_est}, $e_{\est}(\cm{\widehat{A}}_{\infty})\lesssim \lambda^2 Q_\gamma$. In addition, substituting the upper bound of $|S_\gamma|$, i.e.,  $\log T_1/\log(1/\rho)$, into the sample size condition in Theorem \ref{prop:main}, we have \eqref{eq:thm1}. As a result, the upper bound for $e_{\est}(\cm{\widehat{A}}_{\infty})$ in this theorem holds.  Finally, the result for $e_{\pred}(\cm{\widehat{A}}_{\infty})$ can be proved by a similar method.

%%%%%%%%%%%%%%%%%%%%%%%%%%%%%%%%%%%%%%%%%%%%%%%%%%%%%%%%%%%%%%%%%%%%%%%%%%%%%%
\subsubsection{Proof of Corollary \ref{corollary:statistical}}

This proof largely follows from the proof of Theorem \ref{prop:main} in Section \ref{sec:prop_main}. For the finite-order AR model with a fixed order $T_0$, $\bm{A}_j^* = \bm{0}$ for all $j\geq T_0+1$ leading to $\bm{R} = \bm{0}$ and $\cm{\widetilde{E}} = \cm{E}$. Subsequently, the event $\mathcal{E}_1$ becomes 
\[
\mathcal{E}_1=\left \{ \sup_{\bm{\Delta}\in\bm{\Theta}_\ddagger(2r_1,2r_2)}
\langle\bm{\Delta}_{(1)}, \frac{1}{T_1}\cm{E}\bm{X}^\prime
\rangle \leq \lambda/2\right \},
\]
and its probability can be established directly from Lemma \ref{lemma:devbd1}. Meanwhile, the other intermediate steps are the same as in the proof of Theorem \ref{prop:main} in Section \ref{sec:prop_main}.

%%%%%%%%%%%%%%%%%%%%%%%%%%%%%%%%%%%%%%%%%%%%%%%%%%%%%%%%%%%%%%%%%%%%%%%%%%%%%%
\subsubsection{Four lemmas used in the proof of Theorem \ref{prop:main}}

We first state the four lemmas and then give their technical proofs.
\begin{lemma}\label{lemma:devbd1}
	Suppose that Assumptions \ref{assum:glp} -- \ref{assum:error} hold.  If $T_1\gtrsim (r_1 \wedge r_2)N + \log T_0$, then
	\begin{equation*}
		\mathbb{P}\left\{ \sup_{\bm{\Delta}\in\bm{\Theta}_\ddagger(2r_1,2r_2)} \langle\bm{\Delta}_{(1)}, \frac{1}{T_1}\cm{E}{\bm{X}}^\prime\rangle \geq C \sqrt{\frac{(r_1 \wedge r_2)N + \log T_0}{T_1}}\right\}\leq C e^{ - (r_1 \wedge r_2)N - \log T_0}.
	\end{equation*}
\end{lemma}
%%%%%%%%%%%%%%%%%%%%%%%%%%%%%%%%%%%%%%%%%%%%%%%%%%%%%%%%%%%%%%%%%%%%%%%%%%%%%%
\begin{lemma}\label{lemma:devbd2}
	Suppose that Assumptions \ref{assum:glp} -- \ref{assum:T0} hold. If $T_1\gtrsim (r_1 \wedge r_2)N + \log T_0$, then
	\begin{equation*}
		\mathbb{P}\left\{\sup_{\bm{\Delta}\in\bm{\Theta}_\ddagger(2r_1,2r_2)}\langle\bm{\Delta}_{(1)},\frac{1}{T_1}\bm{R}\bm{X}^\prime\rangle \geq 
		C \frac{ \sqrt{N} }{T_1} \right\}
		\leq C e^{ - (r_1 \wedge r_2)N - \log T_0}.
	\end{equation*}
\end{lemma}
%%%%%%%%%%%%%%%%%%%%%%%%%%%%%%%%%%%%%%%%%%%%%%%%%%%%%%%%%%%%%%%%%%%%%%%%%%%%%%
\begin{lemma}[RE condition]\label{lemma:RE}
	Suppose that Assumptions \ref{assum:glp} -- \ref{assum:error} hold.  If $T_1\gtrsim N + \log T_0$, then with probability at least
	$1-Ce^{ - N - \log T_0}$,
	the following RE condition is satisfied:
	\begin{equation*}%\label{eq:RE}
		\frac{1}{T_1}\|\bm{\Delta}_{(1)}\bm{X}\|_{\Fr}^2 \geq \kappa_{\mathrm{RSC}}\|\bm{\Delta}\|_{\Fr}^2 - \tau^2\|\bm{\Delta}\|_{\ddagger}^2 \quad\text{for all}\quad \bm{\Delta}\in\mathbb{R}^{N\times N\times T_0},
	\end{equation*}
	where $\kappa_{\mathrm{RSC}}=\lambda_{\min}(\bm{\Sigma}_\varepsilon)\mu_{\min}(\bm{\Psi}_*)$ and
	$\tau^2=C \sqrt{ (N + \log T_0) /T_1}$.
\end{lemma}
%%%%%%%%%%%%%%%%%%%%%%%%%%%%%%%%%%%%%%%%%%%%%%%%%%%%%%%%%%%%%%%%%%%%%%%%%%%%%%
\begin{lemma}\label{lemma:errtrunc1}
	If Assumptions \ref{assum:glp} -- \ref{assum:T0} hold, then
	\[
	\mathbb{P}\left\{ \frac{1}{T_1}\left\|\bm{R}\right\|_{\Fr}^2 \geq  C  \frac{N^{3/2}\rho^{T_0}}{T_1}  \right\}\leq  C e^{ - N } .
	\]
\end{lemma}

%%%%%%%%%%%%%%%%%%%%%%%%%%%%%%%%%%%%%%%%%%%%%%%%%%%%%%%%%%%%%%%%%%%%%%%%%%%%%%
\paragraph{Proof of Lemma \ref{lemma:devbd1}}\label{sec:devbd1}

The proof of Lemma \ref{lemma:devbd1} relies on the following discretization result. 
%%%%%%%%%%%%%%%%%%%%%%%%%%%%%%%%%%%%%%%%%%%%%%%%%%%%%%%%%%%%%%%%%%%%%%%%%%%%%%
\begin{lemma}[Discretization]\label{lemma:discretize}
	Let $\widebar{\bm{\Theta}}_{\Fr}(2r_{\min})$ be a minimal $1/2$-net for $\bm{\Theta}_{\Fr}(2r_{\min})$ in the Frobenius norm.
	\begin{equation*}
		\sup_{\bm{\Delta}\in\bm{\Theta}_\ddagger(2r_1,2r_2)} \langle\bm{\Delta}_{(1)}, \cm{E}{\bm{X}}^\prime\rangle 
		\leq
		4\max_{1\leq j\leq T_0}\max_{\bm{M}\in\widebar{\bm{\Theta}}_{\Fr}(2r_{\min})}\langle \bm{M}, \cm{E}\bm{X}_j^\prime \rangle,
	\end{equation*}
	where $r_{\min}=r_1\wedge r_2$, and $\bm{\Theta}_\ddagger(2r_1,2r_2)$ is defined as in \eqref{eq:space1}.
\end{lemma}

\begin{proof}[Proof of Lemma \ref{lemma:discretize}]
	For any $\bm{\Delta}\in\bm{\Theta}_\ddagger(2r_1,2r_2)$, there exist $\cm{H}\in\mathbb{R}^{2r_1\times 2r_2\times T_0}$, $\bm{U}_1\in\mathcal{O}^{N\times 2r_1}$ and $\bm{U}_2\in\mathcal{O}^{N\times 2r_2}$ such that
	$$\bm{\Delta} = \cm{H}\times_1 \bm{U}_1\times_2 \bm{U}_2.$$ 
	By the orthonormality of $\bm{U}_1$ and $\bm{U}_2$, we have
	$\|\cm{H}\|_\ddagger=\|\bm{\Delta}\|_\ddagger$.  Let $\cm{H}_{(1)}=(\bm{H}_1, \dots, \bm{H}_{T_0})$, where $\bm{H}_j\in\mathbb{R}^{2r_1\times 2r_2}$ are the frontal slices. Note that 
	\[
	\langle\bm{\Delta}_{(1)}, \cm{E}{\bm{X}}^\prime\rangle=\langle\bm{U}_1 \cm{H}_{(1)}(\bm{I}_{T_0}\otimes\bm{U}_2)^\prime, \cm{E}{\bm{X}}^\prime\rangle=\langle \cm{H}_{(1)}, \bm{U}_1^\prime \cm{E}{\bm{X}}^\prime(\bm{I}_{T_0}\otimes\bm{U}_2)\rangle=\sum_{j=1}^{T_0}\langle \bm{H}_j, \bm{U}_1^\prime \cm{E}\bm{X}_j^\prime\bm{U}_2\rangle.
	\]
	Then we can show that
	\begin{align}
		\begin{split}
			\sup_{\bm{\Delta}\in\bm{\Theta}_\ddagger(2r_1,2r_2)} \langle\bm{\Delta}_{(1)}, \cm{E}{\bm{X}}^\prime\rangle 
			&\leq \sup_{\bm{U}_1\in\mathcal{O}^{N\times 2r_1}, \bm{U}_2\in\mathcal{O}^{N\times 2r_2}}\sup_{\sum_{j=1}^{T_0}\|\bm{H}_j\|_{\Fr}=1} \sum_{j=1}^{T_0}\langle \bm{H}_j, \bm{U}_1^\prime \cm{E}\bm{X}_j^\prime\bm{U}_2\rangle \\
			&=\sup_{\bm{U}_1\in\mathcal{O}^{N\times 2r_1}, \bm{U}_2\in\mathcal{O}^{N\times 2r_2}} \max_{1\leq j\leq T_0}  \|\bm{U}_1^\prime \cm{E}\bm{X}_j^\prime\bm{U}_2 \|_{\Fr} \\
			&=\max_{1\leq j\leq T_0} \sup_{\bm{U}_1\in\mathcal{O}^{N\times 2r_1}, \bm{U}_2\in\mathcal{O}^{N\times 2r_2}} \sup_{\bm{M}\in\mathcal{S}^{2r_1\times 2r_2}}  \langle\bm{M},\bm{U}_1^\prime \cm{E}\bm{X}_j^\prime\bm{U}_2 \rangle \\
			&=\max_{1\leq j\leq T_0} \sup_{\bm{U}_1\in\mathcal{O}^{N\times 2r_1}, \bm{U}_2\in\mathcal{O}^{N\times 2r_2}}\sup_{\bm{M}\in\mathcal{S}^{2r_1\times 2r_2}}  \langle \bm{U}_1\bm{M}\bm{U}_2^\prime, \cm{E}\bm{X}_j^\prime \rangle \\
			&=\max_{1\leq j\leq T_0}\sup_{\bm{M}\in\bm{\Theta}_{\Fr}(2r_{\min})} \langle \bm{M}, \cm{E}\bm{X}_j^\prime \rangle.\label{eq:discrete1}
		\end{split}
	\end{align}

	Since $\widebar{\bm{\Theta}}_{\Fr}(2r_{\min})$ is a minimal $1/2$-net for $\bm{\Theta}_{\Fr}(2r_{\min})$ in the Frobenius norm, for any $\bm{M}\in\bm{\Theta}_{\Fr}(2r_{\min})$, there exists $\widebar{\bm{M}}\in\widebar{\bm{\Theta}}_{\Fr}(2r_{\min})$ such that $\|\bm{M}-\widebar{\bm{M}}\|_{\Fr}\leq 1/2$. Notice that $\bm{M}-\widebar{\bm{M}}$ is of rank at most $4r_{\min}$. Then, we can find $\bm{T}_1, \bm{T}_2$, each with rank at most $2r_{\min}$, such that $\bm{M}-\widebar{\bm{M}} =\bm{T}_1 + \bm{T}_2$ and $\langle\bm{T}_1,\bm{T}_2\rangle = 0$. 	Moreover, it holds
	$\|\bm{T}_1\|_{\Fr}+\|\bm{T}_2\|_{\Fr}\leq \sqrt{2} \|\bm{T}_1+\bm{T}_2\|_{\Fr}= \sqrt{2}\|\bm{M}-\widebar{\bm{M}}\|_{\Fr}\leq \sqrt{2}/2$, and ${\bm{T}_i/\|\bm{T}_i\|_{\Fr}}\in\bm{\Theta}_{\Fr}(2r_{\min})$. As a result, we can show that
	\begin{align*}
		\sup_{\bm{M}\in\bm{\Theta}_{\Fr}(2r_{\min})}\langle \bm{M}, \cm{E}\bm{X}_j^\prime \rangle & = \sup_{\bm{M}\in\bm{\Theta}_{\Fr}(2r_{\min})} \left \{\langle \widebar{\bm{M}}, \cm{E}\bm{X}_j^\prime \rangle +  \sum_{i=1}^2 \langle \frac{\bm{T}_i}{\|\bm{T}_i\|_{\Fr}}, \cm{E}\bm{X}_j^\prime \rangle \|\bm{T}_i\|_{\Fr} \right \}\\
		&\leq \max_{\widebar{\bm{M}}\in\widebar{\bm{\Theta}}_{\Fr}(2r_{\min})} \langle \widebar{\bm{M}}, \cm{E}\bm{X}_j^\prime \rangle +   (\|\bm{T}_1\|_{\Fr}+\|\bm{T}_2\|_{\Fr}) \sup_{\bm{M}\in\bm{\Theta}_{\Fr}(2r_{\min})}\langle \bm{M}, \cm{E}\bm{X}_j^\prime \rangle\\
		&\leq \max_{\widebar{\bm{M}}\in\widebar{\bm{\Theta}}_{\Fr}(2r_{\min})} \langle \widebar{\bm{M}}, \cm{E}\bm{X}_j^\prime \rangle +   \frac{\sqrt{2}}{2} \sup_{\bm{M}\in\bm{\Theta}_{\Fr}(2r_{\min})}\langle \bm{M}, \cm{E}\bm{X}_j^\prime \rangle,
	\end{align*}
	which implies
	\begin{equation}\label{eq:discrete2}
		\sup_{\bm{M}\in\bm{\Theta}_{\Fr}(2r_{\min})}\langle \bm{M}, \cm{E}\bm{X}_j^\prime \rangle\leq	4  \max_{\bm{M}\in\widebar{\bm{\Theta}}_{\Fr}(2r_{\min})} \langle \bm{M}, \cm{E}\bm{X}_j^\prime \rangle.
	\end{equation}
	By combining \eqref{eq:discrete1} and \eqref{eq:discrete2}, we accomplish the proof of this lemma.
\end{proof}

%%%%%%%%%%%%%%%%%%%%%%%%%%%%%%%%%%%%%%%%%%%%%%%%%%%%%%%%%%%%%%%%%%%%%%%%%%%%%%
Now we are ready to prove Lemma \ref{lemma:devbd1}.

\begin{proof}[Proof of Lemma \ref{lemma:devbd1}]
	Let $\widebar{\bm{\Theta}}_{\Fr}(2r_{\min})$ be a minimal $1/2$-net of $\bm{\Theta}_{\Fr}(2r_{\min})$ in the Frobenius norm, where $r_{\min}=r_1\wedge r_2$. Then its cardinality satisfies
	\begin{equation}\label{eq:card2}
		\log |\widebar{\bm{\Theta}}_{\Fr}(2r_{\min})|\leq (4N+2)r_{\min} \log 18 \leq  18 r_{\min} N;
	\end{equation}
	see Lemma 3.1 of \cite{candes2011tight}.
	Denote $\widecheck{D}_{\mathcal{M}} = 18 r_{\min} N + \log T_0$. By \eqref{eq:card2} and Lemma \ref{lemma:discretize}, for any $K>0$,  we have
	\begin{align}\label{eq:discretizeE}
		\begin{split}
			&\mathbb{P}\left(\sup_{\bm{\Delta}\in\bm{\Theta}_\ddagger(2r_1,2r_2)}\langle\bm{\Delta}_{(1)}, \frac{1}{T_1}\cm{E}\bm{X}^\prime\rangle \geq K \right)\\
			&\hspace{5mm}\leq \mathbb{P}\left(\max_{1\leq k\leq T_0}\max_{\bm{M}\in\widebar{\bm{\Theta}}_{\Fr}(2r_{\min})}\frac{1}{T_1}\langle \bm{M}, \cm{E}\bm{X}_k^\prime \rangle \geq K/4 \right) \\
			&\hspace{5mm}\leq T_0 |\widebar{\bm{\Theta}}_{\Fr}(2r_{\min})| \max_{1\leq k\leq T_0}\max_{\|\bm{M}\|_{\Fr}=1} \mathbb{P}\left(\frac{1}{T_1}\langle \bm{M}, \cm{E}\bm{X}_k^\prime \rangle \geq K/4 \right) \\
			&\hspace{5mm}\leq \exp(\widecheck{D}_{\mathcal{M}}) \max_{1\leq k\leq T_0}\max_{\|\bm{M}\|_{\Fr}=1} \mathbb{P}\left(\frac{1}{T_1}\langle \bm{M}, \cm{E}\bm{X}_k^\prime \rangle \geq K/4 \right).
		\end{split}
	\end{align}
	
	Since $\bm{\varepsilon}_t=\bm{\Sigma}_\varepsilon^{1/2}\bm{\xi}_{t}$ and $\mathbf{y}_t = \sum_{j=0}^{\infty}\bm{\Psi}_j^*\bm{\Sigma}_\varepsilon^{1/2}\bm{\xi}_{t-j}$, for any $1\leq k\leq T_0$, we have
	\begin{align*}
		\langle \bm{M}, \cm{E}\bm{X}_k^\prime \rangle 
		=\sum_{t=T_0+1}^{T}\langle \bm{M}, \bm{\varepsilon}_t\mathbf{y}_{t-k}^\prime \rangle
		&=\sum_{t=T_0+1}^{T} \left \langle \bm{M}, \bm{\Sigma}_\varepsilon^{1/2}\bm{\xi}_{t}\sum_{j=0}^{\infty}\bm{\xi}_{t-k-j}^\prime \bm{\Sigma}_\varepsilon^{1/2} \bm{\Psi}_j^{*\prime}   \right \rangle  \notag\\
		&=\sum_{j=0}^{\infty} \sum_{t=T_0+1}^{T} \langle  \bm{M},\bm{\Sigma}_\varepsilon^{1/2}\bm{\xi}_{t}\bm{\xi}_{t-k-j}^\prime \bm{\Sigma}_\varepsilon^{1/2} \bm{\Psi}_j^{*\prime}    \rangle \notag\\
		&=\sum_{j=0}^{\infty}  \sum_{t=T_0+1}^{T} \langle  \widetilde{\bm{M}}_j  \bm{\xi}_{t-k-j}, \bm{\xi}_{t} \rangle,
	\end{align*}
	where 
	\[
	\widetilde{\bm{M}}_j =\bm{\Sigma}_\varepsilon^{1/2} \bm{M} \bm{\Psi}_j^{*} \bm{\Sigma}_\varepsilon^{1/2}.
	\]
	Note that for any fixed $k$ and $j$ and any $\delta_j\in(0,1]$, by Lemma \ref{lemma:xi_dev}(i), we have
	\begin{equation*}%\label{eq:xi_dev1}
		\mathbb{P}\left\{ \sum_{t=T_0+1}^{T} \langle \widetilde{\bm{M}}_j\bm{\xi}_{t-k-j}, \bm{\xi}_{t} \rangle \geq C\sigma^2 \|\widetilde{\bm{M}}_j\|_{\Fr} \left \{\log(1/\delta_j) + \sqrt{T_1\log(1/\delta_j)} \right \} \right\}\leq 2\delta_j.
	\end{equation*}	
	For simplicity, denote
	\[
	a_j=\sigma^2 \|\widetilde{\bm{M}}_j\|_{\Fr} \left \{\log(1/\delta_j) + \sqrt{T_1 \log(1/\delta_j)} \right \}. 
	\]
	Then it follows that
	\begin{equation}\label{eq:devbd9}
		\mathbb{P}\left( \frac{1}{T_1} \langle \bm{M}, \cm{E}\bm{X}_k^\prime \rangle \geq  \frac{C}{T_1} \sum_{j=0}^{\infty} a_{j}	\right)
		\leq  2\sum_{j=0}^{\infty} \delta_{j}.
	\end{equation}
	
	Moreover, by Assumption \ref{assum:Adecay}, if $\|\bm{M}\|_{\Fr}=1$, then 
	$\|\widetilde{\bm{M}}_{j}\|_{\Fr} \leq  \lambda_{\max}(\bm{\Sigma}_\varepsilon)  \|\bm{M}\|_{\Fr} \|\bm{\Psi}_j^*\|_{\op} \leq C  \rho^{j}$. 
	By choosing
	\[
	\delta_{j} = \exp\left \{-4\rho^{-(j+1)/2} \widecheck{D}_{\mathcal{M}} /\log(1/\rho) \right \},
	\]
	i.e., $\log(1/\delta_{j})= 4\rho^{-(j+1)/2} \widecheck{D}_{\mathcal{M}}/\log(1/\rho)$, we can show that
	\begin{align}\label{eq:devbd10}
		\frac{1}{T_1}\sum_{j=0}^{\infty} a_{j}
		& \leq  \frac{4\sigma^2  }{T_1}  \sum_{j=0}^{\infty}  \left \{ \frac{\rho^{(j-1)/2} \widecheck{D}_{\mathcal{M}}}{\log(1/\rho)} + \rho^{(3j-1)/4}  \sqrt{\frac{ T_1 \widecheck{D}_{\mathcal{M}}}{\log(1/\rho)}}  \right \}\notag\\
		& = 4 \sigma^2    \left \{\frac{1}{\sqrt{\rho}-\rho} \cdot \frac{\widecheck{D}_{\mathcal{M}}}{T_1\log(1/\rho)}
		+\frac{1}{\rho^{1/4}-\rho} \cdot \sqrt{\frac{\widecheck{D}_{\mathcal{M}}}{T_1\log(1/\rho)}}
		\right \} \notag\\
		& \leq  C  \sqrt{\frac{\widecheck{D}_{\mathcal{M}}}{T_1}}.
	\end{align}
	In addition, using the inequality $x^k \geq k\log x$ for any $k\geq 0$ and $x>1$, we can show that
	\begin{equation}\label{eq:devbd11}
		\sum_{j=0}^{\infty}\delta_{j} \leq  \sum_{j=0}^{\infty}e^{- 2(j+1)\widecheck{D}_{\mathcal{M}}}
		= \frac{e^{ -  2\widecheck{D}_{\mathcal{M}}} }{1- e^{- 2\widecheck{D}_{\mathcal{M}}}}.
	\end{equation}
	
	Combining \eqref{eq:devbd9}--\eqref{eq:devbd11}, if $\|\bm{M}\|_{\Fr}=1$, for any $1\leq k\leq P$, we  have
	\[
	\mathbb{P}\left\{ \frac{1}{T_1} \langle \bm{M}, \cm{E}\bm{X}_k^\prime \rangle \geq  C  \sqrt{\frac{\widecheck{D}_{\mathcal{M}}}{T_1}}
	\right\}\leq  \frac{2e^{ - 2\widecheck{D}_{\mathcal{M}}} }{1- e^{- 2\widecheck{D}_{\mathcal{M}}}} \leq C e^{- \widecheck{D}_{\mathcal{M}} },
	\]
	which together with \eqref{eq:discretizeE} and $T_1\gtrsim \widecheck{D}_{\mathcal{M}}$ implies the result of Lemma \ref{lemma:devbd1}.
\end{proof}

%%%%%%%%%%%%%%%%%%%%%%%%%%%%%%%%%%%%%%%%%%%%%%%%%%%%%%%%%%%%%%%%%%%%%%%%%%%%%%
\paragraph{Proof of Lemma \ref{lemma:devbd2}}
Note that $\bm{R}\bm{X}^\prime=(\bm{R}\bm{X}_1^\prime, \dots, \bm{R}\bm{X}_{T_0}^\prime)$. Along the lines of Lemma \ref{lemma:discretize} in Section \ref{sec:devbd1}, we can show the following discretization result:
\begin{equation*}
	\sup_{\bm{\Delta}\in\bm{\Theta}_\ddagger(2r_1,2r_2)} \langle\bm{\Delta}_{(1)}, \bm{R} \bm{X}^\prime\rangle 
	\leq 
	4\max_{1\leq k\leq T_0}\max_{\bm{M}\in\widebar{\bm{\Theta}}_2(2r_{\min})}\langle \bm{M}, \bm{R}\bm{X}_k^\prime \rangle.
\end{equation*}
Then, similarly to \eqref{eq:discretizeE}, for any $K>0$, by \eqref{eq:card2}, we have
\begin{align}\label{eq:discretizeR}
	&\mathbb{P}\left(\sup_{\bm{\Delta}\in\bm{\Theta}_\ddagger(2r_1,2r_2)}\langle\bm{\Delta}_{(1)}, \frac{1}{T_1}\bm{R}\bm{X}^\prime\rangle \geq K \right)\notag\\
	&\hspace{5mm}\leq \exp(\widecheck{D}_{\mathcal{M}}) \max_{1\leq k\leq T_0}\max_{\|\bm{M}\|_{\Fr}=1} \mathbb{P}\left(\frac{1}{T_1}\langle \bm{M}, \bm{R}\bm{X}_k^\prime \rangle \geq K/4 \right),
\end{align}
where $\widecheck{D}_{\mathcal{M}} = 18 (r_1 \wedge r_2) N + \log T_0$.

Since $\mathbf{y}_t = \sum_{j=0}^{\infty}\bm{\Psi}_j^*\bm{\Sigma}_\varepsilon^{1/2}\bm{\xi}_{t-j}$, for any $1\leq k\leq T_0$, we have
\begin{align}\label{eq:devbd3}
	\langle \bm{M}, \bm{R}\bm{X}_k^\prime \notag\rangle 
	&=\sum_{t=T_0+1}^{T}\langle \bm{M}, \bm{r}_t\mathbf{y}_{t-k}^\prime \rangle\\
	&=\sum_{t=T_0+1}^{T} \left \langle \bm{M},  \sum_{j=T_0+1}^{\infty}\bm{A}_j^* \sum_{i=0}^{\infty}\bm{\Psi}_i^*\bm{\Sigma}_\varepsilon^{1/2}\bm{\xi}_{t-j-i}\sum_{\ell=0}^{\infty}\bm{\xi}_{t-k-\ell}^\prime \bm{\Sigma}_\varepsilon^{1/2} \bm{\Psi}_\ell^{*\prime} \right  \rangle  \notag\\
	&=\sum_{j=T_0+1}^{\infty} \sum_{i=0}^{\infty} \sum_{\ell=0}^{\infty}  \sum_{t=T_0+1}^{T} \langle  \bm{M},\bm{A}_j^*\bm{\Psi}_i^*\bm{\Sigma}_\varepsilon^{1/2}\bm{\xi}_{t-j-i}\bm{\xi}_{t-k-\ell}^\prime \bm{\Sigma}_\varepsilon^{1/2}\bm{\Psi}_\ell^{*\prime} \rangle \notag\\
	&=\sum_{j=T_0+1}^{\infty} \sum_{i=0}^{\infty} \sum_{\ell=0}^{\infty}  \sum_{t=T_0+1}^{T} \langle  \widetilde{\bm{M}}_{i,j,\ell}  \bm{\xi}_{t-k-\ell}, \bm{\xi}_{t-j-i} \rangle := B_1+B_2,
\end{align}
where 
\begin{align*}
	B_1&=\sum_{j=T_0+1}^{\infty} \sum_{i=0}^{\infty} \sum_{\ell=0}^{\infty}  \sum_{t=T_0+1}^{T} \langle  \widetilde{\bm{M}}_{i,j,\ell}  \bm{\xi}_{t-k-\ell}, \bm{\xi}_{t-j-i} \rangle I\{\ell\neq i+j-k\},\\
	B_2&=\sum_{j=T_0+1}^{\infty} \sum_{i=0}^{\infty} \sum_{\ell=0}^{\infty}  \sum_{t=T_0+1}^{T} \langle  \widetilde{\bm{M}}_{i,j,\ell}  \bm{\xi}_{t-k-\ell}, \bm{\xi}_{t-j-i} \rangle I\{\ell= i+j-k\},
\end{align*}
with $I(\cdot)$ being the indicator function, and
\[
\widetilde{\bm{M}}_{i,j,\ell}=\bm{\Sigma}_\varepsilon^{1/2} \bm{\Psi}_i^{*\prime} \bm{A}_j^{*\prime}\bm{M}\bm{\Psi}_\ell^*\bm{\Sigma}_\varepsilon^{1/2}.
\]
Note that for any fixed $i,j,\ell$, we can write $\sum_{t=T_0+1}^{T}\langle  \widetilde{\bm{M}}_{i,j,\ell}  \bm{\xi}_{t-k-\ell}, \bm{\xi}_{t-j-i} \rangle = \sum_{t=T_0'}^{T_1'} \langle \widetilde{\bm{M}}_{i,j,\ell}\bm{\xi}_{t-s}, \bm{\xi}_{t} \rangle$, with $T_0'=T_0+1-j-i$, $T_1'=T-j-i$ and $s=k+\ell-j-i$ (cf. Lemma \ref{lemma:xi_dev}). In addition, note that $T_1'-T_0'+1=T_1$.

Then, by Lemma \ref{lemma:xi_dev} and a method similar to that for \eqref{eq:devbd9}, for any fixed $k$, we can show that
\[
\mathbb{P}\left( B_1 \geq  \sum_{j=T_0+1}^{\infty} \sum_{i=0}^{\infty} \sum_{\ell=0}^{\infty}  C b_{i,j,\ell} I\{\ell\neq i+j-k\} \right)
\leq  \sum_{j=T_0+1}^{\infty}\sum_{i=0}^{\infty} \sum_{\ell=0}^{\infty} 2\delta_{i,j,\ell} I\{\ell\neq i+j-k\}
\]
and
\begin{align*}
	&\mathbb{P}\left\{ B_2 \geq  \sum_{j=T_0+1}^{\infty} \sum_{i=0}^{\infty} \sum_{\ell=0}^{\infty}
	(C b_{i,j,\ell} + c_{i,j,\ell} ) I\{\ell=i+j-k\}
	\right\}
	\leq  \sum_{j=T_0+1}^{\infty}\sum_{i=0}^{\infty} \sum_{\ell=0}^{\infty} \delta_{i,j,\ell} I\{\ell=i+j-k\},
\end{align*}
where $\delta_{i,j,\ell}\in(0,1)$ will be specified shortly (see \eqref{eq:deltaijl} below), and
\[
b_{i,j,\ell}=\sigma^2 \|\widetilde{\bm{M}}_{i,j,\ell}\|_{\Fr} \left \{\log(1/\delta_{i,j,\ell}) + \sqrt{T_1 \log(1/\delta_{i,j,\ell})} \right \}  \quad\text{and}\quad
c_{i,j,\ell}=T_1 \sqrt{N}\|\widetilde{\bm{M}}_{i,j,\ell}\|_{\Fr}.
\]
Combining the above inequalities with \eqref{eq:devbd3}, we have
\begin{align}\label{eq:devbd7}
	&\mathbb{P}\left\{ \frac{1}{T_1} \langle \bm{M}, \bm{R}\bm{X}_k^\prime \rangle \geq  \frac{C}{T_1} \sum_{j=T_0+1}^{\infty} \sum_{i=0}^{\infty} \sum_{\ell=0}^{\infty}
	b_{i,j,\ell} 
	+ \frac{1}{T_1}\sum_{j=T_0+1}^{\infty} \sum_{i=0}^{\infty} \sum_{\ell=0}^{\infty}
	c_{i,j,\ell} I\{\ell=i+j-k\}
	\right\} \notag \\
	&\hspace{5mm}\leq  2\sum_{j=T_0+1}^{\infty}\sum_{i=0}^{\infty} \sum_{\ell=0}^{\infty} \delta_{i,j,\ell}.
\end{align}

Note that by Assumption \ref{assum:Adecay}, if $\|\bm{M}\|_{\Fr}=1$, we can upper bound each $\|\widetilde{\bm{M}}_{i,j,\ell}\|_{\Fr}$ as follows:
\begin{align*}
	\|\widetilde{\bm{M}}_{i,j,\ell}\|_{\Fr} &\leq  \lambda_{\max}(\bm{\Sigma}_\varepsilon) \|\bm{A}_j^{*}\|_{\op}  \|\bm{\Psi}_i^*\|_{\op} \|\bm{\Psi}_\ell^*\|_{\op} \|\bm{M}\|_{\Fr}
	\leq C \rho^{i+j+\ell}.
\end{align*}
Then, for any $1\leq k\leq T_0$,  we have
\begin{align}\label{eq:devbd4}
	\sum_{j=T_0+1}^{\infty} \sum_{i=0}^{\infty} \sum_{\ell=0}^{\infty}
	c_{i,j,\ell} I\{\ell=i+j-k\} 
	&\leq C T_1 \sqrt{N}  \sum_{j=T_0+1}^{\infty} \sum_{i=0}^{\infty}  \rho^{2i+2j-T_0} \notag\\
	&= \frac{C \rho^2}{(1-\rho^2)^2}   \sqrt{N} \rho^{T_0} T_1 \notag\\
	&\leq C  \sqrt{N},
\end{align}
where the last inequality follows from Assumption \ref{assum:T0}.

Moreover, by choosing
\begin{equation}\label{eq:deltaijl}
	\delta_{i,j,\ell} = \exp\left \{-2\rho^{-(i+j+\ell)/2} \widecheck{D}_{\mathcal{M}} /\log(1/\rho) \right \},
\end{equation}
i.e., $\log(1/\delta_{i,j,\ell})= 2\rho^{-(i+j+\ell)/2} \widecheck{D}_{\mathcal{M}}/\log(1/\rho)$,  we can similarly show that
\begin{align}\label{eq:devbd5}
	\sum_{j=T_0+1}^{\infty} \sum_{i=0}^{\infty} \sum_{\ell=0}^{\infty}
	b_{i,j,\ell}
	& \leq  C \sigma^2   \sum_{j=T_0+1}^{\infty} \sum_{i=0}^{\infty} \sum_{\ell=0}^{\infty}  \left \{ \frac{\rho^{(i+j+\ell)/2} \widecheck{D}_{\mathcal{M}}}{\log(1/\rho)} + \rho^{3(i+j+\ell)/4}  \sqrt{\frac{T_1 \widecheck{D}_{\mathcal{M}}}{\log(1/\rho)}}  \right \}\notag\\	
	& = C \sigma^2 
	\left \{\frac{ \sqrt{\rho}}{(1-\sqrt{\rho})^3} \cdot \frac{\rho^{T_0/2}\widecheck{D}_{\mathcal{M}}}{\log(1/\rho)}
	+\frac{\rho^{3/4} }{(1-\rho^{3/4})^3} \cdot \sqrt{ \frac{\rho^{3T_0/2} T_1 \widecheck{D}_{\mathcal{M}}}{\log(1/\rho)}}
	\right \} \notag\\
	&\leq C \left ( \frac{\widecheck{D}_{\mathcal{M}}}{T_1} + \frac{\sqrt{\widecheck{D}_{\mathcal{M}}}}{T_1} \right )\notag\\
	&\leq C \frac{\widecheck{D}_{\mathcal{M}}}{T_1},
\end{align}
where Assumption \ref{assum:T0} is used  in the second to last inequality.

In addition, similarly to \eqref{eq:devbd11}, we can show that
\begin{align}\label{eq:devbd6}
	\sum_{j=T_0+1}^{\infty}\sum_{i=0}^{\infty} \sum_{\ell=0}^{\infty} \delta_{i,j,\ell} &\leq  \sum_{j=T_0+1}^{\infty}\sum_{i=0}^{\infty} \sum_{\ell=0}^{\infty} \exp\left \{- (i+j+\ell) \widecheck{D}_{\mathcal{M}} \right \}
	= \frac{e^{ - (T_0+1)  \widecheck{D}_{\mathcal{M}}} }{\left(1- e^{- \widecheck{D}_{\mathcal{M}}} \right)^3 }.
\end{align}
Combining \eqref{eq:devbd7}, \eqref{eq:devbd4}, \eqref{eq:devbd5} and \eqref{eq:devbd6}, if $\|\bm{M}\|_{\Fr}=1$, for any $1\leq k\leq T_0$, we  have
\[
\mathbb{P}\left\{ \frac{1}{T_1} \langle \bm{M}, \bm{R}\bm{X}_k^\prime \rangle \geq  C  \frac{ \sqrt{N} }{T_1} \right\}\leq \frac{2e^{ - (T_0+1)  \widecheck{D}_{\mathcal{M}}} }{\left(1- e^{- \widecheck{D}_{\mathcal{M}}} \right)^3 } \leq C e^{- \widecheck{D}_{\mathcal{M}}}.
\]
Combining this with \eqref{eq:discretizeR} and $T_1\gtrsim \widecheck{D}_{\mathcal{M}}$, we accomplish the proof of this lemma.

%%%%%%%%%%%%%%%%%%%%%%%%%%%%%%%%%%%%%%%%%%%%%%%%%%%%%%%%%%%%%%%%%%%%%%%%%%%%%%
\paragraph{Proof of Lemma \ref{lemma:RE}}	

The proof of Lemma \ref{lemma:RE} relies on the following result.

%%%%%%%%%%%%%%%%%%%%%%%%%%%%%%%%%%%%%%%%%%%%%%%%%%%%%%%%%%%%%%%%%%%%%%%%%%%%%%
\begin{lemma}\label{lemma:RE1}
	Suppose that Assumptions \ref{assum:glp} -- \ref{assum:error} hold.  If $T_1\gtrsim N + \log T_0$, then
	\[
	\mathbb{P}\left\{ \max_{1\leq i\leq T_0}\max_{1\leq j\leq T_0} \left \|\frac{\bm{X}_i\bm{X}_j^\prime}{T_1}-\bm{\Gamma}(i-j)\right \|_{\op} \geq \tau^2
	\right \}  
	\leq  C e^{ - N - \log T_0 },
	\]	
	where $\tau^2=C\sqrt{(N+\log T_0)/T_1}$.
\end{lemma}
\begin{proof}
	Let $\widebar{\mathcal{S}}^{N-1}$ be a minimal $1/4$-net of $\mathcal{S}^{N-1}$ in the Euclidean norm. By \cite{vershynin2010introduction}, its cardinality satisfies $|\widebar{\mathcal{S}}^{N-1}|\leq 9^N$, and for any fixed $1\leq i, j \leq T_0$,
	\[
	\left \|\frac{\bm{X}_i\bm{X}_j^\prime}{T_1}-\bm{\Gamma}(i-j)\right \|_{\op}  \leq 2 \max_{\bm{u}\in \widebar{\mathcal{S}}^{N-1}} \left | \bm{u}^\prime \left \{ \frac{\bm{X}_i\bm{X}_j^\prime}{T_1}-\bm{\Gamma}(i-j)\right \} \bm{u}\right |.
	\]
	Denote $\widetilde{D}_{\mathcal{M}}=N\log 9+2\log T_0$. Hence, for any $K>0$, we have
	\begin{align}\label{eq:RE3}
		&\mathbb{P}\left\{ \max_{1\leq i\leq T_0}\max_{1\leq j\leq T_0}  \left \|\frac{\bm{X}_i\bm{X}_j^\prime}{T_1}-\bm{\Gamma}(i-j)\right \|_{\op}  \geq K \right\} \notag\\
		&\hspace{5mm} \leq T_0^2 \max_{1\leq i\leq T_0}\max_{1\leq j\leq T_0} \mathbb{P}\left\{ \left \|\frac{\bm{X}_i\bm{X}_j^\prime}{T_1}-\bm{\Gamma}(i-j)\right \|_{\op} \geq K \right \} \notag\\
		&\hspace{5mm} \leq T_0^2 |\widebar{\mathcal{S}}^{N-1}| \max_{1\leq i\leq T_0}\max_{1\leq j\leq T_0} \max_{\|\bm{u}\|_2=1} \mathbb{P}\left[ \left | \bm{u}^\prime \left \{ \frac{\bm{X}_i\bm{X}_j^\prime}{T_1}-\bm{\Gamma}(i-j)\right \} \bm{u}\right | \geq K/2 \right] \notag\\
		&\hspace{5mm} \leq \exp(\widetilde{D}_{{\mathcal{M}}})\max_{1\leq i\leq T_0}\max_{1\leq j\leq T_0} \max_{\|\bm{u}\|_2=1} \mathbb{P}\left[ \left | \bm{u}^\prime \left \{ \frac{\bm{X}_i\bm{X}_j^\prime}{T_1}-\bm{\Gamma}(i-j)\right \} \bm{u}\right | \geq K/2 \right].
	\end{align}
	
	Since $\mathbf{y}_t = \sum_{j=0}^{\infty}\bm{\Psi}_j^*\bm{\Sigma}_\varepsilon^{1/2}\bm{\xi}_{t-j}$, for any $1\leq i,j\leq T_0$, we have
	\begin{align}\label{eq:RE4}
		\frac{\bm{u}^\prime\bm{X}_i\bm{X}_j^\prime \bm{u}}{T_1}	
		=\sum_{t=T_0+1}^{T} \frac{\bm{u}^\prime\mathbf{y}_{t-i}\mathbf{y}_{t-j}^\prime \bm{u}}{T_1} &=\frac{1}{T_1}\sum_{k=0}^{\infty}\sum_{\ell=0}^{\infty} \sum_{t=T_0+1}^{T}  \bm{u}^\prime\bm{\Psi}_k^*\bm{\Sigma}_\varepsilon^{1/2}\bm{\xi}_{t-i-k}\bm{\xi}_{t-j-\ell}^\prime \bm{\Sigma}_\varepsilon^{1/2}\bm{\Psi}_\ell^{*\prime} \bm{u} \notag\\
		&=\frac{1}{T_1}\sum_{k=0}^{\infty}\sum_{\ell=0}^{\infty} \sum_{t=T_0+1}^{T} \langle  \widetilde{\bm{M}}_{k,\ell} \bm{\xi}_{t-j-\ell}, \bm{\xi}_{t-i-k} \rangle:=G_1+G_2,
	\end{align}
	where 
	\begin{align*}
		G_1&=\frac{1}{T_1}\sum_{k=0}^{\infty}\sum_{\ell=0}^{\infty} \sum_{t=T_0+1}^{T} \langle  \widetilde{\bm{M}}_{k,\ell} \bm{\xi}_{t-j-\ell}, \bm{\xi}_{t-i-k} \rangle I\{\ell\neq i+k-j\},\\
		G_2&=\frac{1}{T_1}\sum_{k=0}^{\infty}\sum_{\ell=0}^{\infty} \sum_{t=T_0+1}^{T} \langle  \widetilde{\bm{M}}_{k,\ell} \bm{\xi}_{t-j-\ell}, \bm{\xi}_{t-i-k} \rangle I\{\ell= i+k-j\},
		%B_2&=\sum_{j=P+1}^{\infty} \sum_{i=0}^{\infty} \sum_{t=P+1}^{T} \langle  \widetilde{\bm{M}}_{i,j,i+j-k}  \bm{\xi}_{t-j-i}, \bm{\xi}_{t-j-i} \rangle,
	\end{align*}
	with $I(\cdot)$ being the indicator function, and
	\[
	\widetilde{\bm{M}}_{k,\ell} =\bm{\Sigma}_\varepsilon^{1/2} \bm{\Psi}_k^{*\prime} \bm{u} \bm{u}^\prime  \bm{\Psi}_\ell^{*}\bm{\Sigma}_\varepsilon^{1/2}.
	\]
	By Lemma \ref{lemma:xi_dev} and a method similar to the proof of Lemma \ref{lemma:devbd2}, we can show that 
	\[
	\mathbb{P}\left( |G_1| \geq \frac{C}{T_1} \sum_{k=0}^{\infty}\sum_{\ell=0}^{\infty}    b_{k,\ell} I\{\ell\neq i+k-j\} \right)
	\leq  \sum_{k=0}^{\infty}\sum_{\ell=0}^{\infty} 4\delta_{k,\ell} I\{\ell\neq i+k-j\}
	\]
	and
	\begin{align*}
		&\mathbb{P}\left\{ |G_2 - \mathbb{E}(G_2)| \geq  \frac{C}{T_1} \sum_{k=0}^{\infty}\sum_{\ell=0}^{\infty} 
		b_{k,\ell} I\{\ell=i+k-j\}
		\right\}
		\leq  \sum_{k=0}^{\infty}\sum_{\ell=0}^{\infty} 2 \delta_{k,\ell} I\{\ell=i+k-j\},
	\end{align*}
	where $\delta_{k,\ell}\in(0,1)$ will be specified below, and
	\[
	b_{k,\ell}=\sigma^2 \|\widetilde{\bm{M}}_{k,\ell}\|_{\Fr} \left \{\log(1/\delta_{k,\ell}) + \sqrt{T_1 \log(1/\delta_{k,\ell})} \right \}.
	\]
	Since $\mathbb{E}(G_1)=0$, combining the above results with \eqref{eq:RE4}, we have
	\begin{align}\label{eq:RE6} 
		&\mathbb{P}\left\{ \left | \bm{u}^\prime \left \{ \frac{\bm{X}_i\bm{X}_j^\prime}{T_1}-\bm{\Gamma}(i-j)\right \} \bm{u}\right | \geq  \frac{C}{T_1} \sum_{k=0}^{\infty}\sum_{\ell=0}^{\infty}
		b_{k,\ell}\right\} \notag \\
		&\hspace{5mm}\leq \mathbb{P}\left\{ |G_1| +|G_2 - \mathbb{E}(G_2)| \geq  \frac{C}{T_1} \sum_{k=0}^{\infty}\sum_{\ell=0}^{\infty}
		b_{k,\ell}\right\}  \notag \\
		&\hspace{5mm}\leq 4 \sum_{k=0}^{\infty}\sum_{\ell=0}^{\infty} \delta_{k,\ell}.
	\end{align}
	
	Note that by Assumption \ref{assum:Adecay}, if $\|\bm{u}\|_{2}=1$, we have
	\begin{align*}
		\|\widetilde{\bm{M}}_{k,\ell}\|_{\Fr} &\leq  \lambda_{\max}(\bm{\Sigma}_\varepsilon) \|\bm{\Psi}_k^*\|_{\op} \|\bm{\Psi}_\ell^*\|_{\op} 
		\leq C  \rho^{k+\ell}.
	\end{align*}
	Then, by choosing
	\begin{equation*}
		\delta_{k,\ell} = \exp\left \{-4\rho^{-(k+\ell+1)/2} \widetilde{D}_{{\mathcal{M}}} /\log(1/\rho) \right \}.
	\end{equation*}
	i.e., $\log(1/\delta_{k,\ell})= 4\rho^{-(k+\ell+1)/2} \widetilde{D}_{{\mathcal{M}}}/\log(1/\rho)$,  we have
	\begin{align}\label{eq:RE7}
		\frac{1}{T_1}\sum_{k=0}^{\infty}\sum_{\ell=0}^{\infty} b_{k,\ell}
		& \leq  \frac{C\sigma^2  }{T_1} \sum_{k=0}^{\infty}\sum_{\ell=0}^{\infty}  \left \{ \frac{\rho^{(k+\ell-1)/2} \widetilde{D}_{{\mathcal{M}}}}{\log(1/\rho)} + \rho^{(3k+3\ell-1)/4}  \sqrt{\frac{T_1 \widetilde{D}_{{\mathcal{M}}}}{\log(1/\rho)}}  \right \}\notag\\	
		& = C \sigma^2  \left \{\frac{ 1}{\sqrt{\rho}(1-\sqrt{\rho})^2} \cdot \frac{\widetilde{D}_{{\mathcal{M}}}}{T_1\log(1/\rho)}
		+\frac{1}{\rho^{1/4}(1-\rho^{3/4})^2} \cdot \sqrt{ \frac{ \widetilde{D}_{{\mathcal{M}}}}{T_1\log(1/\rho)}}
		\right \} \notag\\
		& \leq C \sqrt{\frac{\widetilde{D}_{{\mathcal{M}}}}{T_1}}.
	\end{align}
	In addition, similarly to the proof of Lemma \ref{lemma:devbd1}, we can show that
	\begin{equation}\label{eq:RE8}
		\sum_{k=0}^{\infty}\sum_{\ell=0}^{\infty} \delta_{k,\ell} \leq  \sum_{k=0}^{\infty}\sum_{\ell=0}^{\infty} e^{- 2(k+\ell+1)\widetilde{D}_{\widecheck{\mathcal{M}}}}
		= \frac{e^{ -  2\widetilde{D}_{{\mathcal{M}}}} }{(1- e^{- 2\widetilde{D}_{{\mathcal{M}}}})^2}.
	\end{equation}
	
	By \eqref{eq:RE6}--\eqref{eq:RE8}, if $\|\bm{u}\|_{2}=1$, for any $1\leq i,j\leq T_0$, we  have
	\[
	\mathbb{P}\left\{ \left | \bm{u}^\prime \left \{ \frac{\bm{X}_i\bm{X}_j^\prime}{T_1}-\bm{\Gamma}(i-j)\right \} \bm{u}\right | \geq  C  \sqrt{\frac{\widetilde{D}_{{\mathcal{M}}}}{T_1}}
	\right\}\leq  \frac{4e^{ - 2\widetilde{D}_{{\mathcal{M}}}} }{(1- e^{- 2\widetilde{D}_{{\mathcal{M}}}})^2} \leq C e^{- \widetilde{D}_{{\mathcal{M}}}}.
	\]
	Combining this with \eqref{eq:RE3} and $T_1\gtrsim \widetilde{D}_{{\mathcal{M}}}$, we accomplish the proof of this lemma.
\end{proof}

Now we are ready to prove Lemma \ref{lemma:RE}.

\begin{proof}[Proof of Lemma \ref{lemma:RE}]
	By \cite{basu2015regularized}, we have $\sigma_{\min}(\bm{\Sigma}_{T_0}) \geq \lambda_{\min}(\bm{\Sigma}_\varepsilon)\mu_{\min}(\bm{\Psi}_*)=\kappa_{\mathrm{RSC}}$, and hence
	\begin{equation} \label{eq:exp-min}
		\frac{\mathbb{E}\left (\|\bm{\Delta}_{(1)}\bm{X}\|_{\Fr}^2 \right )}{T_1}
		=\trace \left (\bm{\Delta}_{(1)} \bm{\Sigma}_{T_0} \bm{\Delta}_{(1)}^\prime \right )
		\geq \sigma_{\min}(\bm{\Sigma}_{T_0})  \|\bm{\Delta}\|_{\Fr}^2 \geq \kappa_{\mathrm{RSC}}  \|\bm{\Delta}\|_{\Fr}^2.
	\end{equation}
	
	Moreover, observe that 
	\begin{align*}
		\frac{\left |\|\bm{\Delta}_{(1)}\bm{X}\|_{\Fr}^2 -  \mathbb{E}\left (\|\bm{\Delta}_{(1)}\bm{X}\|_{\Fr}^2 \right )\right |}{T_1} 
		& = \trace \left \{\bm{\Delta}_{(1)} \left (\frac{\bm{X}\bm{X}^\prime}{T_1} - \bm{\Sigma}_{T_0}  \right ) \bm{\Delta}_{(1)}^\prime \right \} \\
		& = \sum_{i=1}^{T_0}\sum_{j=1}^{T_0} \trace \left [\bm{\Delta}_{i} \left \{\frac{\bm{X}_i\bm{X}_j^\prime}{T_1} -\bm{\Gamma}(i-j)\right \} \bm{\Delta}_{j}^\prime \right ]\\
		&\leq \sum_{i=1}^{T_0}\sum_{j=1}^{T_0} \|\bm{\Delta}_{i}\|_{\Fr}\|\bm{\Delta}_{j}\|_{\Fr} \left \|\frac{\bm{X}_i\bm{X}_j^\prime}{T_1} -\bm{\Gamma}(i-j)\right \|_{\op}\\
		&\leq\|\bm{\Delta}\|_{\ddagger}^2 \max_{1\leq i\leq T_0}\max_{1\leq j\leq T_0} \left \|\frac{\bm{X}_i\bm{X}_j^\prime}{T_1}-\bm{\Gamma}(i-j)\right \|_{\op}.
	\end{align*}

	As a result, for any $\bm{\Delta}\in\mathbb{R}^{N\times N\times T_0}$, we have
	\begin{align*}
		\frac{1}{T_1}\|\bm{\Delta}_{(1)}\bm{X}\|_{\Fr}^2 &\geq  \frac{\mathbb{E}\left (\|\bm{\Delta}_{(1)}\bm{X}\|_{\Fr}^2 \right ) }{T_1} - \frac{\left |\|\bm{\Delta}_{(1)}\bm{X}\|_{\Fr}^2 -  \mathbb{E}\left (\|\bm{\Delta}_{(1)}\bm{X}\|_{\Fr}^2 \right )\right |}{T_1} \notag \\
		&\geq \kappa_{\mathrm{RSC}} \|\bm{\Delta}\|_{\Fr}^2 - \|\bm{\Delta}\|_{\ddagger}^2 \max_{1\leq i\leq T_0}\max_{1\leq j\leq T_0} \left \|\frac{\bm{X}_i\bm{X}_j^\prime}{T_1}-\bm{\Gamma}(i-j)\right \|_{\op}.
	\end{align*}
	Combining this with Lemma \ref{lemma:RE1}, we accomplish the proof of Lemma \ref{lemma:RE}.
\end{proof}

%%%%%%%%%%%%%%%%%%%%%%%%%%%%%%%%%%%%%%%%%%%%%%%%%%%%%%%%%%%%%%%%%%%%%%%%%%%%%%
\paragraph{Proof of Lemma \ref{lemma:errtrunc1}}
Since $\bm{r}_t=\sum_{j=T_0+1}^{\infty}\bm{A}_j^*\mathbf{y}_{t-j}=\sum_{j=T_0+1}^{\infty} \sum_{\ell=0}^\infty\bm{A}_j^* \bm{\Psi}_\ell^*\bm{\Sigma}_\varepsilon^{1/2}\bm{\xi}_{t-j-\ell}$, we have 
\begin{align*}
	\left\|\bm{R}\right\|_{\Fr}^2 =\sum_{t=T_0+1}^{T}\langle \bm{r}_t, \bm{r}_t\rangle
	&=\sum_{t=T_0+1}^{T} \sum_{i=T_0+1}^{\infty}\sum_{j=T_0+1}^{\infty} \sum_{k=0}^{\infty} \sum_{\ell=0}^{\infty}\left\langle\bm{A}_j^*\bm{\Psi}_\ell^*\bm{\Sigma}_\varepsilon^{1/2}\bm{\xi}_{t-j-\ell}, \bm{A}_i^*\bm{\Psi}_k^*\bm{\Sigma}_\varepsilon^{1/2}\bm{\xi}_{t-i-k}\right\rangle\\
	&=\sum_{i=T_0+1}^{\infty}\sum_{j=T_0+1}^{\infty} \sum_{k=0}^{\infty} \sum_{\ell=0}^{\infty}\sum_{t=T_0+1}^{T}\langle \widetilde{\bm{M}}_{i,j,k, \ell} \bm{\xi}_{t-i-k},  \bm{\xi}_{t-j-\ell} \rangle,
\end{align*}
and
\[
\widetilde{\bm{M}}_{i,j,k,\ell}=\bm{\Sigma}_\varepsilon^{1/2} \bm{\Psi}_\ell^{*\prime} \bm{A}_j^{*\prime}\bm{A}_i^*\bm{\Psi}_k^*\bm{\Sigma}_\varepsilon^{1/2}.
\]
Then, by Lemma \ref{lemma:xi_dev} and a method similar to the proof of Lemma \ref{lemma:devbd2}, we can show that
\begin{align}\label{eq:errtrunc2}
	&\mathbb{P}\left\{\left\|\bm{R}\right\|_{\Fr}^2 \geq  C \sum_{i=T_0+1}^{\infty}\sum_{j=T_0+1}^{\infty} \sum_{k=0}^{\infty} \sum_{\ell=0}^{\infty}
	b_{i,j,k,\ell} 
	+ \sum_{i=T_0+1}^{\infty}\sum_{j=T_0+1}^{\infty} \sum_{k=0}^{\infty} \sum_{\ell=0}^{\infty}
	c_{i,j,k,\ell} I\{\ell=i+k-j\}
	\right\} \notag \\
	&\hspace{5mm}\leq  2\sum_{i=T_0+1}^{\infty}\sum_{j=T_0+1}^{\infty} \sum_{k=0}^{\infty} \sum_{\ell=0}^{\infty}\delta_{i,j,k,\ell},
\end{align}
where $I(\cdot)$ is the indicator function, $\delta_{i,j,k,\ell}\in(0,1)$ will be specified shortly, 
\[
b_{i,j,k,\ell}=\sigma^2 \|\widetilde{\bm{M}}_{i,j,k,\ell}\|_{\Fr} \left \{\log(1/\delta_{i,j,k,\ell}) + \sqrt{T_1 \log(1/\delta_{i,j,k,\ell})} \right \}\]
and
\[
c_{i,j,k,\ell}=T_1 \sqrt{N}\|\widetilde{\bm{M}}_{i,j,k,\ell}\|_{\Fr}.
\]

By Assumption \ref{assum:Adecay} and \ref{assum:error},
\[
\|\widetilde{\bm{M}}_{i,j,k,\ell}\|_{\Fr} \leq  \lambda_{\max}(\bm{\Sigma}_\varepsilon)   \|\bm{\Psi}_\ell^*\|_{\op} \|\bm{A}_j^{*}\|_{\Fr} \|\bm{A}_i^{*}\|_{\op} \|\bm{\Psi}_k^*\|_{\op}
\leq C \rho^{i+k+\ell}\|\bm{A}_j^*\|_{\Fr}.
\]
Then, since $c_{i,j,k,\ell} \geq 0$, we have
\begin{align}\label{eq:errtrunc3}
	\sum_{i=T_0+1}^{\infty}\sum_{j=T_0+1}^{\infty} \sum_{k=0}^{\infty} \sum_{\ell=0}^{\infty}
	c_{i,j,k,\ell} I\{\ell=i+k-j\} 
	&\leq  \sum_{i=T_0+1}^{\infty}\sum_{j=T_0+1}^{\infty} \sum_{k=0}^{\infty} \sum_{\ell=0}^{\infty}
	c_{i,j,k,\ell} \notag \\
	&\leq C T_1 \sqrt{N}       \sum_{i=T_0+1}^{\infty} \sum_{k=0}^{\infty} \sum_{\ell=0}^{\infty}  \rho^{i+k+\ell} \sum_{j=T_0+1}^{\infty}\|\bm{A}_j^*\|_{\Fr} \notag \\
	&=  \frac{C \rho}{(1-\rho)^3}  \sqrt{N} \rho^{T_0} T_1 \sum_{j=T_0+1}^{\infty}\left\|\bm{A}_j^*\right\|_{\Fr} \notag \\
	&\leq  \frac{C \rho^2}{(1-\rho)^4} \lambda_{\max}(\bm{\Sigma}_\varepsilon) N \rho^{2T_0} T_1 \notag \\
	&\leq C  N \rho^{3T_0/2},
\end{align}
where the last but two inequality follows from Assumption \ref{assum:Adecay} and the last inequality follows from Assumption \ref{assum:T0}.	
Moreover, by choosing
\begin{equation*}
	\delta_{i,j,k,\ell} = \exp\left \{-2\rho^{-(i+k+\ell+j)/2} N /\log(1/\rho) \right \},
\end{equation*}
i.e., $\log(1/\delta_{i,j,k,\ell})= 2\sqrt{j}\rho^{-(i+k+\ell)/2} N /\log(1/\rho)$,  we can show that
\begin{align}\label{eq:errtrunc4}
	\begin{split}
		& \sum_{i=T_0+1}^{\infty}\sum_{j=T_0+1}^{\infty} \sum_{k=0}^{\infty} \sum_{\ell=0}^{\infty}
		b_{i,j,k,\ell} \\
		& \hspace{5mm}\leq  C \sigma^2 \sqrt{N} \sum_{i=T_0+1}^{\infty} \sum_{j=T_0+1}^{\infty}  \sum_{k=0}^{\infty} \sum_{\ell=0}^{\infty} \left \{ \frac{\rho^{(i+k+\ell+j)/2} N }{\log(1/\rho)} + \rho^{3(i+k+\ell+j)/4}  \sqrt{\frac{T_1 N }{\log(1/\rho)}}  \right \}  \\
		& \hspace{5mm} = C \sigma^2  \sqrt{N}  \left \{\frac{\rho}{(1-\sqrt{\rho})^4} \cdot \frac{\rho^{T_0} N }{\log(1/\rho)} 
		+\frac{\rho^{3/2} }{(1-\rho^{3/4})^2} \cdot \sqrt{ \frac{\rho^{3T_0} T_1 N }{\log(1/\rho)}}
		\right \} \\
		& \hspace{5mm}  \leq C \sigma^2  \sqrt{N} \left ( \rho^{T_0}N
		+ \sqrt{ \rho^{5T_0/2} N}
		\right )\\
		& \hspace{5mm}  \leq
		C  N^{3/2} \rho^{T_0},
	\end{split}
\end{align}
where we used Assumption \ref{assum:Adecay} in the first inequality and Assumption \ref{assum:T0} in the second to last inequality.

In addition, similarly to \eqref{eq:devbd11}, we can show that
\begin{align}\label{eq:errtrunc5}
	\sum_{i=T_0+1}^{\infty}\sum_{j=T_0+1}^{\infty} \sum_{k=0}^{\infty} \sum_{\ell=0}^{\infty} \delta_{i,j,k,\ell} 
	&\leq  \sum_{i=T_0+1}^{\infty}\sum_{j=T_0+1}^{\infty} \sum_{k=0}^{\infty} \sum_{\ell=0}^{\infty} \exp\left \{- (i+k+\ell+j) N \right \} \notag\\
	& \leq \frac{e^{ -(T_0+1)^2  N }}{\left(1- e^{- N} \right)^4 } \leq \frac{e^{ -T_0^2 N } }{\left(1- e^{- N} \right)^4 } \leq C e^{- N } .
\end{align}

Combining \eqref{eq:errtrunc2}--\eqref{eq:errtrunc5}, the proof of this lemma is complete.

%%%%%%%%%%%%%%%%%%%%%%%%%%%%%%%%%%%%%%%%%%%%%%%%%%%%%%%%%%%%%%%%%%%%%%%%%%%%%%
\subsubsection{Two auxiliary lemmas}
Below we present two auxiliary lemmas. Lemma \ref{lemma:xi_dev} is used in the proofs of Lemmas 
\ref{lemma:devbd1}, \ref{lemma:devbd2}, \ref{lemma:errtrunc1}, and  \ref{lemma:RE1}, while   Lemma \ref{lemma:martgl} establishes the basic martingale concentration bound which is used to prove Lemmas \ref{lemma:xi_dev} and \ref{lemma:RE}.

\begin{lemma}\label{lemma:xi_dev}
	Let $T_0< T$ be arbitrary  fixed time points, and $T_1=T-T_0$. Suppose that $\bm{\xi}_{t}$ is a random vector with independent, zero-mean, unit-variance, $\sigma^2$-sub-Gaussian coordinates for any $T_0+1\leq t\leq T$. Then 
	\begin{itemize}
		\item[(i)] for any $\delta\in(0,1]$, $\bm{M}\in\mathbb{R}^{N\times N}$, and fixed nonzero integer $1\leq s\leq T_0$, 
		\begin{equation}\label{eq:xis_1side}
			\mathbb{P}\left\{ \sum_{t=T_0+1}^{T} \langle \bm{M}\bm{\xi}_{t-s}, \bm{\xi}_{t} \rangle \geq C\sigma^2 \|\bm{M}\|_{\Fr} \left \{\log(1/\delta) + \sqrt{T_1 \log(1/\delta)} \right \} \right\}\leq 2\delta
		\end{equation}	
		and 
		\begin{equation}\label{eq:xis_2sides}
			\mathbb{P}\left\{ \left | \sum_{t=T_0+1}^{T} \langle \bm{M}\bm{\xi}_{t-s}, \bm{\xi}_{t} \rangle \right |\geq C\sigma^2 \|\bm{M}\|_{\Fr} \left \{\log(1/\delta) + \sqrt{T_1 \log(1/\delta)} \right \} \right\}\leq 4\delta;
		\end{equation}	
		\item[(ii)] for any   $\delta\in(0,1]$ and $\bm{M}\in\mathbb{R}^{N\times N}$,
		\begin{equation}\label{eq:xi_1side}
			\mathbb{P}\left\{ \sum_{t=T_0+1}^{T} \langle \bm{M}\bm{\xi}_{t}, \bm{\xi}_{t} \rangle \geq C\sigma^2 \|\bm{M}\|_{\Fr} \left \{\log(1/\delta) + \sqrt{T_1 \log(1/\delta)} \right \} + T_1 \sqrt{N}\|\bm{M}\|_{\Fr} \right\}\leq \delta
		\end{equation}	
		and 
		\begin{equation}\label{eq:xi_2sides}
			\begin{split}
				&\mathbb{P}\left\{ \left | \sum_{t=T_0+1}^{T} \langle \bm{M}\bm{\xi}_{t}, \bm{\xi}_{t} \rangle - E\left (\sum_{t=T_0+1}^{T} \langle \bm{M}\bm{\xi}_{t}, \bm{\xi}_{t} \rangle \right ) \right | \geq C\sigma^2 \|\bm{M}\|_{\Fr} \left \{\log(1/\delta) + \sqrt{T_1 \log(1/\delta)} \right \} \right\}\\
				&\hspace{40mm}\leq 2\delta.
			\end{split}
		\end{equation}	
	\end{itemize}
\end{lemma}

\begin{proof}	
	For any integer $s$, denote $\bm{\xi}_{[s]}=(\bm{\xi}_{T-s}^\prime, \bm{\xi}_{T-1-s}^\prime, \dots, \bm{\xi}_{T_0+1-s}^\prime)^\prime \in\mathbb{R}^{T_1 N}$,
	$S_{\xi,[s]}=\sum_{t=T_0+1}^{T} \langle \bm{M}\bm{\xi}_{t-s}, \bm{\xi}_{t} \rangle$, and
	$V_{\xi,[s]}=\sum_{t=T_0+1}^{T}\|\bm{M}\bm{\xi}_{t-s}\|^2$. Note that $\bm{\xi}_{[s]}$ is a random vector with independent, zero-mean, unit-variance, $\sigma^2$-sub-Gaussian coordinates for any fixed integer $s$. For simplicity, in the following we omit the subscript $[s]$ in $\bm{\xi}_{[s]}$, $S_{\xi,[s]}$, and $V_{\xi,[s]}$ whenever $s\neq 0$.

	We first prove claim (i) of this lemma. Without loss of generality, assume that $s$ is a positive integer. 
	Note that $V_\xi=\|(\bm{I}_{T_1}\otimes \bm{M} ) \bm{\xi}\|^2=\bm{\xi}^\prime \bm{Q} \bm{\xi}$, where $\bm{Q}=\bm{I}_{T_1}\otimes \bm{M}^\prime \bm{M}$.
	Then $\mathbb{E}(V_\xi)=\|\bm{I}_{T_1}\otimes \bm{M}\|_{\Fr}^2=T_1 \|\bm{M}\|_{\Fr}^2$.  By the (one-sided) Hanson-Wright inequality \citep{vershynin2010introduction}, for any $\eta\geq0$,
	\begin{align*}
		\mathbb{P}\left ( V_{\xi} - T_1 \|\bm{M}\|_{\Fr}^2 \geq \eta \right ) 
		&\leq   \exp\left\{-c\min\left(\frac{\eta}{\sigma^2\|\bm{Q}\|_{\op}},\frac{\eta^2}{\sigma^4\|\bm{Q}\|_{\Fr}^2}\right)\right\},
	\end{align*}
	where $c>0$ is an absolute constant.
	Since $\|\bm{Q}\|_{\op}=\|\bm{M}\|_{\op}^2 \leq\|\bm{M}\|_{\Fr}^2$, and $\|\bm{Q}\|_{\Fr}\leq \|\bm{I}_{T_1}\otimes \bm{M}\|_{\Fr} \|\bm{I}_{T_1}\otimes \bm{M}\|_{\op} \leq \sqrt{T_1} \|\bm{M}\|_{\Fr}^2$, we can show that
	\begin{align*}
		\mathbb{P}\left ( V_{\xi} \geq T_1\|\bm{M}\|_{\Fr}^2 + \eta \right ) 
		&\leq  \exp\left\{-c\min\left(\frac{\eta}{\sigma^2\|\bm{M}\|_{\Fr}^2},\frac{\eta^2}{\sigma^4 T_1 \|\bm{M}\|_{\Fr}^4}\right)\right\} \leq \delta.
	\end{align*} 
	by choosing 
	\begin{equation}\label{eq:xi_dev1}
		\eta=C \sigma^2 \|\bm{M}\|_{\Fr}^2 \left \{\log(1/\delta)+\sqrt{T_1\log(1/\delta)}\right\},
	\end{equation}
	where  $C$ is dependent on $c$.
	Moreover, by Lemma \ref{lemma:martgl}, for any $\alpha, \beta>0$, we have
	\begin{align*}
		\mathbb{P}(S_\xi \geq \alpha) &\leq \mathbb{P}(S_\xi \geq \alpha, \; V_{\xi} \leq \beta) + \mathbb{P}(V_{\xi} \geq \beta) 
		\leq \exp\left (-\frac{\alpha^2}{2\sigma^2\beta}\right ) + \mathbb{P}(V_{\xi} \geq \beta).
	\end{align*}
	This implies that, if $\beta=T_1\|\bm{M}\|_{\Fr}^2 + \eta$ and $\alpha \geq \sqrt{2\sigma^2\beta\log(1/\delta)}$, then $\mathbb{P}(S_\xi \geq \alpha) \leq \delta$.  Hence, we can establish \eqref{eq:xis_1side} by choosing
	\[
	\alpha= C\sigma^2 \|\bm{M}\|_{\Fr} \left \{\log(1/\delta) + \sqrt{T_1 \log(1/\delta)} \right \}.
	\]
	Furthermore, applying  \eqref{eq:xis_1side} to $-\bm{M}$, we directly have
	\begin{equation*}
		\mathbb{P}\left\{ \sum_{t=T_0+1}^{T} \langle \bm{M}\bm{\xi}_{t-s}, \bm{\xi}_{t} \rangle \leq - C\sigma^2 \|\bm{M}\|_{\Fr} \left \{\log(1/\delta) + \sqrt{T_1 \log(1/\delta)} \right \} \right\}\leq 2\delta,
	\end{equation*}	
	which, combined with \eqref{eq:xis_1side}, yields the two-sided bound in \eqref{eq:xis_2sides}.

	The proof of claim (ii) is similar to the analysis of $V_{\xi}$ above. Note that
	$S_{\xi,[0]}= \bm{\xi}_{[0]}^\prime (\bm{I}_{T_1}\otimes \bm{M})^\prime \bm{\xi}_{[0]}$. Applying the (two-sided) Hanson-Wright inequality, we have
	\[
	\mathbb{P}\left \{ \left |S_{\xi,[0]} -\mathbb{E}(S_{\xi,[0]}) \right | \geq \eta_0 \right \}\leq  2\exp\left\{-c\min\left(\frac{\eta_0}{\sigma^2\|\bm{M}\|_{\op}},\frac{\eta_0^2}{\sigma^4 T_1\|\bm{M}\|_{\Fr}^2}\right)\right\} \leq 2\delta
	\]
	if we choose 
	\begin{equation}\label{eq:xi_dev2}
		\eta_0=C \sigma^2 \|\bm{M}\|_{\Fr} \left \{\log(1/\delta)+\sqrt{T_1\log(1/\delta)}\right\},
	\end{equation}
	where  $C$ is dependent on $c$. This leads to \eqref{eq:xi_2sides} in the lemma. Moreover, since $\mathbb{E}(S_{\xi,[0]})=T_1 \trace(\bm{M})\leq T_1 \sqrt{N} \|\bm{M}\|_{\Fr}$,  similarly we can also obtain the  one-sided result:
	\begin{align*}
		\mathbb{P}\left \{ S_{\xi,[0]} \geq T_1 \trace(\bm{M}) + \eta_0 \right \}
		%&\leq \mathbb{P}\left \{ \left |S_{\xi,[0]} - T_1 \trace(\bm{M}) \right |\geq \eta_0 \right \}\\
		&\leq  \exp\left\{-c\min\left(\frac{\eta_0}{\sigma^2\|\bm{M}\|_{\op}},\frac{\eta_0^2}{\sigma^4 T_1\|\bm{M}\|_{\Fr}^2}\right)\right\} \leq \delta
	\end{align*}
	with the same choice of $\eta_0$ as in \eqref{eq:xi_dev2}. Therefore, \eqref{eq:xi_2sides} is proved as well.
\end{proof}

%%%%%%%%%%%%%%%%%%%%%%%%%%%%%%%%%%%%%%%%%%%%%%%%%%%%%%%%%%%%%%%%%%%%%%%%%%%%%%
\begin{lemma}[Martingale concentration]\label{lemma:martgl} 
	Let $\{\mathcal{F}_t, t\in\mathbb{Z}\}$ be a filtration. Suppose that $\{\bm{w}_t\}$ and $\{\bm{e}_t\}$ are processes taking values in $\mathbb{R}^d$, and for each integer $t$, $\bm{w}_t$ is $\mathcal{F}_{t-1}$-measurable, $\bm{e}_t$ is $\mathcal{F}_{t}$-measurable, and $\bm{e}_t\mid \mathcal{F}_{t-1}$ is mean-zero and $\sigma^2$-sub-Gaussian. Let $T_0< T$ be arbitrary  fixed time points. Then, for any $\alpha,\beta>0$, we have 
	\begin{equation*}
		\mathbb{P}\left \{ \sum_{t=T_0+1}^{T}\langle\bm{w}_t, \bm{e}_t \rangle\geq \alpha, \; \sum_{t=T_0+1}^{T}\lVert \bm{w}_t \rVert^2 \leq \beta\right \} \leq  \exp\left (-\frac{\alpha^2}{2\sigma^2\beta}\right ).
	\end{equation*}
\end{lemma}

\begin{proof}
	See Lemma 4.2 in \cite{simchowitz2018learning}.
\end{proof}

\subsection{Proofs of theoretical results in Section 3}

This section gives the proofs of Theorems \ref{thm:stat}, \ref{thm:optimization} and Corollary \ref{cor:algorithm} in Sections B.3.1-B.3.3, respectively.
Section B.3.4 provides five auxiliary lemmas, which are used in the proof of Theorem \ref{thm:optimization}. 
We first introduce several parameter spaces of $\cm{A}\in\mathbb{R}^{N\times N\times T_0}$ below,
\begin{align*}
	\bm{\Theta}(r_1, r_2) &= \{\cm{A}\in\mathbb{R}^{N\times N\times T_0}\mid \rank(\cm{A}_{(1)})\leq r_1, \; \rank(\cm{A}_{(2)})\leq r_2\},\\
	\bm{\Theta}^{\mathrm{SP}}(r_1,r_2,s)&= \{\cm{A}\in\mathbb{R}^{N\times N\times T_0}\mid  \rank(\cm{A}_{(1)}) \leq r_1, \;\rank(\cm{A}_{(2)}) \leq r_2,\; \|\cm{A}\|_0 \leq s,\},\\
	\bm{\Theta}_\ddagger(r_1,r_2) &= \{\cm{A}\in\bm{\Theta}(r_1, r_2),\; \|\cm{A}\|_{\ddagger}=1\} \hspace{3mm}\text{and}\hspace{3mm}
	\bm{\Theta}^{\mathrm{SP}}_1(r_1,r_2,s)= \{\cm{A}\in\bm{\Theta}^{\mathrm{SP}}(r_1,r_2,s),\; \|\cm{A}\|_{\Fr} =1\}.
\end{align*}

\subsubsection{Proof of Theorem \ref{thm:stat}}
It can be verified that $[\nabla\mathcal{L}(\cm{A})]_{(1)}=-T_1^{-1}(\bm{Y} -\cm{A}_{(1)}\bm{X})\bm{X}^\prime$ and $\bm{Y} = (\cm{A}^*_{S_{\gamma}})_{(1)}\bm{X} + (\cm{A}^*_{S_{\gamma}^c})_{(1)}\bm{X} + \widetilde{\cm{E}}$. As a result, for any $\cm{M}\in\mathbb{R}^{N\times N\times T_0}$,
\begin{align} \label{eq:stat-error-main}
	\langle \nabla\mathcal{L}(\cm{A}_{S_{\gamma}}^*) ,\cm{M}  \rangle = -\langle {T_1}^{-1}\widetilde{\cm{E}} \bm{X}^\prime, \cm{M}_{(1)}\rangle -\langle {T_1}^{-1} (\cm{A}_{S^c_{\gamma}}^*)_{(1)}\bm{X}\bm{X}^\prime, \cm{M}_{(1)}\rangle.
\end{align}
From Lemmas  \ref{lemma:devbd1} and  \ref{lemma:devbd2} and by a method similar to \eqref{eq:devbd}, we can show that, if $T_1\gtrsim (r_1 \wedge r_2)N + \log T_0$ and $\gamma\gtrsim \sqrt{\{(r_1\wedge r_2)N+\log T_0\}/{T_1}}$, then
\begin{equation*}
	\mathbb{P}\left\{ \sup_{\bm{\Delta}\in\bm{\Theta}_\ddagger(r_1,r_2)}
	\langle \frac{1}{T_1}\widetilde{\cm{E}}\bm{X}^\prime, \bm{\Delta}_{(1)}
	\rangle \geq C\gamma \right \}  \leq  C e^{-(r_1\wedge r_2)N-\log T_0},
\end{equation*}
which, together with the fact that	
$\|\cm{M}\|_{\ddagger} \leq \sqrt{s}\|\cm{M}\|_{\Fr} = \sqrt{s}$ for any $\cm{M}\in\bm{\Theta}^{\mathrm{SP}}_{1}(r_1,r_2,s)$, implies that
\begin{equation} \label{eq:stat-error-term1}
	\langle\frac{1}{T_1}\widetilde{\cm{E}}\bm{X}^\prime, \cm{M}_{(1)}\rangle \leq \|\cm{M}\|_{\ddagger} \sup_{\bm{\Delta}\in\bm{\Theta}_\ddagger(r_1,r_2)}
	\langle \frac{1}{T_1}\widetilde{\cm{E}}\bm{X}^\prime, \bm{\Delta}_{(1)}
	\rangle \leq C\gamma\sqrt{s}
\end{equation}
holds with probability at least $1-C e^{-(r_1\wedge r_2)N-\log T_0}$. 

We next handle the second term at the right hand side of \eqref{eq:stat-error-main}. It holds that, from Assumptions \ref{assum:Adecay}  and \ref{assum:error},
\begin{align*}
	\frac{\mathbb{E}\langle(\cm{A}^*_{S_{\gamma}^c})_{(1)}\bm{X}\bm{X}^\prime, \cm{M}_{(1)}\rangle}{T_1} &= \trace \left ((\cm{A}^*_{S_{\gamma}^c})_{(1)} \bm{\Sigma}_{T_0} \cm{M}_{(1)}^\prime \right )\\ 
	&\leq \lambda_{\max}(\bm{\Sigma}_{T_0} ) \|\cm{A}^*_{S_{\gamma}^c}\|_{\Fr}\|\cm{M}\|_{\Fr} \leq C^2\kappa_{\mathrm{RSS}}\|\cm{A}^*_{S_{\gamma}^c}\|_{\Fr}\|\cm{M}\|_{\Fr},
\end{align*}
and
\begin{align*}
	&\frac{|\langle(\cm{A}^*_{S_{\gamma}^c})_{(1)}\bm{X}\bm{X}^\prime, \cm{M}_{(1)}\rangle - \mathbb{E}\langle(\cm{A}^*_{S_{\gamma}^c})_{(1)}\bm{X}\bm{X}^\prime, \cm{M}_{(1)}\rangle|}{T_1}\\
	&\hspace{10mm}\leq \sum_{i=1}^{T_0}\sum_{j=1}^{T_0} \left|\trace\left[\bm{A}_i^* \left\{ \frac{\bm{X}_i\bm{X}_j^\prime}{T_1} - \bm{\Gamma}(i-j)\right\}\bm{M}_j \right] I\{i\in S_{\gamma}^c\} \right|\\
	&\hspace{10mm}\leq\|\cm{A}^*_{S_{\gamma}^c}\|_{\ddagger}\|\cm{M}\|_{\ddagger} \max_{1\leq i\leq T_0}\max_{1\leq j\leq T_0} \left \|\frac{\bm{X}_i\bm{X}_j^\prime}{T_1}-\bm{\Gamma}(i-j)\right \|_{\op}.
\end{align*}
As a result, by Lemma \ref{lemma:RE1} and the fact that $\|\cm{M}\|_{\ddagger} \leq \sqrt{s}\|\cm{M}\|_{\Fr}$,
\begin{align} \label{eq:stat-error-term2}
	{T_1}^{-1}\langle(\cm{A}^*_{S_{\gamma}^c})_{(1)}&\bm{X}\bm{X}^\prime, \cm{M}_{(1)}\rangle \notag\\ &\leq \frac{\mathbb{E}\langle(\cm{A}^*_{S_{\gamma}^c})_{(1)}\bm{X}\bm{X}^\prime, \cm{M}_{(1)}\rangle}{T_1} + \frac{|\langle(\cm{A}^*_{S_{\gamma}^c})_{(1)}\bm{X}\bm{X}^\prime, \cm{M}_{(1)}\rangle - \mathbb{E}\langle(\cm{A}^*_{S_{\gamma}^c})_{(1)}\bm{X}\bm{X}^\prime, \cm{M}_{(1)}\rangle|}{T_1}\notag\\
	& \leq  \left(C^2\kappa_{\mathrm{RSS}}\|\cm{A}^*_{S_{\gamma}^c}\|_{\Fr} + \tau^2\sqrt{s}\|\cm{A}^*_{S_{\gamma}^c}\|_{\ddagger}\right) \|\cm{M}\|_{\Fr}.
\end{align}
holds with probability at least $1-C e^{ - N - \log T_0 }$ when $T_1\gtrsim s^2(N + \log T_0)$.
We accomplish the proof by combining \eqref{eq:stat-error-main} -- \eqref{eq:stat-error-term2} and letting $\|\cm{M}\|_{\Fr} = 1$.

\subsubsection{Proof of Theorem \ref{thm:optimization}}\label{sec:prop_algo}

This proof is divided into six steps. Some notations and conditions are given in the first step, and verified in the last step. 
Without loss of generality, we assume that $0<\sigma_L<1<\sigma_U$ and $0<\kappa_{\mathrm{RSC}}<1<\kappa_{\mathrm{RSS}}$ throughout this proof.

\noindent\textbf{Step 1} (Notations and conditions)
Without confusion, we use $\cm{A}^*$ to denote $\cm{A}^*_{S_\gamma}$ for simplicity in this proof.
Since $\cm{A}^*$ has Tucker ranks $r_1$ and $r_2$ along the first two modes, its Tucker decomposition can be assumed to have the form of $\cm{A}^* = \cm{G}^* \times_1 \bm{U}_1^* \times_2 \bm{U}_2^*$, where $\cm{G}^*\in\mathbb{R}^{r_1\times r_2\times T_0}$, $\bm{U}_i^*\in\mathbb{R}^{N\times r_i}$ and $\bm{U}_i^{*\prime}\bm{U}_i^* = b^2 \bm{I}_{r_i}$ for $1\leq i \leq 2$.
Moreover, at the $k$-th iteration of Algorithm \ref{alg:AGD-HT}, the pre-thresholding estimator $\cm{\widetilde{A}}^{k+1}$ has the Tucker form of $\cm{\widetilde{G}}^{k+1} \times_1 \bm{U}_1^{k+1} \times_2 \bm{U}_2^{k+1}$, where $\cm{\widetilde{G}}^{k+1}, \bm{U}_1^{k+1}$ and $\bm{U}_2^{k+1}$ are obtained by one-step gradient descent, and the hard-thresholding operation gives
\[
\cm{A}^{k+1} = \hardt{\cm{\widetilde{A}}^{k+1}, s} = \cm{G}^{k+1} \times_1 \bm{U}_1^{k+1} \times_2 \bm{U}_2^{k+1}.
\]
Note that the zero frontal slices in $\cm{A}^*$, $\cm{\widetilde{A}}^{k+1}$ and $\cm{A}^{k+1}$ correspond to the zero ones in $\cm{G}^*$, $\cm{\widetilde{G}}^{k+1}$ and $\cm{G}^{k+1}$, respectively.

On the other hand, we denote by $\breve{S}$ the union set $\breve{S}_{k+1,\gamma} =  S_k \cup S_{k+1} \cup S_{\gamma}$, where the subscripts of $k$ and $\gamma$ are suppressed when there is no confusion, and it holds that $|\breve{S}| \leq 3s$ since $s\geq s_{\gamma}$.
Moreover, it can be verified that $\cm{A}^{k}_{\breve{S}}=\cm{A}^{k}$, $\cm{A}^{k+1}_{\breve{S}}=\cm{A}^{k+1}$ and
\[
\cm{A}^{k+1} = \hardt{\cm{\widetilde{A}}^{k+1}_{\it{\breve{S}}}, s} \hspace{2mm}\text{with }\hspace{2mm}\cm{\widetilde{A}}^{k+1}_{\breve{S}} = \cm{\widetilde{G}}^{k+1}_{\breve{S}} \times_1 \bm{U}_1^{k+1} \times_2 \bm{U}_2^{k+1}.
\]
The distances between $\cm{\widetilde{A}}^{k+1}_{\breve{S}}$, $\cm{A}^{k+1}$ and $\cm{A}^*$ can be measured by
\[
\widetilde{E}^{k+1} = \min_{\bm{R}_i\in \mathcal{O}^{r_i\times r_i}, 1\leq i\leq 2} \sum_{i=1}^{2} \|\bm{U}_i^{k+1} - \bm{U}_i^*\bm{R}_i\|_{\Fr}^2 + \|\cm{\widetilde{G}}_{\breve{S}}^{k+1} - \cm{G}^* \times_1 \bm{R}_1^{ \prime} \times_2 \bm{R}_2^{\prime} \|_{\Fr}^2,
\]
\[
E^{k+1} = \min_{\bm{R}_i\in \mathcal{O}^{r_i\times r_i}, 1\leq i\leq 2} \sum_{i=1}^{2} \|\bm{U}_i^{k+1} - \bm{U}_i^*\bm{R}_i\|_{\Fr}^2 + \|\cm{G}^{k+1} - \cm{G}^* \times_1 \bm{R}_1^{ \prime} \times_2 \bm{R}_2^{\prime} \|_{\Fr}^2,
\]
respectively, and their optimizers are denoted by $(\bm{\widetilde{R}}^{k+1}_1, \bm{\widetilde{R}}^{k+1}_2)$ and $(\bm{R}^{k+1}_1, \bm{R}^{k+1}_2)$.
We next introduce or rephrase the following list of conditions:
\begin{itemize}
	\item Suppose that there exist $\alpha, \beta>0$ such that
	\begin{align}\label{eq:RGC}
		\langle\nabla\mathcal{L}(\cm{A}) - \nabla\mathcal{L}(\cm{A}^*), \cm{A} - \cm{A}^*\rangle \geq \alpha\|\cm{A} - \cm{A}^*\|_{\Fr}^2 + \beta\|\nabla\mathcal{L}(\cm{A}) - \nabla\mathcal{L}(\cm{A}^*)\|_{\Fr}^2,
	\end{align}
	holds for any $\cm{A}\in\bm{\Theta}^{\mathrm{SP}}(r_1,r_2,3s)$. Moreover, by Cauchy-Schwarz and the fact that $xy\leq \alpha x^2+0.25\alpha^{-1} y^2$, it can be further verified that $\alpha\beta\leq 0.25$.
	
	\item We assume that $b = \sigma_U^{1/4}$ for simplicity, and the proof can be easily adjusted if $c\sigma_U^{1/4} \leq b \leq C\sigma_U^{1/4}$ for two absolute constants $0<c<C$.
	Moreover, for $0\leq k\leq K$ and $i=1$ and 2, 
	\begin{align}\label{eq:U&tildeG-upperbounds}
		\|\bm{U}_i^k\|_{\op} \leq 1.1\sigma_U^{1/4}, \hspace{2mm}\sigma_{\min}(\bm{U}_i^k) \geq 0.9\sigma_U^{1/4} \hspace{2mm}\text{and}\hspace{2mm} \|\cm{G}^k_{(i)}\|_{\op} \leq 1.1\sigma_U^{1/2}.
	\end{align}
	Finally, for all $0\leq k\leq K$,
	\begin{align} \label{eq:E-upperbounds}
		E^k \leq  \frac{c_0\sigma_L^{1/2}}{\kappa^{3/2}} =: C_1,
	\end{align}
	where $\kappa=\sigma_U/\sigma_L$, and $c_0>0$ is a small absolute constant (smaller than 1) determined later.
	
	\item Let $(\cm{A},\bm{U}_i, E)$ be $(\cm{A}^k, \bm{U}_i^k, E^k)$, $(\cm{\widetilde{A}}_{\breve{S}}^{k+1}, \bm{U}_i^{k+1}, \widetilde{E}^{k+1})$ or $(\cm{A}^0, \bm{U}_i^0, E^0)$, respectively. An important two-sided inequality is derived from Lemma \ref{lemma:A->E->A} by letting $b = \sigma_U^{1/4}$ and $c_e = 0.1$, namely
	\begin{align}\label{eq:A->E->A}
		C_L \|\cm{A}- \cm{A}^*\|_{\Fr}^2 \leq E \leq C_{U,1}\|\cm{A} - \cm{A}^*\|_{\Fr}^2 + C_{U,2}\sum_{i=1}^{2}\|\bm{U}_i^{\prime}\bm{U}_i - b^2 \bm{I}_{r_i}\|_{\Fr}^2,
	\end{align} 
	where $C_L = [5(\sigma_U+2\sigma_U^{3/2})]^{-1}$, $C_{U,1} = 3\sigma_U^{-1} + 8\sigma_L^{-2}\sigma_U^{-1/2}+40\sigma_L^{-2}$ and $C_{U,2} = 2\sigma_U^{-1/2}+10$.
\end{itemize}

\noindent\textbf{Step 2} (Descent of $\widetilde{E}^{k+1}$)
This step aims to establish \begin{equation}\label{eq:step2}
	\widetilde{E}^{k+1} \leq E^{k} + \eta^2 (	Q_{\mathcal{G},1} + \sum_{i=1}^{2}Q_{i,1}) - 2\eta (	Q_{\mathcal{G},2} + \sum_{i=1}^{2}Q_{i,2}),
\end{equation}
where $Q_{1,j}, Q_{2,j}$ and $Q_{\mathcal{G},j}$ with $j=1$ and 2 are defined in \eqref{eq:(U1(k+1)-U1*)-main}, \eqref{eq:(U2(k+1)-U2*)-main} and \eqref{eq:(tildeG(k+1)-G*)-main}, respectively.
Note that, by the definition of $\widetilde{E}^{k+1}$,
\begin{align} \label{eq:tildeE(k+1)}
	\widetilde{E}^{k+1} \leq \sum_{i=1}^{2} \|\bm{U}_i^{k+1} - \bm{U}_i^*\bm{R}_i^k\|_{\Fr}^2 + \|\cm{\widetilde{G}}_{\breve{S}}^{k+1} - \cm{G}^* \times_1 (\bm{R}_1^k)^{ \prime} \times_2 (\bm{R}_2^k)^{\prime} \|_{\Fr}^2.
\end{align}

For the first term of \eqref{eq:tildeE(k+1)}, the gradient descent update of $\bm{U}_1^{k+1}$ gives
\begin{align}\label{eq:U(k+1)-U*}
	\|\bm{U}_1^{k+1} - \bm{U}_1^*\bm{R}_1^k\|_{\Fr}^2	&= \|\bm{U}_1^{k} - \eta[\nabla_{{U}_1}\mathcal{L}(\cm{A}^k)+a\bm{U}_1^k(\bm{U}_1^{k\prime}\bm{U}_1^k - b^2\bm{I}_{r_1})] - \bm{U}_1^*\bm{R}_1^k\|_{\Fr}^2\notag\\
	&= \|\bm{U}_1^{k} - \bm{U}_1^*\bm{R}_1^k\|_{\Fr}^2 +\eta^2\|\nabla_{{U}_1}\mathcal{L}(\cm{A}^k)+a\bm{U}_1^k(\bm{U}_1^{k\prime}\bm{U}_1^k - b^2\bm{I}_{r_1})\|_{\Fr}^2 \notag\\
	&\hspace{5mm}- 2\eta\langle \nabla_{{U}_1}\mathcal{L}(\cm{A}^k), \bm{U}_1^{k} - \bm{U}_1^*\bm{R}_1^k\rangle - 2a\eta \langle \bm{U}_1^k(\bm{U}_1^{k\prime}\bm{U}_1^k- b^2\bm{I}_{r_1}), \bm{U}_1^{k} - \bm{U}_1^*\bm{R}_1^k\rangle,
\end{align}
where  $\nabla_{{U}_1}\mathcal{L}(\cm{A}) = [\nabla\mathcal{L}(\cm{A})]_{(1)}(\bm{I}_{T_0}\otimes \bm{U}_2)\cm{G}_{(1)}^\prime$ is the partial derivative of the loss function $\mathcal{L}(\cm{A})$ with respect to $\bm{U}_1$.
We will first handle the last three terms of \eqref{eq:U(k+1)-U*} one-by-one (without the scaling constants). Starting with the second term,
\begin{align*}
	\|\nabla_{{U}_1}\mathcal{L}(\cm{A}^k)+a\bm{U}_1^k(\bm{U}_1^{k\prime}\bm{U}_1^k - b^2\bm{I}_{r_1})\|_{\Fr}^2 &\leq 2\|\nabla_{{U}_1}\mathcal{L}(\cm{A}^k)\|_{\Fr}^2 + 2a^2\|\bm{U}_1^k(\bm{U}_1^{k\prime}\bm{U}_1^k - b^2\bm{I}_{r_1})\|_{\Fr}^2.
\end{align*}
Let $C_2 = 1.5\sigma_U^{3/4}$ and, by the definition of dual norm, the conditions at \eqref{eq:U&tildeG-upperbounds} and $\|\cm{G}^k\|_0=s $,
\begin{align*}
	\|\nabla_{{U}_1}\mathcal{L}(\cm{A}^k)\|_{\Fr}^2 &= \sup_{\bm{M}\in\mathbb{R}^{N\times r_1},\|\bm{M}\|_{\Fr}=1}\langle\nabla_{{U}_1}\mathcal{L}(\cm{A}^k), \bm{M} \rangle^2\\
	&= \sup_{\bm{M}\in\mathbb{R}^{N\times r_1},\|\bm{M}\|_{\Fr}=1}\langle\nabla\mathcal{L}(\cm{A}^k), \cm{G}^k\times_1\bm{M}\times_2\bm{U}_2^k\rangle^2\\
	&\leq 2\sup_{\bm{M}\in\mathbb{R}^{N\times r_1},\|\bm{M}\|_{\Fr}=1}\langle\nabla\mathcal{L}(\cm{A}^*), \cm{G}^k\times_1\bm{M}\times_2\bm{U}_2^k\rangle^2 \\
	&\hspace{5mm}+ 2\sup_{\bm{M}\in\mathbb{R}^{N\times r_1},\|\bm{M}\|_{\Fr}=1}\langle\nabla\mathcal{L}(\cm{A}^k) - \nabla\mathcal{L}(\cm{A}^*), \cm{G}^k\times_1\bm{M}\times_2\bm{U}_2^k\rangle^2\\
	&\leq 2C_2^2\left(e_\mathrm{stat}^2+\|\nabla\mathcal{L}(\cm{A}^k)-\nabla\mathcal{L}(\cm{A}^*)\|_{\Fr}^2\right),
\end{align*}
which leads to
\begin{align} \label{eq:(U(k+1)-U*)-quadratic}
	\|\nabla_{{U}_1}&\mathcal{L}(\cm{A}^k)+a\bm{U}_1^k(\bm{U}_1^{k\prime}\bm{U}_1^k - b^2\bm{I}_{r_1})\|_{\Fr}^2 \notag\\
	&\leq 4C_2^2\left(e_\mathrm{stat}^2+\|\nabla\mathcal{L}(\cm{A}^k)-\nabla\mathcal{L}(\cm{A}^*)\|_{\Fr}^2\right)  + 3\sigma_U^{1/2}a^2\|\bm{U}_1^{k\prime}\bm{U}_1^k - b^2\bm{I}_{r_1}\|_{\Fr}^2 = : Q_{1,1}.
\end{align}

We next consider the third term at \eqref{eq:U(k+1)-U*}. Let $\cm{A}_{U_1} = \cm{G}\times_1(\bm{U}_1 - \bm{U}_1^*\bm{R}_1)\times_2\bm{U}_2$ and, by the conditions at \eqref{eq:U&tildeG-upperbounds}, $\|\cm{A}_{U_1}^k\|_{\Fr} \leq C_2\|\bm{U}_1^{k} - \bm{U}_1^*\bm{R}_1^k\|_{\Fr}$ for all $k\geq 1$. 
Moreover, by the conditions at \eqref{eq:U&tildeG-upperbounds} and the fact that $xy \leq 0.5x^2 + 0.5y^2$, 
\begin{align*}
	\langle \nabla\mathcal{L}(\cm{A}^*),  \cm{A}_{U_1}^k\rangle &\leq  \sup_{\cmt{M}\in\Theta_1^{\mathrm{SP}}(r_1,r_2,s)}\langle\nabla\mathcal{L}(\cm{A}^*),\cm{M}\rangle\cdot\|\cm{A}_{U_1}^k\|_{\Fr}\\
	&\leq 0.5C_2^2 C_1^{-1}  e_\mathrm{stat}^2 + 0.5 C_1\|\bm{U}_1^{k} - \bm{U}_1^*\bm{R}_1^k\|_{\Fr}^2.
\end{align*}
As a result,
\begin{align} \label{eq:(U(k+1)-U*)-cross1}
	\langle \nabla_{{U}_1}\mathcal{L}(\cm{A}^k), \bm{U}_1^{k} - \bm{U}_1^*\bm{R}_1^k\rangle &= \langle \nabla\mathcal{L}(\cm{A}^k), \cm{G}^k\times_1(\bm{U}_1^k - \bm{U}_1^*\bm{R}_1^k)\times_2\bm{U}_2^k \rangle \notag\\
	&=
	\langle \nabla\mathcal{L}(\cm{A}^k) - \nabla\mathcal{L}(\cm{A}^*), \cm{A}_{U_1}^k\rangle + \langle \nabla\mathcal{L}(\cm{A}^*), \cm{A}_{U_1}^k\rangle \notag\\
	&\geq \langle \nabla\mathcal{L}(\cm{A}^k) - \nabla\mathcal{L}(\cm{A}^*), \cm{A}_{U_1}^k\rangle - 0.5 C_1\|\bm{U}_1^{k} - \bm{U}_1^*\bm{R}_1^k\|_{\Fr}^2\notag\\
	&\hspace{5mm} - 0.5C_2^2 C_1^{-1}  e_\mathrm{stat}^2.
\end{align}

For the last term at \eqref{eq:U(k+1)-U*}, it holds that
\begin{align}\label{eq:(U(k+1)-U*)-cross2}
	\begin{split}
		&\langle \bm{U}_1^k(\bm{U}_1^{k\prime}\bm{U}_1^k- b^2\bm{I}_{r_1}), \bm{U}_1^{k} - \bm{U}_1^*\bm{R}_1^k\rangle = \langle\bm{U}_1^{k\prime}\bm{U}_1^k- b^2\bm{I}_{r_1},  \bm{U}_1^{k\prime}\bm{U}_1^{k} - \bm{U}_1^{k\prime}\bm{U}_1^*\bm{R}_1^k\rangle\\
		&\hspace{20mm}= 0.5\|\bm{U}_1^{k\prime}\bm{U}_1^k- b^2\bm{I}_{r_1}\|_{\Fr}^2 + 0.5\langle\bm{U}_1^{k\prime}\bm{U}_1^k- b^2\bm{I}_{r_1},  (\bm{U}_1^k-\bm{U}_1^{*}\bm{R}_1^k)^\prime(\bm{U}_1^k-\bm{U}_1^*\bm{R}_1^k)\rangle \\
		&\hspace{20mm}\geq 0.25\|\bm{U}_1^{k\prime}\bm{U}_1^k- b^2\bm{I}_{r_1}\|_{\Fr}^2 - 0.25\|\bm{U}_1^k-\bm{U}_1^*\bm{R}_1^k\|_{\Fr}^4\\
		&\hspace{20mm}\geq 0.25\|\bm{U}_1^{k\prime}\bm{U}_1^k- b^2\bm{I}_{r_1}\|_{\Fr}^2 - 0.25C_1 \|\bm{U}_1^k-\bm{U}_1^*\bm{R}_1^k\|_{\Fr}^2,
	\end{split}
\end{align}
where the second equality is due to the fact that, for a symmetric matrix $\bm{W}\in\mathbb{R}^{r_1\times r_1}$, 
$$\langle\bm{W}, \bm{U}_1^\prime\bm{U}_1 - \bm{U}_1^\prime\bm{U}_1^*\bm{R}_1\rangle=0.5\langle\bm{W}, (\bm{U}_1-\bm{U}_1^*\bm{R}_1)^\prime(\bm{U}_1-\bm{U}_1^*\bm{R}_1) \rangle +0.5\langle\bm{W}, \bm{U}_1^\prime\bm{U}_1-b^2\bm{I}_{r_1} \rangle, $$
the first inequality is due to Cauchy-Schwarz inequality and the fact that $xy\leq 0.5x^2 + 0.5y^2$, and the last inequality is due to the fact that $\|\bm{U}_1^k-\bm{U}_1^*\bm{R}_1^k\|_{\Fr}^2 \leq E^k \leq C_1$ at \eqref{eq:E-upperbounds}.
We combine the last two terms of \eqref{eq:U(k+1)-U*}, i.e. \eqref{eq:(U(k+1)-U*)-cross1} and \eqref{eq:(U(k+1)-U*)-cross2}, with the scaling constant $a>0$, and it leads to
\begin{align}\label{eq:(U(k+1)-U*)-cross}
	&\langle \nabla_{{U}_1}\mathcal{L}(\cm{A}^k), \bm{U}_1^{k} - \bm{U}_1^*\bm{R}_1^k\rangle + a \langle \bm{U}_1^k(\bm{U}_1^{k\prime}\bm{U}_1^k- b^2\bm{I}_{r_1}), \bm{U}_1^{k} - \bm{U}_1^*\bm{R}_1^k\rangle \notag\\
	& \geq \langle \nabla\mathcal{L}(\cm{A}^k) - \nabla\mathcal{L}(\cm{A}^*), \cm{A}_{U_1}^k\rangle - (1+a) C_1\|\bm{U}_1^{k} - \bm{U}_1^*\bm{R}_1^k\|_{\Fr}^2 \notag\\
	&\hspace{5mm}+ 0.25a\|\bm{U}_1^{k\prime}\bm{U}_1^k- b^2\bm{I}_{r_1}\|_{\Fr}^2 - 0.5C_2^2 C_1^{-1}  e_\mathrm{stat}^2 =: Q_{1,2}.
\end{align}

By plugging \eqref{eq:(U(k+1)-U*)-quadratic} and \eqref{eq:(U(k+1)-U*)-cross} into  \eqref{eq:U(k+1)-U*}, it holds that
\begin{align} \label{eq:(U1(k+1)-U1*)-main}
	\|\bm{U}_1^{k+1} - \bm{U}_1^*\bm{R}_1^k\|_{\Fr}^2 \leq \|\bm{U}_1^{k} - \bm{U}_1^*\bm{R}_1^k\|_{\Fr}^2  + \eta^2 Q_{1,1} - 2\eta Q_{1,2}.
\end{align}
We can similarly define $Q_{2,1}$ and $Q_{2,2}$ such that
\begin{align} \label{eq:(U2(k+1)-U2*)-main}
	\|\bm{U}_2^{k+1} - \bm{U}_2^*\bm{R}_2^k\|_{\Fr}^2 \leq \|\bm{U}_2^{k} - \bm{U}_2^*\bm{R}_2^k\|_{\Fr}^2  + \eta^2 Q_{2,1} - 2\eta Q_{2,2}.
\end{align}

We now return to deal with the last component in \eqref{eq:tildeE(k+1)}. Let $\nabla_{\mathcal{G}}\mathcal{L}(\cm{A})$ be the partial derivative of the loss function $\mathcal{L}(\cm{A})$ with respect to $\cm{G}$, and it holds that $\nabla_{\mathcal{G}}\mathcal{L}(\cm{A}) = \nabla \mathcal{L}(\cm{A}) \times_1 \bm{U}_1^\prime \times_2 \bm{U}_2^\prime$. Then the gradient descent update of $\cm{\widetilde{G}}_{\breve{S}}^{k+1}$ gives
\begin{align} \label{eq:tildeG(k+1)-G*}
	\|\cm{\widetilde{G}}_{\breve{S}}^{k+1} - \cm{G}^* \times_1 (\bm{R}_1^k)^{ \prime} \times_2 (\bm{R}_2^k)^{\prime} \|_{\Fr}^2 &= \|\cm{G}^k - \eta[\nabla_{\mathcal{G}}\mathcal{L}(\cm{A}^k)]_{\breve{S}} - \cm{G}^*\times_1(\bm{R}_1^{k})^\prime \times_2 (\bm{R}_2^{k})^\prime\|_{\Fr}^2 \notag\\
	&=\|\cm{G}^k - \cm{G}^*\times_1(\bm{R}_1^{k})^\prime \times_2 (\bm{R}_2^{k})^\prime\|_{\Fr}^2 + \eta^2\|[\nabla_{\mathcal{G}}\mathcal{L}(\cm{A}^k)]_{\breve{S}}\|_{\Fr}^2\notag\\
	&\hspace{5mm}-2\eta\langle \nabla_{\mathcal{G}}\mathcal{L}(\cm{A}^k), \cm{G}^k - \cm{G}^*\times_1(\bm{R}_1^{k})^\prime \times_2 (\bm{R}_2^{k})^\prime\rangle.
\end{align}
For the second term at \eqref{eq:tildeG(k+1)-G*}, by the definition of dual norm, $|\breve{S}| \leq 3s$ and the conditions at \eqref{eq:U&tildeG-upperbounds},
\begin{align} \label{eq:(tildeG(k+1)-G*)-quadratic}
	\|[\nabla_{\mathcal{G}}\mathcal{L}(\cm{A}^k)]_{\breve{S}}\|_{\Fr}^2 
	&= \sup_{\cmt{M}\in\mathbb{R}^{r_1\times r_2\times T_0},\|\cmt{M}\|_{\Fr}=1}\langle[\nabla_{\mathcal{G}}\mathcal{L}(\cm{A}^k)]_{\breve{S}}, \cm{M}\rangle^2\notag\\
	&\leq \sup_{\cmt{M}\in\mathbb{R}^{r_1\times r_2\times T_0},\|\cmt{M}\|_{\Fr}=1,\|\cmt{M}\|_0\leq 3s}\langle\nabla\mathcal{L}(\cm{A}^k), \cm{M}\times_1\bm{U}_1^k\times_2\bm{U}_2^k\rangle^2\notag\\
	&\leq 2\sup_{\cmt{M}\in\mathbb{R}^{r_1\times r_2\times T_0},\|\cmt{M}\|_{\Fr}=1,\|\cmt{M}\|_0\leq 3s}\langle\nabla\mathcal{L}(\cm{A}^*), \cm{M}\times_1\bm{U}_1^k\times_2\bm{U}_2^k\rangle^2 \notag\\
	&\hspace{5mm}+ 2\sup_{\cmt{M}\in\mathbb{R}^{r_1\times r_2\times T_0},\|\cmt{M}\|_{\Fr}=1,\|\cmt{M}\|_0\leq 3s}\langle\nabla\mathcal{L}(\cm{A}^k) - \nabla\mathcal{L}(\cm{A}^*), \cm{M}\times_1\bm{U}_1^k\times_2\bm{U}_2^k\rangle^2 \notag\\
	&\leq 2C_3^2\left(e_\mathrm{stat}^2+\|\nabla\mathcal{L}(\cm{A}^k)-\nabla\mathcal{L}(\cm{A}^*)\|_{\Fr}^2\right) =: Q_{\mathcal{G},1},
\end{align}
with $C_3 = 1.5\sigma_U^{1/2}$. For the last term at \eqref{eq:tildeG(k+1)-G*}, let $\cm{A}_{\mathcal{G}} = (\cm{G} - \cm{G}^*\times_1\bm{R}_1^\prime \times_2 \bm{R}_2^\prime)\times_1\bm{U}_1\times_2\bm{U}_2$ and, by the conditions at \eqref{eq:U&tildeG-upperbounds}, $\|\cm{A}_{\mathcal{G}}^k\|_{\Fr} \leq C_3\|\cm{G}^k - \cm{G}^*\times_1(\bm{R}_1^{k})^\prime \times_2 (\bm{R}_2^{k})^\prime\|_{\Fr}$. From the definition of $\nabla_{\mathcal{G}}\mathcal{L}(\cm{A}^k)$, it holds that
\begin{align*}
	\langle \nabla_{\mathcal{G}}\mathcal{L}(\cm{A}^k), \cm{G}^k - \cm{G}^*\times_1(\bm{R}_1^{k})^\prime \times_2 (\bm{R}_2^{k})^\prime\rangle
	&=
	\langle \nabla\mathcal{L}(\cm{A}^k) - \nabla\mathcal{L}(\cm{A}^*), \cm{A}_{\mathcal{G}}^k\rangle + \langle \nabla\mathcal{L}(\cm{A}^*), \cm{A}_{\mathcal{G}}^k\rangle.
\end{align*}
By the conditions at \eqref{eq:U&tildeG-upperbounds}, $|\breve{S}| \leq 3s$ and that $xy \leq 0.5x^2 + 0.5y^2$, 
\begin{align*}
	\langle \nabla\mathcal{L}(\cm{A}^*),  \cm{A}_{\mathcal{G}}^k\rangle &\leq  \sup_{\cmt{M}\in\bm{\Theta}^{\mathrm{SP}}_{1}(r_1,r_2,2s)}\langle\nabla\mathcal{L}(\cm{A}^*),\cm{M}\rangle\cdot\|\cm{A}_{\mathcal{G}}^k\|_{\Fr}\\
	&\leq 0.5C_3^2 C_1^{-1}  e_\mathrm{stat}^2 + 0.5 C_1\|\cm{G}^k - \cm{G}^*\times_1(\bm{R}_1^{k})^\prime \times_2 (\bm{R}_2^{k})^\prime\|_{\Fr}^2,
\end{align*}
which leads to 
\begin{align} \label{eq:(tildeG(k+1)-G*)-cross}
	\begin{split}
		&\langle \nabla_{\mathcal{G}}\mathcal{L}(\cm{A}^k), \cm{G}^k - \cm{G}^*\times_1(\bm{R}_1^{k})^\prime \times_2 (\bm{R}_2^{k})^\prime\rangle\\
		&\hspace{3mm}\geq \langle \nabla\mathcal{L}(\cm{A}^k) - \nabla\mathcal{L}(\cm{A}^*), \cm{A}_{\mathcal{G}}^k\rangle - C_1\|\cm{G}^k - \cm{G}^*\times_1(\bm{R}_1^{k})^\prime \times_2 (\bm{R}_2^{k})^\prime\|_{\Fr}^2- C_3^2 C_1^{-1}  e_\mathrm{stat}^2 =: Q_{\mathcal{G},2}.
	\end{split}
\end{align}
By plugging \eqref{eq:(tildeG(k+1)-G*)-quadratic} and \eqref{eq:(tildeG(k+1)-G*)-cross} into \eqref{eq:tildeG(k+1)-G*}, we can obtain that
\begin{align}\label{eq:(tildeG(k+1)-G*)-main}
	\|\cm{\widetilde{G}}_{\breve{S}}^{k+1} - \cm{G}^* \times_1 (\bm{R}_1^k)^{ \prime} \times_2 (\bm{R}_2^k)^{\prime} \|_{\Fr}^2 \leq \|\cm{G}^k - \cm{G}^*\times_1(\bm{R}_1^{k})^\prime \times_2 (\bm{R}_2^{k})^\prime\|_{\Fr}^2 + \eta^2 Q_{\mathcal{G},1} - 2\eta Q_{\mathcal{G},2},
\end{align}
which, together with \eqref{eq:(U(k+1)-U*)-cross} and \eqref{eq:(U1(k+1)-U1*)-main}, leads to \eqref{eq:step2}.

\noindent\textbf{Step 3} This step aims to develop a lower bound for $Q_{\mathcal{G},2} + \sum_{i=1}^{2}Q_{i,2}$ at \eqref{eq:step2}.
From \eqref{eq:(U(k+1)-U*)-cross}, a similar definition for $Q_{2,2}$, \eqref{eq:(tildeG(k+1)-G*)-cross} and the definition of $E^k$, we have
\begin{align} \label{eq:Q_{k,2}}
	Q_{\mathcal{G},2} + \sum_{i=1}^{2}Q_{i,2} \geq& \langle \nabla\mathcal{L}(\cm{A}^k)-\nabla\mathcal{L}(\cm{A}^*), \cm{A}_{\mathcal{G}}^k + \sum_{i=1}^{2}\cm{A}_{U_i}^k\rangle - (1+a)C_1 E^k \notag\\
	& + 0.25a\sum_{i=1}^{2}\|\bm{U}_i^{k\prime}\bm{U}_i^k - b^2\bm{I}_{r_i}\|_{\Fr}^2 - C_4e_\mathrm{stat}^2,
\end{align}
where $C_4 = (C_2^2 + C_3^2)C_1^{-1}$. By the conditions at \eqref{eq:U&tildeG-upperbounds} and \eqref{eq:E-upperbounds}, Lemma \ref{lemma:Hepsilon} holds with $B \leq (E^k)^{1/2} \leq (C_1E^k)^{1/4}$. Then, by plugging $c_e = 0.1$ to Lemma \ref{lemma:Hepsilon}, the first term becomes
\begin{align*}
	\langle \nabla\mathcal{L}(\cm{A}^k)-\nabla\mathcal{L}(\cm{A}^*), \cm{A}_{\mathcal{G}}^k + \sum_{i=1}^{2}\cm{A}_{U_i}^k\rangle =&\langle \nabla\mathcal{L}(\cm{A}^k)-\nabla\mathcal{L}(\cm{A}^*), \cm{A}^k-\cm{A}^*\rangle\\
	& + \langle \nabla\mathcal{L}(\cm{A}^k)-\nabla\mathcal{L}(\cm{A}^*), \cm{H}_\epsilon^k\rangle,
\end{align*} 
where $\|\cm{H}_\epsilon^k\|_{\Fr} \leq 3(\sigma_U^{1/2} + \sigma_U^{1/4})(C_1 E^k)^{1/2}$. Given \eqref{eq:RGC} holds for some given $\alpha, \beta > 0$, we can use $xy \leq 0.5x^2 + 0.5y^2$ to further lower bound the above term,  

\begin{align*}
	\langle \nabla\mathcal{L}(\cm{A}^k)-\nabla\mathcal{L}(\cm{A}^*), \cm{A}_{\mathcal{G}}^k + \sum_{i=1}^{2}\cm{A}_{U_i}^k\rangle &\geq \alpha \|\cm{A}^k-\cm{A}^*\|_{\Fr}^2\\
	&\hspace{5mm} + 0.5\beta\|\nabla\mathcal{L}(\cm{A}^k)-\nabla\mathcal{L}(\cm{A}^*)\|_{\Fr}^2 - 0.5\beta^{-1}\|\cm{H}_\epsilon^k\|_{\Fr}^2\\
	&\geq   \frac{\alpha\sigma_L^{3/2}}{\kappa^{1/2}}C_5E^k - \alpha C_{U,1}^{-1}C_{U,2}\sum_{i=1}^{2}\|\bm{U}_i^{k\prime}\bm{U}_i^k - b^2\bm{I}_{r_i}\|_{\Fr}^2\\
	%\alpha c_{\sigma} \frac{\sigma_L^{3/2}}{\kappa^{1/2}} E^k - 2 \alphac_{\sigma} \frac{\sigma_L^{3/2}}{\kappa^{1/2}}(\sigma_U^{-1/2} + 4(1.1)^4)\sum_{i=1}^{2}\|\bm{U}_i^{k\prime}\bm{U}_i^k - b^2\bm{I}_{r_i}\|_{\Fr}^2\\
	&\hspace{5mm}+0.5\beta\|\nabla\mathcal{L}(\cm{A}^k)-\nabla\mathcal{L}(\cm{A}^*)\|_{\Fr}^2 - 0.5\beta^{-1}\|\cm{H}_\epsilon^k\|_{\Fr}^2,
\end{align*}
where the second inequality is obtained from \eqref{eq:A->E->A} with $C_5 = [\sigma_L^{3/2}\kappa^{-1/2}C_{U,1}]^{-1}$.
Then, this jointly with $\|\cm{H}_\epsilon^k\|_{\Fr} \leq 3(\sigma_U^{1/2} + \sigma_U^{1/4})(C_1 E^k)^{1/2}$ can be used to further lower bound the inequality at \eqref{eq:Q_{k,2}}, which becomes
\begin{align*}
	Q_{\mathcal{G},2} + \sum_{i=1}^{2}Q_{i,2} \geq& \left( \frac{\alpha\sigma_L^{3/2}}{\kappa^{1/2}}C_5 -[1+a+5\beta^{-1}(\sigma_U^{1/2}+\sigma_U^{1/4})^2]C_1 \right) E^k\notag\\
	& + [0.25a - \alpha C_{U,1}^{-1}C_{U,2}]\sum_{i=1}^{2}\|\bm{U}_i^{k\prime}\bm{U}_i^k - b^2\bm{I}_{r_i}\|_{\Fr}^2 \notag\\
	&+0.5\beta\|\nabla\mathcal{L}(\cm{A}^k)-\nabla\mathcal{L}(\cm{A}^*)\|_{\Fr}^2- C_4e_\mathrm{stat}^2.
\end{align*}
Then, we choose $c_0>0$ contained in the constant term $C_1$ small enough such that $[1+a+5\beta^{-1}(\sigma_U^{1/2}+\sigma_U^{1/4})^2]C_1 \leq 0.5 \alpha\sigma_L^{3/2}\kappa^{-1/2}C_5$. From the definition right after \eqref{eq:A->E->A}, we have $C_{U,1} \geq 3\sigma_U^{-1}$ and $C_{U,2}\leq 10(\sigma_U^{-1/2}+1)$, and they lead to $C_{U,1}^{-1}C_{U,2} \leq 4(\sigma_U^{1/2}+\sigma_U)$. Let $a = 80\alpha (\sigma_U^{1/2}+\sigma_U)$ leading to $\alpha C_{U,1}^{-1}C_{U,2} \leq 0.05a$, and the above inequality takes the simple form
\begin{align} \label{eq:Q-2}
	\begin{split}
		Q_{\mathcal{G},2} + \sum_{i=1}^{2}Q_{i,2} \geq&  \frac{\alpha\sigma_L^{3/2}}{2\kappa^{1/2}}C_5E^k + 0.2a \sum_{i=1}^{2}\|\bm{U}_i^{k\prime}\bm{U}_i^k - b^2\bm{I}_{r_i}\|_{\Fr}^2  -C_4e_\mathrm{stat}^2\\
		&+0.5\beta\|\nabla\mathcal{L}(\cm{A}^k)-\nabla\mathcal{L}(\cm{A}^*)\|_{\Fr}^2.
	\end{split}
\end{align}

\noindent\textbf{Step 4} 
This step uses the intermediate results at Steps 2 and 3 to bound $\widetilde{E}^{k+1}$ with $E^{k}$. From \eqref{eq:(U(k+1)-U*)-quadratic}, a similar definition for $Q_{2,1}$ and \eqref{eq:(tildeG(k+1)-G*)-quadratic}, we have 
\begin{align*}
	Q_{\mathcal{G},1} + \sum_{i=1}^{2}Q_{i,1} =& (8C_2^2 + 2C_3^2)\left\{\|\nabla\mathcal{L}(\cm{A}^k)-\nabla\mathcal{L}(\cm{A}^*)\|_{\Fr}^2 + e_\mathrm{stat}^2\right \}\\
	&+3\sigma_U^{1/2}a^2\sum_{i=1}^{2}\|\bm{U}_i^{k\prime}\bm{U}_i^k - b^2\bm{I}_{r_i}\|_{\Fr}^2,
\end{align*}
which, together with \eqref{eq:Q-2}, implies that 
\begin{align}\label{eq:Q-main}
	\begin{split}
		&\eta^2(Q_{\mathcal{G},1} + \sum_{i=1}^{2}Q_{i,1}) - 2\eta (Q_{\mathcal{G},2} + \sum_{i=1}^{2}Q_{i,2}) \leq -\frac{\eta\alpha\sigma_L^{3/2}}{\kappa^{1/2}}C_5 E^k \\
		&\hspace{60mm}+[\eta^2(8C_2^2 + 2C_3^2) + 2\eta C_4]e_\mathrm{stat}^2\\
		&\hspace{60mm}+\left(3\eta^2\sigma_U^{1/2}a^2 - 0.4\eta a \right)\sum_{i=1}^{2}\|\bm{U}_i^{k\prime}\bm{U}_i^k - b^2\bm{I}_{r_i}\|_{\Fr}^2\\
		&\hspace{60mm}+\left(\eta^2(8C_2^2+2C_3^2) - \eta \beta\right)\|\nabla\mathcal{L}(\cm{A}^k)-\nabla\mathcal{L}(\cm{A}^*)\|_{\Fr}^2.
	\end{split}
\end{align}

Recall that $C_2=1.5\sigma_U^{3/4}$, $C_3=1.5\sigma_U^{1/2}$ and $a=80\alpha (\sigma_U^{1/2}+\sigma_U)$, and it holds that $8C_2^2 + 2C_3^2 \leq 18\sigma_U(1+\sigma_U^{1/2})$ and $\sigma_U^{1/2}a \leq 80\alpha \sigma_U(1+\sigma_U^{1/2}) \leq 20\beta^{-1}\sigma_U(1+\sigma_U^{1/2})$, where the second inequality comes from $\alpha\beta\leq 0.25$; see the discussion after \eqref{eq:RGC}. Set $\eta = \eta_0\beta[(1+\sigma_U)(1+\sigma_U^{1/2})]^{-1}$ and $\eta_0 \leq 1/150$, and it can be verified that
\begin{align}\label{eq:negative-cond}
	\begin{split}
		3\eta^2\sigma_U^{1/2}a^2 - 0.4\eta a  \leq 0\hspace{2mm}\text{and}\hspace{2mm}\eta^2(8C_2^2+2C_3^2) - \eta \beta \leq 0.
	\end{split}
\end{align}
%Recall $C_{U,1}=3\sigma_U^{-1} + 8\sigma_L^{-2}\sigma_U^{-1/2}+40\sigma_L^{-2}$, and $C_5 = [\sigma_L^{3/2}\kappa^{-1/2}C_{U,1}]^{-1}$. 
Again, note that $\eta = \eta_0\beta[(1+\sigma_U)(1+\sigma_U^{1/2})]^{-1}$ with $\eta_0 \leq 1/150$ and $\eta(C_2^2 + C_3^2) \leq 0.02\beta$, leading to $\eta^2(8C_2^2 + 2C_3^2) \leq 0.02\beta^2$ and $\eta C_4 =  \eta(C_2^2 + C_3^2)C_1^{-1} \leq 0.02\beta C_1^{-1}$. This, together with \eqref{eq:step2}, \eqref{eq:Q-main} and \eqref{eq:negative-cond}, implies that
\begin{align}\label{eq:tildeE(k+1)<E(k)-sub}
	\widetilde{E}^{k+1} \leq \left(1 - \eta_0\alpha\beta C_6 \right)E^k +  0.02\beta(\beta+C_1^{-1})e_\mathrm{stat}^2,
\end{align}
where 
$C_6 = \sigma_L^{3/2}\kappa^{-1/2}(1+\sigma_U)^{-1}(1+\sigma_U^{1/2})^{-1} C_5$, and
\begin{equation*}
	\eta_0\alpha\beta C_6 \leq \frac{\eta_0 \alpha\beta}{C_{U,1}} < \frac{\eta_0\sigma_U}{204}<1
\end{equation*}
since $\eta_0 < 204\sigma_U^{-1}$ and $\alpha\beta\leq 0.25$.
Moreover,
\begin{equation*}\label{eq:cvg-constant}
	\alpha \beta C_6 = \frac{\alpha\beta}{(1+\sigma_U)(1+\sigma_U^{1/2})C_{U,1}}\geq
	\frac{\alpha\beta}{204\kappa^2} :=\delta_{\alpha,\beta},
\end{equation*}
which, together with \eqref{eq:tildeE(k+1)<E(k)-sub}, implies that
\begin{equation}\label{eq:tildeE(k+1)<E(k)}
	\widetilde{E}^{k+1} \leq \left(1 - \eta_0 \delta_{\alpha,\beta} \right)E^k +  0.02\beta(\beta+C_1^{-1})e_\mathrm{stat}^2.
\end{equation}

\noindent\textbf{Step 5} 
This step upper bounds $E^{k+1}$ by $\widetilde{E}^{k+1}$, and hence by ${E}^{k}$.
To begin with, we first establish the inequality between $E^{k+1}$ and $\widetilde{E}^{k+1}$,
\begin{equation}
	E^{k+1}
	\leq  \sum_{i=1}^{2} \|\bm{U}_i^{k+1} - \bm{U}_i^*\bm{\widetilde{R}}_i^{k+1}\|_{\Fr}^2 + \|\cm{G}^{k+1} - \cm{G}^* \times_1 (\bm{\widetilde{R}}_1^{ k+1})^{ \prime} \times_2 (\bm{\widetilde{R}}_2^{ k+1})^{ \prime} \|_{\Fr}^2,\label{eq:E(k+1)}
\end{equation}
and $\widetilde{E}^{k+1} = \sum_{i=1}^{2} \|\bm{U}_i^{k+1} - \bm{U}_i^*\bm{\widetilde{R}}_i^{k+1}\|_{\Fr}^2 + \|\cm{\widetilde{G}}^{k+1}_{\breve{S}} - \cm{G}^* \times_1 (\bm{\widetilde{R}}_1^{ k+1})^{ \prime} \times_2 (\bm{\widetilde{R}}_2^{ k+1})^{ \prime} \|_{\Fr}^2$, where $\breve{S} = S_{k}\cup S_{k+1} \cup S_{\gamma}$. Denote by $\breve{s}$ the cardinality of $\breve{S}$, i.e., $\breve{s} = |S_{k+1} \cup S_{k} \cup S_{\gamma}|$.
Note that  $2xy\leq \mu x^2 + \mu^{-1}y^2$ for any $\mu>0$ and $x,y\in\mathbb{R}$, and then the second term at the right hand side of \eqref{eq:E(k+1)} satisfies the following inequality,
\begin{align}
	\begin{split}\label{eq:G(k+1)-G*}
		\|\cm{G}^{k+1} - &\cm{G}^* \times_1 (\bm{\widetilde{R}}_1^{ k+1})^{ \prime} \times_2 (\bm{\widetilde{R}}_2^{ k+1})^{ \prime}\|_{\Fr}^2 \\ &\leq 	(1+\mu_s)\|\cm{\widetilde{G}}^{k+1}_{\breve{S}} - \cm{G}^* \times_1 (\bm{\widetilde{R}}_1^{ k+1})^{ \prime} \times_2 (\bm{\widetilde{R}}_2^{ k+1})^{ \prime}\|_{\Fr}^2 + (1+\frac{1}{\mu_s})\|\cm{G}^{k+1} - \cm{\widetilde{G}}^{k+1}_{\breve{S}}\|_{\Fr}^2,
	\end{split}
\end{align}
where $\mu_s = \sqrt{(\breve{s}-s)/(\breve{s}-s_{\gamma})}<1$. 

Moreover, $\cm{A}^{k+1}=\cm{G}^{k+1}\times_1\bm{U}_1^{k+1}\times_2\bm{U}_2^{k+1}$, $\cm{\widetilde{A}}^{k+1}=\cm{\widetilde{G}}^{k+1}_{\breve{S}}\times_1\bm{U}_1^{k+1}\times_2\bm{U}_2^{k+1}$, and $\cm{A}^{k+1} = \hardt{\cm{\widetilde{A}}_{\breve{S}}^{k+1},s}$. Since $s_{\gamma} < s \leq \breve{s}$, it can be verified that
\begin{align} \label{eq:G(k+1)-tildeG(k+1)}
	&\|\cm{G}^{k+1} - \cm{\widetilde{G}}^{k+1}_{\breve{S}}\|_{\Fr}^2 \leq \left[\prod_{i=1}^{2}\sigma_{\min}(\bm{U}_i^{k+1})\right]^{-2} \|\cm{A}^{k+1} - \cm{\widetilde{A}}^{k+1}_{\breve{S}}\|_{\Fr}^2\notag\\
	& \hspace{20mm}\overset{(\text{Lemma \ref{lemma:HTineq}})}{\leq} \mu_s^2\left[\prod_{i=1}^{2}\sigma_{\min}(\bm{U}_i^{k+1})\right]^{-2}\|\cm{\widetilde{A}}^{k+1}_{\breve{S}} - \cm{A}^*\|_{\Fr}^2  \overset{\text{\eqref{eq:U&tildeG-upperbounds}\&\eqref{eq:A->E->A}}}{\leq} C_7\mu_s^2 \widetilde{E}^{k+1},
\end{align}
which, together with \eqref{eq:E(k+1)} and \eqref{eq:G(k+1)-G*}, implies that
\begin{equation} \label{eq:E(k+1)<tildeE(k+1)}
	E^{k+1} \leq [1 + (1+2C_7)\mu_s] \widetilde{E}^{k+1},
\end{equation}
where $C_7 = 2\sigma_U^{-1}C_L^{-1}=10(1+2\sigma_U^{1/2})$.

Note that $\mu_s$ is a function of $\breve{s}$ with $\breve{s} \geq s$ and, by the condition of $|S_k \cup S_{k+1}| \leq (1+\nu) s$, it holds that $\breve{s} \leq (1+\nu_k)s+s_{\gamma}$. Note that $s_{\gamma}\leq \nu_k s$, and it can be verified that $\mu_s \leq \sqrt{2\nu_k/(1+\nu_k)} <  \sqrt{2\nu_k}$, and hence $(1+2C_7)\mu_s\leq \eta_0 \delta_{\alpha,\beta}$ as long as $\nu_k \leq (1/128) (3+5\sigma_U^{1/2})^{-2}\eta_0^2\delta_{\alpha,\beta}^2$.
As a result, from \eqref{eq:tildeE(k+1)<E(k)} and \eqref{eq:E(k+1)<tildeE(k+1)}, 
\begin{align}\label{eq:E-main}
	E^{k+1} \leq \left(1 - \eta_0^2 \delta_{\alpha,\beta}^2 \right)E^k +  C_8e_\mathrm{stat}^2,
\end{align}
where $C_8 =  0.02[1 + (1+2C_7)\mu_s]\beta(\beta+C_1^{-1})$.
By unfolding this iteration, we can obtain that
\begin{align*}
	E^{K} \leq  \left(1 - \eta_0^2 \delta_{\alpha,\beta}^2 \right)^K E^0 + \eta_0^{-2} \delta_{\alpha,\beta}^{-2}C_8e_\mathrm{stat}^2,
\end{align*}
which, together with \eqref{eq:A->E->A} and $\bm{U}_i^{0\prime}\bm{U}_i^{0} = b^2\bm{I}_{r_i}$ for $1\leq i \leq 2$, leads to
\begin{align}\label{eq:final}
	\|\cm{A}^K - \cm{A}^*\|_{\Fr}^2 &\leq C_L^{-1} E^K \leq  C_L^{-1}\left(1 - \eta_0^2 \delta_{\alpha,\beta}^2 \right)^K E^0 + C_L^{-1}C_8 \eta_0^{-2} \delta_{\alpha,\beta}^{-2} e_\mathrm{stat}^2\notag\\
	&\leq C_{U,1} C_L^{-1}\left(1 - \eta_0^2 \delta_{\alpha,\beta}^2 \right)^K\|\cm{A}^0 - \cm{A}^*\|_{\Fr}^2 + C_{U,1}C_L^{-1}C_8\eta_0^{-2} \delta_{\alpha,\beta}^{-2}  e_\mathrm{stat}^2.
\end{align}

\noindent\textbf{Step 6} (Verifying conditions at \eqref{eq:RGC}, \eqref{eq:U&tildeG-upperbounds} and \eqref{eq:E-upperbounds})
%In this step, we first show that the conditions at  \eqref{eq:RGC} holds with certain values of $\alpha$ and $\beta$. By plugging in $\alpha$ and $\beta$, we are able to renew conditions on some quantities ($a$,$\nu_k$), restate the results in \eqref{eq:E-main} and subsequently obtain the final results shown in the theorem.  Finally, we verify the inequalities \eqref{eq:U&tildeG-upperbounds} and \eqref{eq:E-upperbounds} to conclude the proof. 
We first show that the conditions at  \eqref{eq:RGC} hold with certain values of $\alpha$ and $\beta$.
From Lemmas \ref{lemma:RSC/RSM-algo} and \ref{lemma:RE->RGC}, if $T_1\gtrsim s^2(N + \log T_0)$ and we choose
\begin{equation*} \label{eq:alpha_beta}
	\alpha = \frac{3\kappa_{\mathrm{RSC}}\kappa_{\mathrm{RSS}}}{\kappa_{\mathrm{RSC}} + 3\kappa_{\mathrm{RSS}}}\hspace{4mm}\text{and}\hspace{4mm}\beta = \frac{1}{\kappa_{\mathrm{RSC}} + 3\kappa_{\mathrm{RSS}}},
\end{equation*}
then the inequality at \eqref{eq:RGC} holds with probability at least $1-Ce^{ - N - \log T_0}$.
Note that $\kappa_{\mathrm{RSC}}\leq \kappa_{\mathrm{RSS}}$ and $\alpha\beta\leq0.25$, and it can be further verified that
\begin{align}\label{eq:alphaxbeta}
	\beta \leq 0.25\kappa_{\mathrm{RSC}}^{-1} \hspace{5mm} \text{and} \hspace{5mm}\frac{3\kappa_{\mathrm{RSC}}}{16\kappa_{\mathrm{RSS}}} \leq \alpha\beta \leq \frac{1}{4},
\end{align}
which can be used to update quantities, $\delta_{\alpha,\beta}$, $a$ and $\nu_k$ in Steps 4, 3 and 5, respectively. 
Specifically,
\[
\delta_{\alpha,\beta} = \frac{\alpha\beta}{204\kappa^2}\geq \frac{\kappa_{\mathrm{RSC}}}{1088\kappa_{\mathrm{RSS}}\kappa^2}:=\delta,
\]
\begin{equation*}
	a =80\alpha (\sigma_U^{1/2}+\sigma_U) = \frac{240(\sigma_U^{1/2}+\sigma_U)}{\kappa_{\mathrm{RSS}}^{-1} + 3\kappa_{\mathrm{RSC}}^{-1}}\hspace{5mm}\text{and}\hspace{5mm}\nu_k \leq  10^{-10}\cdot \frac{\eta_0^2\kappa_{\mathrm{RSC}}^2}{\kappa_{\mathrm{RSS}}^2\kappa^4}, 
\end{equation*}
where the final inequality ensures that $\nu_k \leq (1/128)\eta_0^2\delta_{\alpha,\beta}^2(3+5\sigma_U^{1/2})^{-2}$.
Moreover, since $\kappa_{\mathrm{RSC}}<1$ and $c_0 < 1$,
\begin{equation}\label{eq:C8}
	C_8 = 0.02[1 + (1+2C_7)\mu_s]\beta(\beta+C_1^{-1}) \leq 0.62 c_0^{-1} \kappa^{2}\kappa_{\mathrm{RSC}}^{-2},
\end{equation}
which, together with the fact that $ C_{U,1} C_L^{-1} \lesssim \kappa^{3/2}\sigma_L^{-1/2} $, can be used to rewrite \eqref{eq:final} into
\begin{equation*}
	\|\cm{A}^K - \cm{A}^*\|_{\Fr}^2 
	\lesssim \kappa^{3/2}\sigma_L^{-1/2} \left(1 - \eta_0^2 \delta^2 \right)^K\|\cm{A}^0 - \cm{A}^*\|_{\Fr}^2 + \kappa^{7/2}\sigma_L^{-1/2}\kappa_{\mathrm{RSC}}^{-2} \eta_0^{-2} \delta^{-2}  e_\mathrm{stat}^2.
\end{equation*}

We next verify \eqref{eq:E-upperbounds}. Note that $\|\bm{U}_i^{0\prime} \bm{U}_i^0 -b^2\bm{I}_{r_i}\|_{\Fr}^2 = 0$ for $1\leq i\leq 2$ and, by the initialization error bound and \eqref{eq:A->E->A}, it holds that
\begin{align}
	E^0 \leq C_{U,1}\|\cm{A}^0 - \cm{A}^*\|_{\Fr}^2 \leq c_0 \frac{\sigma_L^{1/2}}{\kappa^{3/2}}.
\end{align}
Suppose that the above inequality holds for $E^k$.
Then, by \eqref{eq:E-main} and \eqref{eq:C8},
\begin{align*}
	E^{k+1} &\leq \left(1 - \eta_0^2 \delta^2 \right)E^k +  0.62 c_0^{-1} \kappa^{2}\kappa_{\mathrm{RSC}}^{-2} e_\mathrm{stat}^2\\
	&\leq \left(1 - \eta_0^2 \delta^2 \right)\cdot c_0 \frac{\sigma_L^{1/2}}{\kappa^{3/2}} +  0.62 c_0^{-1} \kappa^{2}\kappa_{\mathrm{RSC}}^{-2}e_\mathrm{stat}^2\\
	&= c_0 \frac{\sigma_L^{1/2}}{\kappa^{3/2}} - \left( c_0 \frac{\eta_0^2 \delta^2\sigma_L^{1/2}}{\kappa^{3/2}}  - 0.62 c_0^{-1} \kappa^{2}\kappa_{\mathrm{RSC}}^{-2} e_\mathrm{stat}^2 \right)\leq c_0 \frac{\sigma_L^{1/2}}{\kappa^{3/2}}
\end{align*}
since, when $e_\mathrm{stat}^2 \leq (1/800) c_0^2  \eta_0^2\kappa^{-8}\kappa_{\mathrm{RSS}}^{-4}\kappa_{\mathrm{RSC}}^4$, 
\[
c_0 \frac{\eta_0^2 \delta^2\sigma_L^{1/2}}{\kappa^{3/2}} \geq 0.62 c_0^{-1} \kappa^{2}\kappa_{\mathrm{RSC}}^{-2} e_\mathrm{stat}^2.
\]
Hence the inequality at \eqref{eq:E-upperbounds} holds. 

Finally, we verify \eqref{eq:U&tildeG-upperbounds}. Since $\kappa \geq 1$ and $b=\sigma_U^{1/4}$, it holds $E^k \leq c_0 {\sigma_L^{1/2}}{\kappa^{-3/2}} \leq c_0\sigma_U^{1/2}$ where $c_0<0.01$ is a very small number. From the definition of $E^k$, we can verify that, for $i=1$ or 2:
\begin{align*}
	\|\bm{U}_i^k\|_{\op} &\leq \|\bm{U}_i^*\bm{R}_i^k\|_{\op} + \|\bm{U}_i^k-\bm{U}_i^*\bm{R}_i^k\|_{\op} \leq \sigma_U^{1/4} + \|\bm{U}_i^k-\bm{U}_i^*\bm{R}_i^k\|_{\Fr} \leq 1.1\sigma_U^{1/4} ,\\
	\sigma_{\min}(\bm{U}_i^k) &\geq \|\bm{U}_i^*\bm{R}_i^k\|_{\op} - \|\bm{U}_i^k-\bm{U}_i^*\bm{R}_i^k\|_{\op} \geq \sigma_U^{1/4}  - \|\bm{U}_i^k-\bm{U}_i^*\bm{R}_i^k\|_{\Fr} \geq 0.9\sigma_U^{1/4} ,\hspace{2mm}\text{and}\\
	\|\cm{G}^k_{(i)}\|_{\op} &\leq \|\bm{R}_i^k\cm{G}^*_{(i)} (\bm{I}_{T_0}\otimes \bm{R}_{i-1}^k)^\prime\|_{\op} + \|\cm{G}^k_{(i)} -\bm{R}_i^k\cm{G}^*_{(i)} (\bm{I}_{T_0}\otimes \bm{R}_{i-1}^k)^\prime\|_{\op}\\
	&\leq \sigma_U^{1/2} + \sqrt{c_0}\sigma_U^{1/4}  \leq 1.1\sigma_U^{1/2}.
\end{align*}
We hence accomplish the whole proof.

%%%%%%%%%%%%%%%%%%%%%%%%%%%%%%%%%%%%%%%%%%%%%%%%%%%%%%%%%%%%%%%
\subsubsection{Proof of Corollary \ref{cor:algorithm}}

Note that, by Assumption \ref{assum:Adecay} and the low-rank conditions at \eqref{eq:trunc_tensor}, $\|\bm{A}_j^*\|_\Fr\leq C\sqrt{r_1\wedge r_2}\rho^j$ for $j\geq 1$. 
Let $Q_\gamma$ be the smallest integer such that $C\sqrt{r_1\wedge r_2}\rho^j\leq \gamma$ for all $j\geq Q_\gamma$, and it can be verified that $Q_\gamma =  \lceil \log(C\sqrt{r_1 \wedge r_2}/\gamma)/\log(1/\rho) \rceil$ and $\tau^2 Q_{\gamma}\lesssim 1$. As a result, by a method similar to \eqref{eq:approx_bound}, we can show that
\begin{align}\label{eq:approx_bound2}
	\|\cm{A}^*_{S_\gamma^c}\|_{\Fr}^2 \lesssim \gamma^2 Q_{\gamma}\hspace{2mm}\text{and}\hspace{2mm}\tau^2\|\cm{A}^*_{S_\gamma^c}\|_{\ddagger}^2 \lesssim \gamma^2 Q_{\gamma}.
\end{align}
Moreover, by choosing $\gamma \asymp \sqrt{\{(r_1\wedge r_2)N+\log T_0\}/{T_1}}$, we have $Q_{\gamma}\asymp \log T_1/\log(1/\rho)$. Let $s \asymp Q_{\gamma} \asymp \log T_1/ \log(1/\rho)$, and it is implied by Theorem \ref{thm:stat} that
\begin{align} \label{eq:e_stat^2}
	e_{\mathrm{stat}}^2  \lesssim \frac{[(r_1\wedge r_2)N+\log T_0]s}{T_1}.
\end{align}

%Note that the computational error of $\cm{A}_{\infty}^{K+1}$
%\begin{align*}
%	e_{\mathrm{est}}(\cm{A}_{\infty}^{K+1}) = \|\cm{A}^{K+1} - \cm{A}^*\|_{\Fr}^2 + e_{\trunc}.
%\end{align*} 
%The first term on the right hand side $\|\cm{A}^{K+1} - \cm{A}^*\|_{\Fr}^2$ contains the statistical error $e_{\mathrm{stat}}^2$, and dominates the truncation error, which by \eqref{eq:e_trunc1} is of order $(r_1\wedge r_2)T_1^{-4}$.
%Hence f
%From Theorem \ref{thm:stat}, for all $k = 1, 2, \dots,$
%\begin{align*}
%		e_{\mathrm{est}}(\cm{A}_{\infty}^{K+1}) \lesssim \left(1 - C\frac{\eta_0^2\kappa_{\mathrm{RSC}}^2\kappa_{\mathrm{RSS}}^{-2}}{\kappa^4} \right)^k\|\cm{A}^0 - \cm{A}^*\|_{\Fr}^2 + \|\cm{A}^*_{S_{\gamma}^c}\|_{\Fr}^2 + 
%	\frac{\kappa^4}{\eta_0^2\kappa_{\mathrm{RSC}}^2\kappa_{\mathrm{RSS}}^{-2}} e_{\mathrm{stat}}^2.
%\end{align*}

On the other hand, by plugging the values of $\gamma$ and $Q_{\gamma}$ into the first term in \eqref{eq:approx_bound2}, we have $\|\cm{A}^*_{S_{\gamma}^c}\|_{\Fr}^2 \leq Ce_{\mathrm{stat}}^2$ for an absolute constant $C$. 
Note that $\|\cm{A}^{K} - \cm{A}^*\|_{\Fr}^2 \leq 2 \|\cm{A}^{K} - \cm{A}_{S_{\gamma}}^*\|_{\Fr}^2 + 2 \|\cm{A}_{S_{\gamma}^c}^*\|_{\Fr}^2$ and, from Theorem \ref{thm:optimization},
\begin{align*}
	\|\cm{A}^{K} - \cm{A}^*\|_{\Fr}^2 \leq 2D_1 \left(1 - D_2 \right)^{K}\|\cm{A}^0 - \cm{A}^*_{S_{\gamma}}\|_{\Fr}^2 + 2(C+D_3) e_{\mathrm{stat}}^2,
\end{align*}
where $D_1=\kappa^{3/2}\sigma_L^{-1/2}, D_2=\eta_0^2\delta^2$, $D_3=\kappa^{7/2}\sigma_L^{-1/2}\kappa_{\mathrm{RSC}}^{-2}\eta_0^{-2}\delta^{-2}$, and $C\lesssim D_3$.

As a result, when 
\begin{align}\label{eq:K}
	K \geq \frac{\log(2D_3) + \log e_{\mathrm{stat}}^2 - \log D_1 - \log\|\cm{A}^0 - \cm{A}^*_{S_{\gamma}}\|_{\Fr}^2}{\log (1-D_2)},
\end{align}
the optimization error can be shown to be dominated by the statistical error, i.e.
%\begin{align*}
%		e_{\mathrm{est}}(\cm{A}_{\infty}^{K+1}) \lesssim \frac{[(r_1\wedge r_2)N+\log T_0]s}{T_1}.
%\end{align*}
\begin{align*}
	\|\cm{A}^{K} - \cm{A}^*\|_{\Fr}^2 \lesssim \frac{[(r_1\wedge r_2)N+\log T_0]s}{T_1}.
\end{align*}
Moreover, from \eqref{eq:e_stat^2}, $\log e_{\mathrm{stat}}^2$ can be upper-bounded by some absolute positive constant when $T_1 \gtrsim \{(r_1\wedge r_2)+s^2\}N+s^2\log T_0$.  
Recall that the absolute constant $\eta_0\leq 1/150$ and $\|\cm{A}^0 - \cm{A}^*_{S_\gamma}\|_{\Fr}^2 \lesssim \sigma_L^{5/2}\kappa^{-3/2}$, and the bound at \eqref{eq:K} can be further simplified into
\begin{equation*}
	K \gtrsim \frac{\log(\kappa^{7/2}\sigma_L^{-5/2}\kappa_{\mathrm{RSC}}^{-2}\delta^{-2})}{\log(1-\eta_0^2\delta^2)}.
\end{equation*} 
Hence, the proof of this corollary is accomplished.

%%%%%%%%%%%%%%%%%%%%%%%%%%%%%%%%%%%%%%%%%%%%%%%%%%%%%%%%%%%%%
\subsubsection{Five auxiliary lemmas}
This subsection gives five auxiliary lemmas used in the proof of Theorem \ref{thm:optimization}.

\begin{lemma}[Restricted strong convexity and smoothness conditions] \label{lemma:RSC/RSM-algo}
	Suppose that Assumptions \ref{assum:glp}--\ref{assum:T0} are satisfied. If $T_1\gtrsim s^2(N + \log T_0)$, then for any $\cm{A}_1, \cm{A}_2\in\bm{\Theta}^{\mathrm{SP}}(r_1,r_2,s)$, 
	\[
	0.5\kappa_{\mathrm{RSC}}\|\cm{A}_1 - \cm{A}_2\|_{\Fr}^2 \leq \mathcal{L}(\cm{A}_1) - \mathcal{L}(\cm{A}_2) - \langle\nabla\mathcal{L}(\cm{A}_2),\cm{A}_1 - \cm{A}_2 \rangle \leq 1.5\kappa_{\mathrm{RSS}} \|\cm{A}_1 - \cm{A}_2\|_{\Fr}^2
	\]
	holds with probability at least
	$1-Ce^{ - N - \log T_0}$, where $\kappa_{\mathrm{RSC}}$ and $\kappa_{\mathrm{RSS}}$ are defined in Theorem \ref{prop:main}.
\end{lemma}
\begin{proof}
	Recall that
	\begin{align*}
		\mathcal{L}(\cm{A}_1) - \mathcal{L}(\cm{A}_2) - \langle\nabla\mathcal{L}(\cm{A}_2),\cm{A}_1 - \cm{A}_2 \rangle = \frac{1}{2T_1}\|(\cm{A}_1-\cm{A}_2)_{(1)}\bm{X}\|_{\Fr}^2.
	\end{align*}
	Let $\bm{\Delta} = \cm{A}_1-\cm{A}_2$, and it holds that $\bm{\Delta} \in\bm{\Theta}^{\mathrm{SP}}(2r_1,2r_2,2s)$. Using this notation and ignoring the constant scaling, we have the inequalities
	\begin{align}\label{eq:E|DeltaX|<}
		\frac{1}{T_1}\|\bm{\Delta}_{(1)}\bm{X}\|_{\Fr}^2 \leq \frac{\mathbb{E}\left (\|\bm{\Delta}_{(1)}\bm{X}\|_{\Fr}^2 \right )}{T_1} + 	\frac{\left |\|\bm{\Delta}_{(1)}\bm{X}\|_{\Fr}^2 -  \mathbb{E}\left (\|\bm{\Delta}_{(1)}\bm{X}\|_{\Fr}^2 \right )\right |}{T_1} ,
	\end{align}
	and
	\begin{align}\label{eq:E|DeltaX|>}
		\frac{1}{T_1}\|\bm{\Delta}_{(1)}\bm{X}\|_{\Fr}^2 \geq \frac{\mathbb{E}\left (\|\bm{\Delta}_{(1)}\bm{X}\|_{\Fr}^2 \right )}{T_1} - 	\frac{\left |\|\bm{\Delta}_{(1)}\bm{X}\|_{\Fr}^2 -  \mathbb{E}\left (\|\bm{\Delta}_{(1)}\bm{X}\|_{\Fr}^2 \right )\right |}{T_1}.
	\end{align}
	
	First, by \eqref{eq:exp-min}, \cite{basu2015regularized} and Assumptions \ref{assum:Adecay} \& \ref{assum:error}, we have
	\[
	\kappa_{\mathrm{RSC}}  \|\bm{\Delta}\|_{\Fr}^2 \leq \frac{\mathbb{E}\left (\|\bm{\Delta}_{(1)}\bm{X}\|_{\Fr}^2 \right )}{T_1}
	=\trace \left (\bm{\Delta}_{(1)} \bm{\Sigma}_{T_0} \bm{\Delta}_{(1)}^\prime \right ) \leq \lambda_{\max}(\bm{\Sigma}_\varepsilon)\mu_{\max}(\bm{\Psi}_*) \leq \kappa_{\mathrm{RSS}}\|\bm{\Delta}\|_{\Fr}^2.
	\]
	Moreover, by a method similar to the proof of Lemma \ref{lemma:RE}, 
	\begin{align*}
		\frac{\left |\|\bm{\Delta}_{(1)}\bm{X}\|_{\Fr}^2 -  \mathbb{E}\left (\|\bm{\Delta}_{(1)}\bm{X}\|_{\Fr}^2 \right )\right |}{T_1} 
		& = \left| \trace \left \{\bm{\Delta}_{(1)} \left (\frac{\bm{X}\bm{X}^\prime}{T_1} - \bm{\Sigma}_{T_0}  \right ) \bm{\Delta}_{(1)}^\prime \right \} \right| \\
		& \leq \sum_{i=1}^{T_0}\sum_{j=1}^{T_0} \left|\trace \left [\bm{\Delta}_{i} \left \{\frac{\bm{X}_i\bm{X}_j^\prime}{T_1} -\bm{\Gamma}(i-j)\right \} \bm{\Delta}_{j}^\prime \right ] \right|\\
		&\leq \sum_{i=1}^{T_0}\sum_{j=1}^{T_0} \|\bm{\Delta}_{i}\|_{\Fr}\|\bm{\Delta}_{j}\|_{\Fr} \left \|\frac{\bm{X}_i\bm{X}_j^\prime}{T_1} -\bm{\Gamma}(i-j)\right \|_{\op}\\
		&\leq\|\bm{\Delta}\|_{\ddagger}^2 \max_{1\leq i\leq T_0}\max_{1\leq j\leq T_0} \left \|\frac{\bm{X}_i\bm{X}_j^\prime}{T_1}-\bm{\Gamma}(i-j)\right \|_{\op}.
	\end{align*}
	Note that $\|\bm{\Delta}\|_{\ddagger} \leq \sqrt{2s}\|\bm{\Delta}\|_{\Fr}$ and then, by Lemma \ref{lemma:RE1}, it can be verified that, when $T_1\gtrsim s^2(N + \log T_0)$, 
	\[
	\frac{\left |\|\bm{\Delta}_{(1)}\bm{X}\|_{\Fr}^2 -  \mathbb{E}\left (\|\bm{\Delta}_{(1)}\bm{X}\|_{\Fr}^2 \right )\right |}{T_1} \leq s\tau^2\|\bm{\Delta}\|_{\Fr}^2 \leq 0.5\kappa_{\mathrm{RSC}}\|\bm{\Delta}\|_{\Fr}^2
	\] 
	holds with probability at least $1-C e^{ - N - \log T_0 }$,	where  $\tau^2=C\sqrt{(N+\log T_0)/T_1}$. This, together with \eqref{eq:E|DeltaX|<} and \eqref{eq:E|DeltaX|>}, accomplishes the proof of this lemma.
\end{proof}

%%%%%%%%%%%%%%%%%%%%%%%%%%%%%%%%%%%%%%%%%%%%%%%%%%%%%%%%%%%%
\begin{lemma} \label{lemma:A->E->A}
	Consider two tensors $\cm{A}^*=\cm{G}^*\times_1\bm{U}_1^*\times_2\bm{U}_2^*\in\mathbb{R}^{N\times N\times T_0}$ and $\cm{A}=\cm{G}\times_1\bm{U}_1\times_2\bm{U}_2\in\mathbb{R}^{N\times N\times T_0}$, where $\cm{G}^*$ and $\cm{G}\in\mathbb{R}^{r_1\times r_2\times T_0}$ are core tensors, and $\bm{U}_i^*, \bm{U}_i\in\mathbb{R}^{N\times r_i}$ with $i=1$ and 2 are factor matrices. We define their distance below,
	\[
	E = \min_{\bm{R}_i\in \mathcal{O}^{r_i\times r_i},\; 1\leq i\leq 2}\left\{ \sum_{i=1}^{2}\|\bm{U}_i-\bm{U}_i^*\bm{R}_i\|_{\Fr}^2 + \|\cm{G} - \cm{G}^*\times_1\bm{R}_1^\prime\times_2\bm{R}_2^\prime\|_{\Fr}^2 \right\}.
	\]
	Suppose that, for $i=1$ and 2, $\|\bm{U}_i\|_{\op} \leq (1+c_e)\sigma_U^{1/4}$, $\|\cm{G}_{(i)}\|_{\op} \leq (1+c_e)\sigma_U^{1/2}$, $\bm{U}_i^{*\prime}\bm{U}_i^* = b^2 \bm{I}_{r_i}$, and $\sigma_L \leq \sigma_{\min}(\cm{A}^*_{(i)}) \leq \|\cm{A}^*_{(i)}\|_{\op} \leq \sigma_U$, where $c_e>0$, $b>0$, $0<\sigma_L \leq \sigma_U$, and $\sigma_{\min}(\bm{A})$ denotes the smallest nonzero singular value of matrix $\bm{A}$. It then holds that
	\[
	c_3 \|\cm{A} - \cm{A}^*\|_{\Fr}^2 \leq E \leq c_2\|\cm{A} - \cm{A}^*\|_{\Fr}^2 + 2(1+c_1)b^{-2}\sum_{i=1}^{2}\|\bm{U}_i^\prime\bm{U}_i - b^2 \bm{I}_{r_i}\|_{\Fr}^2,
	\]
	where $c_1 = 3(1+c_e)^4\sigma_U^{3/2} b^{-4}$, $c_2 = 3b^{-4} + 8(1+c_1)\sigma_L^{-2}b^{-2}$, and $c_3=[3(1+c_e)^4(\sigma_U+2\sigma_U^{3/2})]^{-1}$.
\end{lemma}
\begin{proof}
	We first prove the upper bound. For any $\bm{R}_i\in \mathcal{O}^{r_i\times r_i}$ with $i=1$ and 2,  
	$$\|\cm{G} - \cm{G}^*\times_1 \bm{R}_1^{\prime}\times_2 \bm{R}_2^{\prime}\|_{\Fr}^2 = b^{-4}\|\cm{G}\times_1 \bm{U}_1^* \bm{R}_1 \times_2 \bm{U}_2^*\bm{R}_2 - \cm{A}^*\|_{\Fr}^2,$$ 
	and, by the fact that $(x+y+z)^2 \leq 3x^2 + 3y^2 + 3z^2$,
	\begin{align*}
		&\|\cm{G}\times_1 \bm{U}_1^* \bm{R}_1 \times_2 \bm{U}_2^* \bm{R}_2 - \cm{A}^*\|_{\Fr}^2\\ 
		&\hspace{5mm}\leq 3\| \cm{A} - \cm{A}^*\|_{\Fr}^2  + 3\|\cm{G}\times_1 (\bm{U}_1 - \bm{U}_1^*\bm{R}_1) \times_2 \bm{U}_2\|_{\Fr}^2
		+ 3\|\cm{G}\times_1 \bm{U}_1^*\bm{R}_1 \times_2 (\bm{U}_2 - \bm{U}_2^* \bm{R}_2)\|_{\Fr}^2\\
		&\hspace{5mm}\leq 3\| \cm{A} - \cm{A}^*\|_{\Fr}^2  + 3(1+c_e)^4\sigma_U^{3/2} \sum_{i=1}^{2} \|\bm{U}_i - \bm{U}_i^*\bm{R}_i\|_{\Fr}^2.
	\end{align*}
	As a result,
	\begin{align}\label{eq:G-Ghat}
		\|\cm{G} - \cm{G}^*\times_1 \bm{R}_1^{\prime}\times_2 \bm{R}_2^{\prime}\|_{\Fr}^2 \leq 3b^{-4}\| \cm{A} - \cm{A}^*\|_{\Fr}^2  + c_1 \sum_{i=1}^{2} \|\bm{U}_i - \bm{U}_i^*\bm{R}_i\|_{\Fr}^2,
	\end{align}
	with $c_1 = 3(1+c_e)^4\sigma_U^{3/2} b^{-4}$, which implies that
	\begin{align}\label{eq:E<A+U}
		E \leq  3b^{-4}\| \cm{A} - \cm{A}^*\|_{\Fr}^2  +(1+c_1) \sum_{i=1}^{2} \min_{\bm{R}_i \in \mathcal{O}^{r_i\times r_i}} \|\bm{U}_i - \bm{U}_i^*\bm{R}_i\|_{\Fr}^2.
	\end{align}
	
	We next handle the second term of \eqref{eq:E<A+U}. For $i=1$ and 2, consider an SVD form $\bm{U}_i=\bm{\widebar{U}}_i\bm{\widebar{\Sigma}}_i\bm{\widebar{V}}_i^{\prime}$. Note that
	$\|\bm{U}_i - \bm{U}_i^*\bm{R}_i\|_{\Fr}^2\leq 2\|\bm{U}_i - b\bm{\widebar{U}}_i\bm{\widebar{V}}_i^{\prime}\|_{\Fr}^2+2\| b\bm{\widebar{U}}_i\bm{\widebar{V}}_i^{\prime} -\bm{U}_i^*\bm{R}_i\|_{\Fr}^2$, and then
	\begin{equation}\label{eq:U-U*}
		\begin{split}
			\min_{\bm{R}_i \in \mathcal{O}^{r_i\times r_i}} \|\bm{U}_i - \bm{U}_i^*\bm{R}_i\|_{\Fr}^2
			& \leq 2\| \bm{\widebar{\Sigma}}_i- b\bm{I}_{r_i}\|_{\Fr}^2 + 2 \min_{\bm{R}_i \in \mathcal{O}^{r_i\times r_i}} \|b\bm{\widebar{U}}_i - \bm{U}_i^*\bm{R}_i\|_{\Fr}^2,
		\end{split}
	\end{equation}
	where $\| \bm{\widebar{\Sigma}}_i- b\bm{I}_{r_i}\|_{\Fr}^2 \leq b^{-2}\|\bm{U}_i^\prime\bm{U}_i - b^2 \bm{I}_{r_i}\|_{\Fr}^2$; see (E.3) in \cite{HWZ21}.
	On the other hand, $\bm{\widebar{U}}_i$ and $b^{-1}\bm{U}_i^*$ have orthonormal columns and they span the left singular subspaces of $\cm{A}_{(i)}$ and $\cm{A}_{(i)}^*$, respectively. Then, from Lemma 1 in \cite{cai2018rate},
	\begin{align}\label{eq:U-U*2}
		\min_{\bm{R}_i \in \mathcal{O}^{r_i\times r_i}} \|b\bm{\widebar{U}}_i - \bm{U}_i^*\bm{R}_i\|_{\Fr}^2 \leq 2b^2\|\bm{\widebar{U}}_{i\perp}^\prime(b^{-1}\bm{U}_i^*)\|_{\Fr}^2 \leq 2b^{-2}\sigma_L^{-2} \|\cm{A} - \cm{A}^*\|_{\Fr}^2,
	\end{align}
	where $\bm{\widebar{U}}_{i\perp}\in \mathcal{O}^{N\times (N-r_i)}$ lies in the orthogonal complementary subspace of $\bm{\widebar{U}}_i$ and the last inequality is due to
	\begin{align*}
		\|\cm{A} - \cm{A}^*\|_{\Fr}^2 
		&= \|\cm{A}_{(i)} - \cm{A}_{(i)}^*\|_{\Fr}^2 
		\geq \|\bm{\widebar{U}}_{i\perp}^\prime\cm{A}_{(i)}^*\|_{\Fr}^2
		= b^4\|\bm{\widebar{U}}_{i\perp}^\prime(b^{-1}\bm{U}_i^*)(b^{-1}\bm{U}_i^{*\prime})\cm{A}_{(i)}^* \|_{\Fr}^2\\
		&\geq b^4 \sigma_L^{2} \|\bm{\widebar{U}}_{i\perp}^\prime(b^{-1}\bm{U}_i^*)\|_{\Fr}^2.
	\end{align*}
	From \eqref{eq:U-U*} and \eqref{eq:U-U*2}, we have 
	\begin{equation*}
		\min_{\bm{R}_i \in \mathcal{O}^{r_i\times r_i}}\|\bm{U}_i - \bm{U}_i^*\bm{R}_i\|_{\Fr}^2 \leq 2b^{-2}\|\bm{U}_i^\prime\bm{U}_i - b^2 \bm{I}_{r_i}\|_{\Fr}^2 + 4b^{-2}\sigma_L^{-2} \|\cm{A} - \cm{A}^*\|_{\Fr}^2,
	\end{equation*}
	which, together with \eqref{eq:E<A+U}, implies that
	\[
	E \leq c_2\|\cm{A} - \cm{A}^*\|_{\Fr}^2 + 2(1+c_1)b^{-2}\sum_{i=1}^{2}\|\bm{U}_i^\prime\bm{U}_i - b^2 \bm{I}_{r_i}\|_{\Fr}^2,
	\]
	where $c_2 = 3b^{-4} + 8(1+c_1)\sigma_L^{-2}b^{-2}$.
	
	We next prove the lower bound. Note that $\|\cm{G}_{(i)}^*\|_{\op} = b^{-2}\|\cm{A}_{(i)}^*\|_{\op} \leq \sigma_Ub^{-2}$ for $1\leq i\leq 2$, and then it holds that
	\begin{align*}
		\|\cm{A} - \cm{A}^*\|_{\Fr}^2 &\leq 3\|(\cm{G} - \cm{G}^*\times_1 \bm{R}_1^{\prime}\times_2 \bm{R}_2^{\prime} ) \times_1 \bm{U}_1 \times_2 \bm{U}_2 \|_{\Fr}^2\\
		&+ 3\|(\cm{G}^*\times_1 \bm{R}_1^{\prime}\times_2 \bm{R}_2^{\prime} )\times_1 (\bm{U}_1 - \bm{U}_1^*\bm{R}_1 ) \times_2 \bm{U}_2 \|_{\Fr}^2\\
		&+ 3\|(\cm{G}^* \times_2 \bm{R}_2^{\prime}) \times_1  \bm{U}_1^* \times_2 (\bm{U}_2 - \bm{U}_2^*\bm{R}_2 ) \|_{\Fr}^2\\
		&\leq c_3^{-1} E,
	\end{align*}
	where $c_3 = [3(1+c_e)^4(\sigma_U+2\sigma_U^{3/2})]^{-1}$.
\end{proof}

%%%%%%%%%%%%%%%%%%%%%%%%%%%%%%%%%%%%%%%%%%%%%%%%%%%%%%%%%%%%%%%%%%%%%%%%%%%%%%%%%%
\begin{lemma} \label{lemma:Hepsilon}
	Consider the two tensors, $\cm{A}^*=\cm{G}^*\times_1\bm{U}_1^*\times_2\bm{U}_2^*\in\mathbb{R}^{N\times N\times T_0}$ and $\cm{A}=\cm{G}\times_1\bm{U}_1\times_2\bm{U}_2\in\mathbb{R}^{N\times N\times T_0}$, in Lemma \ref{lemma:A->E->A}, and it holds that, for $i=1$ and 2, $\|\bm{U}_i\|_{\op} \leq (1+c_e)\sigma_U^{1/4}$, $\|\cm{G}_{(i)}\|_{\op} \leq (1+c_e)\sigma_U^{1/2}$, $\bm{U}_i^{*\prime}\bm{U}_i^* = b^2 \bm{I}_{r_i}$, and $ \|\cm{A}^*_{(i)}\|_{\op} \leq \sigma_U$, where $c_e>0$, $b>0$, and $0< \sigma_U$. Define three tensors below,
	\[
	\cm{A}_{\mathcal{G}} = (\cm{G}-\cm{G}^*\times_1\bm{R}_1^\prime\times_2\bm{R}_2^\prime)\times_1 \bm{U}_1 \times_2 \bm{U}_2, \hspace{5mm}\cm{A}_{U_1} = \cm{G} \times_1 (\bm{U}_1-\bm{U}_1^*\bm{R}_1) \times_2 \bm{U}_2,
	\] 
	and $\cm{A}_{U_2}=\cm{A}_{U_1} = \cm{G} \times_1 \bm{U}_1 \times_2 (\bm{U}_2-\bm{U}_2^*\bm{R}_2)$. 
	If there exists a $B\geq0$ such that 
	$ \|\cm{G}-\cm{G}^*\times_1\bm{R}_1^\prime\times_2\bm{R}_2^\prime\|_{\Fr} \leq B$ and $\|\bm{U}_i-\bm{U}_i^*\bm{R}_i\|_{\Fr} \leq B$ for $\bm{R}_i\in \mathcal{O}^{r_i\times r_i}$ with $i=1$ and 2, then
	\[
	\|\cm{H}_{\epsilon}\|_{\Fr} \leq [(1+c_e)\sigma_U^{1/2} + (2+c_e)\sigma_U^{1/4}]B^2 \hspace{2mm}\text{with}\hspace{2mm} \cm{H}_{\epsilon}
	= \cm{A}_{\mathcal{G}} + \sum_{i=1}^{2}\cm{A}_{U_i} - (\cm{A} - \cm{A}^*).
	\]
\end{lemma}
\begin{proof}
	Note that
	\[
	\cm{A}_{\mathcal{G}} + \sum_{i=1}^{2}\cm{A}_{U_i} = \cm{A} - \cm{A}^* + \underbrace{\cm{H}_{\epsilon}^{(1)} + \cm{H}_{\epsilon}^{(2)} + \cm{H}_{\epsilon}^{(3)}}_{=\cm{H}_{\epsilon}},
	\]
	where 
	\begin{align}
		\begin{split}
			&\cm{H}_{\epsilon}^{(1)} = \cm{G} \times_1 (\bm{U}_1-\bm{U}_1^*\bm{R}_1) \times_2 (\bm{U}_2-\bm{U}_2^*\bm{R}_2), \\
			&\cm{H}_{\epsilon}^{(2)} = (\cm{G}-\cm{G}^*\times_1\bm{R}_1^\prime\times_2\bm{R}_2^\prime) \times_1 (\bm{U}_1-\bm{U}_1^*\bm{R}_1) \times_2 \bm{U}_2, \\
			&\cm{H}_{\epsilon}^{(3)} = (\cm{G}-\cm{G}^*\times_1\bm{R}_1^\prime\times_2\bm{R}_2^\prime) \times_1 \bm{U}_1^*\bm{R}_1 \times_2 (\bm{U}_2-\bm{U}_2^*\bm{R}_2).
		\end{split}
	\end{align}
	It can be easily verified that
	\[
	\|\cm{H}_{\epsilon}^{(1)}\|_{\Fr} \leq (1+c_e)\sigma_U^{1/2} B^2, \hspace{5mm}\|\cm{H}_{\epsilon}^{(2)}\|_{\Fr} \leq (1+c_e)\sigma_U^{1/4} B^2\hspace{5mm}\text{and}\hspace{5mm}\|\cm{H}_{\epsilon}^{(1)}\|_{\Fr} \leq \sigma_U^{1/4} B^2.
	\]
	Hence, the proof of this lemma is accomplished.
\end{proof}

%%%%%%%%%%%%%%%%%%%%%%%%%%%%%%%%%%%%%%%%%%%%%%%%%%%%%%%%%%%%%%%%%%%%%%%%%
\begin{lemma}[Contractive projection property (CPP)]\label{lemma:HTineq}
	Consider a tensor $\cm{X}\in\mathbb{R}^{N\times N\times T_0}$ with frontal slices $\{\bm{X}_i, 1\leq i\leq T_0\}$, and let $S_1=\{1\leq j\leq T_0, \|\bm{X}_j\|_{\Fr}>0\}$ be the collection of nonzero frontal slices. Moreover, $\cm{X}^*\in\mathbb{R}^{N\times N\times T_0}$ is another tensor with $S_2$ being the collection of nonzero frontal slices.
	Denote by $s_j$ the cardinality of $S_j$ with $j=1$ and 2.
	If $s_2 < s \leq s_1$ and $S_2\subset S_1$, then
	\[
	\|\hardt{\cm{X},s} - \cm{X}\|_{\Fr}^2 \leq \frac{s_1 - s}{s_1 - s_2}\|\cm{X} - \cm{X}^*\|_{\Fr}^2.
	\]
\end{lemma}
\begin{proof}
	This lemma is a trivial extension of Lemma 1.1 in \cite{jain2014on} to the case with tensors, and the proof is also similar.
\end{proof}

%%%%%%%%%%%%%%%%%%%%%%%%%%%%%%%%%%%%%%%%%%%%%%%%%%%%%%%%%%%%%
\begin{lemma} \label{lemma:RE->RGC}
	Let $f$ be a continuously differentiable function and, for any tensors $\cm{A}$ and $\cm{B}$, 
	\[
	\frac{m}{2}\|\cm{A} - \cm{B}\|_{\Fr}^2 \leq f(\cm{A}) - f(\cm{B}) - \langle \nabla f(\cm{B}), \cm{A} - \cm{B}\rangle \leq \frac{M}{2}\|\cm{A} - \cm{B}\|_{\Fr}^2,
	\]
	where $0<m\leq M<\infty$.
	It then holds that
	\[
	\langle\nabla f(\cm{A}) - \nabla f(\cm{B}), \cm{A} - \cm{B} \rangle \geq \frac{mM}{m+M} \|\cm{A} - \cm{B}\|_{\Fr}^2 + \frac{1}{m+M} \|\nabla f(\cm{A}) - \nabla f(\cm{B})\|_{\Fr}^2.
	\]
\end{lemma}
\begin{proof}
	This lemma is from Theorem 2.1.11 in \cite{nesterov2003introductory} and is provided here to make the proof self-contained.
\end{proof}

\subsection{Additional information on the two empirical datasets}

We provide more detailed descriptions of the variables and their transformations at the two datasets in Section \ref{sec:empirical}. Specifically, Table  \ref{tab:macro-description} is for quarterly macroeconomic variables, and Table \ref{tab:rv-description} is for daily realized volatility.

\newpage
\clearpage

\renewcommand{\arraystretch}{1.1}
\begin{landscape}
	\begin{table}[t]
		\caption{Twenty quarterly macroeconomic variables. FRED MNEMONIC: mnemonic for data in FRED-QD. SW MNEMONIC: mnemonic in \cite{stock2012disentangling}. T: data transformation, where 1 = no transformation, 2 = first difference, and 3 = first difference of log series.  DESCRIPTION: brief definition of the data. G: Group code, where 1 = interest rate, 2 = money and credit, 3 = exchange rate, and 4 = stock market. }
		\label{tab:macro-description}
		\centering
		\resizebox{\columnwidth}{!}{
			\begin{tabular}{@{}llclc@{}}
				\hline
				FRED MNEMONIC  & SW MNEMONIC       & T & DESCRIPTION                                                                                                        & G \\
				\midrule
				FEDFUNDS       & FedFunds          & 2     & Effective Federal Funds Rate (Percent)                                                                             & 1     \\
				TB3MS          & TB-3Mth           & 2     & 3-Month Treasury Bill: Secondary Market Rate (Percent)                                                             & 1     \\
				BAA10YM        & BAA\_GS10         & 1     & Moody's Seasoned Baa Corporate Bond Yield Relative to Yield on 10-Year   Treasury Constant Maturity (Percent)      & 1     \\
				TB6M3Mx        & tb6m\_tb3m        & 1     & 6-Month Treasury Bill Minus 3-Month Treasury Bill, secondary market   (Percent)                                    & 1     \\
				GS1TB3Mx       & GS1\_tb3m         & 1     & 1-Year Treasury Constant Maturity Minus 3-Month Treasury Bill, secondary   market (Percent)                        & 1     \\
				GS10TB3Mx      & GS10\_tb3m        & 1     & 10-Year Treasury Constant Maturity Minus 3-Month Treasury Bill, secondary   market (Percent)                       & 1     \\
				CPF3MTB3Mx     & CP\_Tbill Spread  & 1     & 3-Month Commercial Paper Minus 3-Month Treasury Bill, secondary market   (Percent)                                 & 1     \\
				BUSLOANSx      & Real C\&Lloand    & 3     & Real Commercial and Industrial Loans, All Commercial Banks (Billions of   2009 U.S. Dollars), deflated by Core PCE & 2     \\
				CONSUMERx      & Real ConsLoans    & 3     & Real Consumer Loans at All Commercial Banks (Billions of 2009 U.S.   Dollars), deflated by Core PCE                & 2     \\
				NONREVSLx      & Real NonRevCredit & 3     & Total Real Nonrevolving Credit Owned and Securitized, Outstanding   (Billions of Dollars), deflated by Core PCE    & 2     \\
				REALLNx        & Real LoansRealEst & 3     & Real Real Estate Loans, All Commercial Banks (Billions of 2009 U.S.   Dollars), deflated by Core PCE               & 2     \\
				EXSZUSx        & Ex rate:Switz     & 3     & Switzerland / U.S. Foreign Exchange Rate                                                                           & 3     \\
				EXJPUSx        & Ex rate:Japan     & 3     & Japan / U.S. Foreign Exchange Rate                                                                                 & 3     \\
				EXUSUKx        & Ex rate:UK        & 3     & U.S. / U.K. Foreign Exchange Rate                                                                                  & 3     \\
				EXCAUSx        & EX rate:Canada    & 3     & Canada / U.S. Foreign Exchange Rate                                                                                & 3     \\
				NIKKEI225      &                   & 3     & Nikkei Stock Average                                                                                               & 4     \\
				S\&P 500       &                   & 3     & S\&P's Common Stock Price Index: Composite                                                                         & 4     \\
				S\&P: indust   &                   & 3     & S\&P's Common Stock Price Index: Industrials                                                                       & 4     \\
				S\&P div yield &                   & 2     & S\&P's Composite Common Stock: Dividend Yield                                                                      & 4     \\
				S\&P PE ratio  &                   & 3     & S\&P's Composite Common Stock: Price-Earnings Ratio                                                                & 4    \\
				\bottomrule
		\end{tabular}}
	\end{table}
\end{landscape}

\begin{table}[t]
	\caption{Forty six selected S\&P 500 stocks. CODE: stock code in the New York Stock Exchange. NAME: name of  company. G: group code, where 1 = communication service, 2 = information technology, 3 = consumer, 4 = financials, 5 = healthcare, 6 = materials and industrials, and 7 = energy and utilities.}
	\label{tab:rv-description}
	\centering
	%\resizebox{.7\columnwidth}{!}{
		\begin{tabular}{@{}llllll@{}}
			\hline
			CODE & NAME                                 & G & CODE & NAME                              & G \\
			\midrule
			T    & AT\&T Inc.                           & 1 & JPM  & JPMorgan Chase \& Co.             & 4 \\
			NWSA & News   Corp                          & 1 & WFC  & Wells Fargo \& Company            & 4 \\
			FTR  & Frontier   Communications Parent Inc & 1 & MS   & Morgan Stanley                    & 4 \\
			VZ   & Verizon Communications Inc.          & 1 & AIG  & American International Group Inc. & 4 \\
			IPG  & Interpublic   Group of Companies Inc & 1 & MET  & MetLife Inc.                      & 4 \\
			MSFT & Microsoft Corporation                & 2 & RF   & Regions   Financial Corp          & 4 \\
			HPQ  & HP   Inc                             & 2 & PGR  & Progressive Corporation           & 4 \\
			INTC & Intel Corporation                    & 2 & SCHW & Charles Schwab Corporation        & 4 \\
			EMC  & EMC   Instytut Medyczny SA           & 2 & FITB & Fifth   Third Bancorp             & 4 \\
			ORCL & Oracle Corporation                   & 2 & PFE  & Pfizer Inc.                       & 5 \\
			MU   & Micron Technology Inc.               & 2 & ABT  & Abbott Laboratories               & 5 \\
			AMD  & Advanced Micro Devices Inc.          & 2 & MRK  & Merck \& Co. Inc.                 & 5 \\
			AAPL & Apple Inc.                           & 2 & RAD  & Rite   Aid Corporation            & 5 \\
			YHOO & Yahoo! Inc.                          & 2 & JNJ  & Johnson \& Johnson                & 5 \\
			QCOM & Qualcomm Inc                         & 2 & AA   & Alcoa   Corp                      & 6 \\
			GLW  & Corning   Incorporated               & 2 & FCX  & Freeport-McMoRan Inc.             & 6 \\
			AMAT & Applied Materials Inc.               & 2 & X    & United   States Steel Corporation & 6 \\
			F    & Ford Motor Company                   & 3 & GE   & General Electric Company          & 6 \\
			LVS  & Las   Vegas Sands Corp.              & 3 & CSX  & CSX Corporation                   & 6 \\
			EBAY & eBay Inc.                            & 3 & ANR  & Alpha Natural Resources           & 7 \\
			KO   & Coca-Cola Company                    & 3 & XOM  & Exxon Mobil Corporation           & 7 \\
			BAC  & Bank of America Corp                 & 4 & CHK  & Chesapeake Energy                 & 7 \\
			C    & Citigroup Inc.                       & 4 & EXC  & Exelon Corporation                & 7\\
			\bottomrule
		\end{tabular}
		%}
\end{table}

\clearpage
\newpage

\end{document}